\documentclass[twocolumn,amsmath,amssymb,10pt,a4paper]{revtex4}

\pdfoutput=1

\usepackage{hyperref}

\hypersetup{
        colorlinks=true,
        linktoc=all,
        linkcolor=violet,
        urlcolor=blue
}

\usepackage{geometry}
\usepackage[utf8]{inputenc}

\usepackage{amsmath, amssymb}
\usepackage{graphicx}
\usepackage{bbold}
\usepackage{amsthm}
\usepackage{mathtools}
\usepackage{MnSymbol}
\usepackage{multirow}
\usepackage{marvosym}
\usepackage[normalem]{ulem}

\theoremstyle{definition}

\theoremstyle{remark}

\newtheorem{theorem}{Theorem}
\newtheorem{corollary}{Corollary}

\newcommand{\hobs}{\ensuremath{{\Gamma}}}

\newcommand{\be}{\begin{equation}}
\newcommand{\ee}{\end{equation}}
\newcommand{\bd}{\begin{displaymath}}
\newcommand{\ed}{\end{displaymath}}
\newcommand{\BE}{\begin{eqnarray}}
\newcommand{\EE}{\end{eqnarray}}

\newcommand{\bra}{\left\langle}
\newcommand{\ket}{\right\rangle}

\newcommand{\id}{{\rm 1\!\!I}}

\newcommand{\bs}{\ensuremath{\mathbf{s}}}
\newcommand{\bu}{\ensuremath{\mathbf{u}}}
\newcommand{\bv}{\ensuremath{\mathbf{v}}}

\newcommand{\bx}{\ensuremath{\mathbf{x}}}

\newcommand{\mcD}{\mathcal{D}}

\newcommand{\mcH}{\mathcal{H}}

\newcommand{\mcP}{\mathcal{P}}

\newcommand{\obs}[2][R]{\left.#2\right|_{#1}}
\newcommand{\I}[1]{\text{\sout{$I$}}_{#1}}
\newcommand{\oo}{observer-as-object}
\newcommand{\os}{observer-as-subject}

\usepackage{epigraph,varwidth}
\usepackage{etoolbox}

\usepackage{tcolorbox}
\newtcolorbox{mybox}{colback=green!5!white,colframe=green!75!black}

\usepackage[toc,page]{appendix}

\newcommand{\nocontentsline}[3]{}
\newcommand{\tocless}[2]{\bgroup\let\addcontentsline=\nocontentsline#1{#2}\egroup}

\begin{document}

\title{How can scientists establish an observer-independent science? \\ Embodied cognition, consciousness and quantum mechanics}

\author{John Realpe}\email{john.realpe@gmail.com}
\affiliation{Laboratory for Research in Complex Systems, San Francisco, USA}


\date{\today}

\begin{abstract}
Evidence is growing for the theory of embodied cognition, which posits that action and perception co-determine each other, forming an action-perception loop. This suggests that we humans somehow participate in what we perceive. So, how can scientists escape the action-perception loop to obtain an observer-independent description of the world? Here we present a set of conjectures informed by the philosophy of mind and a reverse-engineering of science and quantum physics to explore this question. We argue that embodiment, as traditionally understood, can manifest aspects of imaginary-time quantum dynamics. We then explore what additional constraints are required to obtain aspects of genuine, real-time quantum dynamics. In particular, we conjecture that an embodied scientist doing experiments must be described from the perspective of another scientist, which is ignored in traditional approaches to embodied cognition, and that observers play complementary roles as both objects experienced by other observers and ``subjects'' that experience other objects. 
\end{abstract}

\maketitle

\section{Introduction}\label{sm:intro}
\begin{figure*}
\includegraphics[width = 0.90\textwidth]{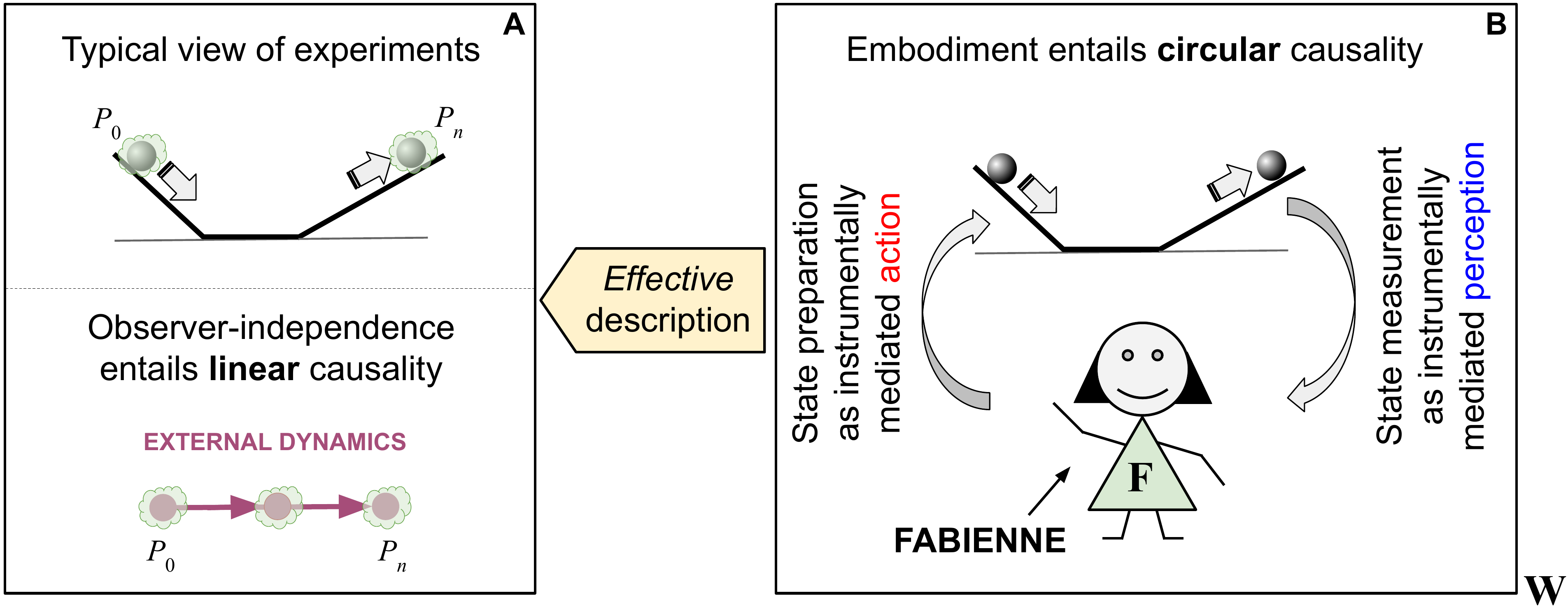}
\caption{{\em Effective seemingly observer-independent description of scientists doing experiments:} 
 (A) Experiments, and the world more generally, are typically viewed as observer-independent and so as instances of linear causality. (B) The instrumentally-mediated action-perception loop associated to an embodied scientist, Fabbiene (denoted by $F$), doing an experiment is usually considered an instance of circular causality. The circular causality of the instrumentally-mediated action-perception loop can indeed be effectively described in terms of a kind of linear causality (A)---the influence of the observer (green clouds) is effectively captured by describing the state of the system by a probability {\em matrix} that follows an imaginary-time quantum dynamics. However, here Fabbiene is described from the perspective of an external observer, Wigner (denoted by $W$---bottom right). Taking into account such an external observer might be a key aspect to obtain a genuine, real-time quantum dynamics (see Sec.~\ref{sm:first}). 
}\label{f:effective}
\end{figure*}

A central theme in modern cognitive science is the idea that action and perception are circularly related and fundamentally inseparable~\cite{varela2017embodied,thompson2010mind,di2017sensorimotor,shapiro2019embodied,djebbara2019sensorimotor,wilson2002six,bridgeman2011embodied}. That is, action and perception co-determine each other, forming an {\em action-perception loop}---i.e., ``a cycle in which perception leads to particular actions, which in turn create new perceptions, which then lead to new actions, and so on''~\cite{shapiro2019embodied}. Djebbara {\em et al.}~\cite{djebbara2019sensorimotor} recently reported experimental evidence for the existence of the action-perception loop. The idea that action and perception are somehow co-dependent has roots in a variety of fields~\cite{wilson2002six}, including Piaget's developmental psychology, which holds that cognitive abilities somehow emerge from sensorimotor skills; Gibson's ecological psychology, which sees perception in terms of potential interactions with the environment; and Merleau-Ponty's phenomenology, which views perception not as something that happens inside an organism that {\em passively} receives information about the world, but as a process wherein the organism {\em actively} seeks out information and interprets it in terms of the bodily actions it enables.

The action-perception loop has been particularly emphasized in the research program of embodied cognition~\cite{varela2017embodied,thompson2010mind,di2017sensorimotor,shapiro2019embodied,djebbara2019sensorimotor,wilson2002six,bridgeman2011embodied}, which arose as a reaction against the view that ``the mind and the world could be treated as separate and independent of each other, with the outside world mirrored by a representational model inside the head''~\cite{thompson2010mind}. In the traditional view, cognition begins with an input to the brain and ends with an output from the brain; so, traditional cognitive science can limit its investigations to processes within the head, without regard for the world outside the organism~\cite{shapiro2019embodied}. 

In contrast, embodied cognition posits that ``cognitive processes emerge from the nonlinear and circular causality of continuous sensorimotor interactions involving the brain, body, and environment''~\cite{thompson2010mind} In this view, perception does not result from {\em passively} sensing the physical world but from {\em actively} engaging with it in an ongoing reciprocal interaction between brain, body and world---such interactions could in principle be mediated by technologies that enhance motor and sensory capabilities. Again, we are involved in an action-perception loop, wherein we act to perceive and vice versa. 

According to Varela {\em et al.}~\cite{varela2017embodied}, the overall concern of embodied cognition is ``not to determine how some perceiver-independent world is to be recovered; it is, rather, to determine the common principles or lawful linkages between sensory and motor systems that explain how action can be perceptually guided in a perceiver-dependent world.'' Along the same lines, more recently di Paolo {\em et al.}~\cite{di2017sensorimotor} say (comments within brackets are our own):

\begin{quote}
``Action in the world is always perceptually guided. And perception is always an active engagement with the world. The situated perceiver does not aim at extracting properties of the world as if these were pregiven, but at understanding the engagement of her body [possibly enhanced by technological devices] with her surroundings, usually in an attempt to bring about a desired change in relation between the two. To understand perception is to understand how these sensorimotor regularities or contingencies are generated by the coupling of body and world [possibly mediated by technologies that can enhance motor and sensory capabilities] and how they are used in the constitution of perceptual and perceptually guided acts.''
\end{quote}

The action-perception loop is also emphasized in the theory of active inference, wherein an agent has a generative model of the external world and its motor systems suppress prediction errors through a dynamic interchange of prediction and action. In other words, ``there are two ways to minimize prediction errors: to adjust predictions to fit the current sensory input and to adapt the unfolding of movement to make predictions come true. This is a unifying perspective on perception and action suggesting that action is both perceived by and caused by perception''~\cite{djebbara2019sensorimotor}. 
 
According to Friston~\cite{friston2013life}, in active inference there is a circular causality analogous to the action-perception loop. Such circular causality means that ``external states cause changes in internal states, via sensory states, while the internal states couple back to the external states through active states---such that internal and external states cause each other in a reciprocal fashion. This circular causality may be a fundamental and ubiquitous causal architecture for self-organization.''

While the research program of embodied cognition and related fields encompass a broad spectrum of views, among which there is still ongoing debate, we here focus only on the action-perception loop, which appears to be a rather uncontroversial feature. Moreover, as already mentioned, Djebbara {\em et al.} recently reported experimental evidence for the existence of the action-perception loop. Depending on the context and on the interest of the authors, the action-perception loop tends to be modeled with different tools and with different degrees of complexity. For instance, the enactive view of embodied cognition tends to emphasize dynamical systems, while active inference tends to emphasize variational Bayesian methods. 
Here we use tools from statistical physics to model the action-perception loop in a rather parsimonious way, focusing exclusively on its main feature: the circular causality between action and perception.


Now, in cognitive science it is routine to model human beings interacting with external systems. Here we investigate the particular case where the human beings are scientists and the external systems are experimental systems. That is, we model scientists performing scientific experiments. This reflexive application of science to itself brings up an interesting question. Indeed, the action-perception loop entails that humans play an active and constructive role in the information they perceive about the world. In contrast, scientists apparently manage to obtain a completely observer-independent view of the world, {\em passively} mirroring an external reality without influencing it in any way (see Fig.~\ref{f:effective}A). How do scientists achieve such a feat? Of course, technology enhances scientists' capacities for perception and action, enabling them to transcend the limitations of their senses and to implement sophisticated interventions, e.g., at the sub-atomic level. However, while it is clear that technology can enable an enhanced, instrumentally-mediated action-perception loop, it is {\em not} at all clear that it can also change its circular topology. In other words, it is not clear that technology can break such an enhanced loop of instrumentally-mediated action and instrumentally-mediated perception (see Fig.~\ref{f:effective}B). 

In brief, our approach allows us to ask: How can scientists establish an observer-independent science, even though this seems to defy the very notion of embodied cognition? In other words, how can scientists escape the action-perception loop? Of course, logical reasoning is another powerful tool that allows scientists to transcend their limitations. However, logical reasoning should be able to acknowledge the existence of the action-perception loop, if it exists, and tell us how is it that we escape it. From a different perspective, our approach could also be considered as a self-consistency check to materialism: instead of {\em a priori} neglecting the physics or embodiment of scientists, as if they were immaterial, we let a scientific analysis tells us {\em a posteriori} how is it that we can do so. 

In principle, scientists differ from generic human beings in that they strive to achieve objectivity, which is often equated with observer-independence. Of course, we cannot start from the assumption of an observer-independent science since how this is established is precisely what we want to explore. Instead, we will use three conditions that, according to Velmans, characterize what in practice we may call a {\em reliable} science. These are~\cite{velmans2009understanding} (p. 219; see also Refs.~\cite{varela2017embodied,thompson2014waking,bitbol2008consciousness}): 

\begin{quote} 
{\bf R1. Standardization:} The procedures we used to investigate the world are standardized and explicit, so we clearly know what we are talking about. 

{\bf R2. Intersubjectivity:} The observations we do are intersubjective and repeatable, so we can mutually agree about the actual scientific facts.

{\bf R3. Truthfulness:} Observers are dispassionate, accurate and truthful, for obvious reasons. 
\end{quote}
Again, we are not {\em a priori} equating the notion of reliability with that of objectivity in the sense of observer-independence. However this does not deny {\em a priori} either that an observer-independent science can be established. 

\

\section{Outline}\label{sm:brief}

This work is to be read as a set of conjectures informed by the philosophy of mind and a reverse-engineering of science and quantum physics---see Sec.~5 in Ref.~\cite{realpe2} for a brief conceptual presentation of the main ideas involved. It is outlined as follows. In Sec.~\ref{sm:third} we discuss the circular dynamics of an embodied scientist interacting with an experimental system, which is similar to that of an action-perception loop~\cite{friston2010free,djebbara2019sensorimotor,di2017sensorimotor}. We show that this manifests aspects of ``imaginary-time'' quantum dynamics, which is described by a von Neumann equation without imaginary unit. While this is real-valued, genuine or ``real-time'' quantum dynamics is complex-valued. 

A natural question is: what conditions would be required to obtain a real-time quantum dynamics? We address this in Sec.~\ref{sm:first}. To explore this question we write the complex-valued von Neumann equation as a pair of real-valued equations by separating its real and imaginary parts. These equations look very similar to the imaginary-time von Neumman equation and its conjugate, except that a term is swapped, effectively coupling the two otherwise independent imaginary-time dynamics. We show that a similar swapping appears when dealing with reflexive systems, such as a pair of mirrors reflecting each other or a pair of video-camera systems pointing at each other. 

This suggests that we might obtain real-time quantum dynamics by reflexively coupling two (sets of) observers mutually observing each other. In this case, observers are relative to each other rather than to an external, unacknowledged observer. We explore this in the rest of Sec.~\ref{sm:first}. Based on an analogy with reflexive systems, we introduce some conjectures characterizing an observer, and we show that these lead to a dynamics with aspects of a genuine, real-time quantum dynamics---we refer to these conjectures collectively as the {\em reflexive coupling hypothesis}. 

However, Sec.~\ref{sm:first} should not be considered as a rigorous derivation of real-time quantum dynamics from reflexivity. The reason is that, reflexivity being a rather subtle and scarcely studied subject, the connection between the ideas of reflexivity and the conjectures we introduce may not be completely transparent. We introduce this section in this work because we find it conceptually plausible and we hope it can suggest future research on potential connections between the philosophy of mind and quantum physics. The literature on reflexive systems is rather scarce and we hope that an interdisciplinary approach to this topic could help further clarify or improve the conjectures we introduce here.

In sum, somewhat analogous to relational quantum mechanics (RQM)~\cite{Rovelli-1996}, in Sec.~\ref{sm:first} we assume that a classical embodied scientist interacting with a classical experimental system must be described from the perspective of another scientist, which is ignored in traditional approaches to embodied cognition ($W$ in Fig.~\ref{f:effective}B). However, to be consistent we should also take into account who observes this new scientist. We conjecture that we can escape the infinite regress that a na\"ive approach would entail in two steps. 

First, we assume that observers play complementary roles as both objects experienced by other observers and ``subjects'' that experience other objects. Here the word ``subject'' is used in a strict technical sense as the opposite of object. In this approach, the physics of objects can in principle be modelled in the traditional way---e.g., as particles following a cause-and-effect mechanism. In contrast, the physics of ``subjects'', being the opposite of objects, cannot be modelled in the same way, but only as random fluctuations irreducible to lower-level mechanisms.

Second, like two mirrors reflecting each other as well as another object, we have to conjecture that two (sets of) observers mutually observe each other as well as the experimental system. In this way, with some further assumptions, we show that it is plausible to obtain two coupled imaginary-time quantum dynamics that can be written as the imaginary and real parts of a dynamics formally analogous to genuine, real-time quantum dynamics. 

Finally, in Sec.~\ref{sm:discussion} we summarize our work and place it in the landscape of the philosophy of mind. Further details are provided in the appendices.

\section{Embodiment and imaginary-time quantum dynamics}\label{sm:third}
\begin{figure*}
\includegraphics[width = 0.90\textwidth]{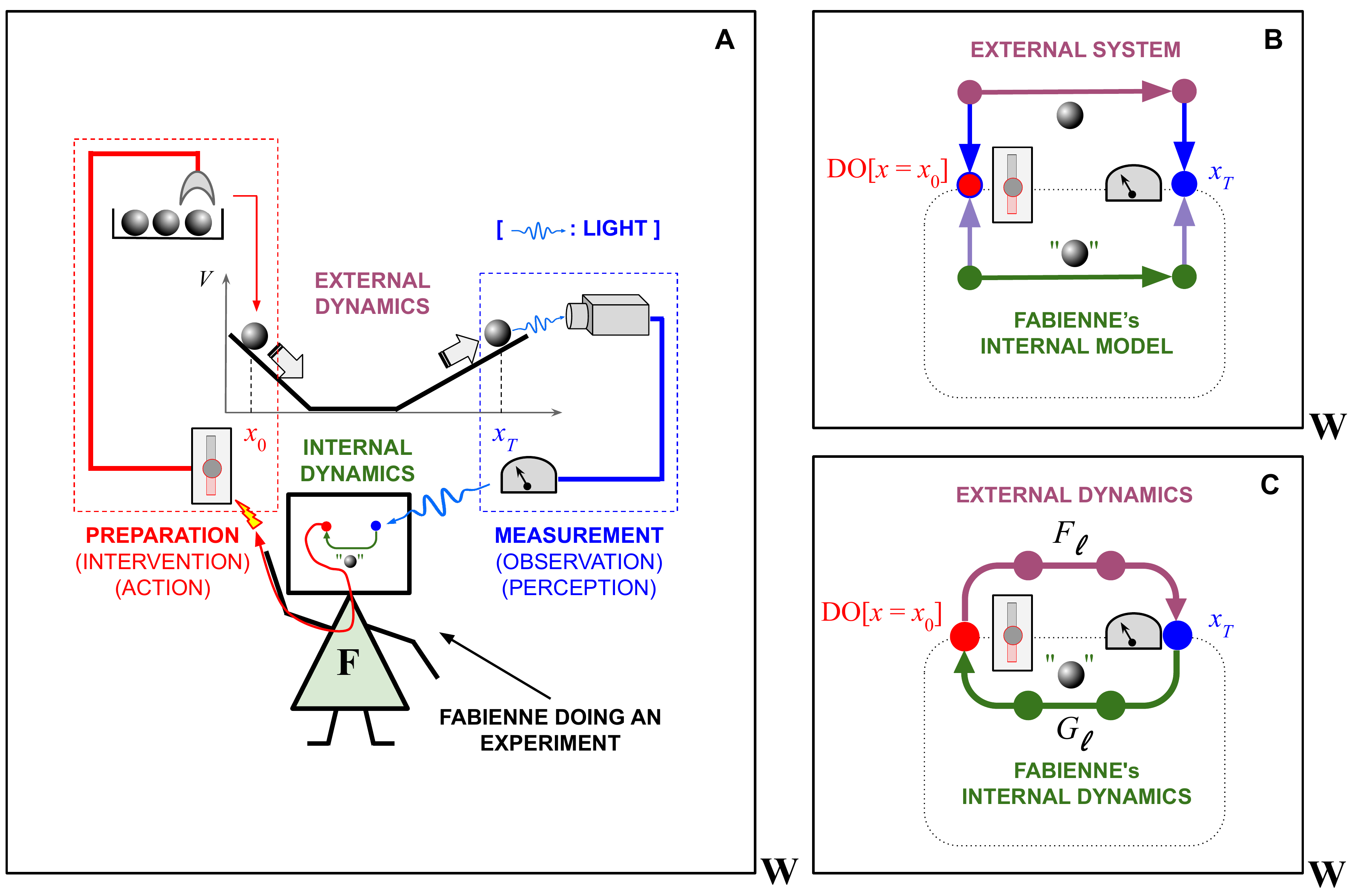}
\caption{{\em Experiments as circular processes:} 
(A) A scientist (Fabbiene) doing an experiment.
(B) Model of a scientist doing an experiment in the spirit of active inference. Fabbiene's actions can always prepare the same initial state, $x = x_0$, effectively implementing a causal intervention (this is denoted here as $\textsc{do}[x = x0]$) 
(C) Enactive model of a scientist and an experimental system as two physical systems involved in a circular interaction (see Appendix~\ref{s:enactive} and Fig.~\ref{f:enactive} therein). Arrows indicate the direction of the circular interaction, not conditional probabilities as in Bayesian networks. Factors $F_\ell$ and $G_\ell = F_{2n-1-\ell}$ describe, respectively, the dynamics external and and internal to Fabienne---here $\ell = 0,\dotsc , n$. In principle, instead of the single variable $x_0$ we should introduce two variables $x_0^{\rm e}$ describing the experimental system and $x_0^{\rm i}$ describing the corresponding Fabienne's physical correlate. However, since Fabienne intervenes the initial state of the system to be $x_0$, we have that $x_0^{\rm e} = x_0$ and $x_0^{\rm i} = x_0$, so we effectively have one variable $x_0$. Similar considerations apply to the final state $x_n$: there should be two variables $x_n^{\rm e}$ and $x_n^{\rm i}$. However, since Fabienne's measures the final state, which is $x_n$, we have that $x_n^{\rm e} = x_n$ and $x_n^{\rm i} = x_n$.
}\label{f:circular_noRQM}
\end{figure*}

\subsection{Embodied scientists doing experiments}

Here we build on enactivism whose task is ``to determine the common principles or lawful linkages between sensory and motor systems that explain how action can be perceptually guided in a perceiver-dependent world''~\cite{varela2017embodied} (p. 173) 
In Appendix~\ref{s:model_scientists} we provide a brief introduction to some aspects of embodied cognition. 

Importantly, we neglect the long and painful learning stage, when scientists are engaged in the invention and fine-tuning of new protocols, devices, and even concepts (e.g., spacetime curvature) that enables them to couple to the world in ways that were not possible before, and thus to enact new kinds of lawful regularities. For instance, the kind of regularities associated to quantum and relativity theories, which are invisible to the naked eye, are enabled by sophisticated experimental protocols and devices, as well as conceptual frameworks, all developed by scientists themselves.

Figure~\ref{f:circular_noRQM}A illustrates the dynamical coupling between an embodied scientist and an experimental system. This can be divided into four stages: (i) scientist's interventions on the experimental system, e.g. via moving some knobs, for preparing the desired initial state---this requires the physical interaction between the knobs and the observer's actuators; (ii) experimental system's dynamics---this is the main process traditionally analyzed in physics; (iii) scientist's measurement of the experimental system---this requires the physical interaction between the experimental system and the observer's sensors via the measuring device; (iv) scientist's internal dynamics which correlate with her experience of the experimental system. 

In the related approach of active inference~\cite{friston2010free,schwobel2018active}, experimental systems would be considered as generative {\em processes} which scientists can only access indirectly via the data generated in their sensorium (see Appendix~\ref{s:active}). Scientists can perturb such generative processes via their actions and have a generative {\em model} of their dynamics, including the effect of their own actions, which they can make as accurate as possible via learning. This is reflected in that, in Fig.~\ref{f:circular_noRQM}B, the topology of the Bayesian network representing the scientist mirrors the topology of the Bayesian network representing the experimental system. In particular, both internal and external dynamics flow in the same direction (horizontal arrows in Fig.~\ref{f:circular_noRQM}B; see Appendix~\ref{s:active} and Fig.~\ref{f:active} therein). 

Following enactivism~\cite{varela2017embodied, thompson2010mind,di2017sensorimotor}, instead, we give more relevance to the dynamical coupling between scientists and experimental systems. Learning scientific lawful regularities is not so much about extracting pre-existent properties of the world as about stabilizing this circular coupling and achieving ``reliability'' (conditions {\bf R1}-{\bf R3} above). This may include the development of new technologies, protocols and concepts. The lawful regularities achieved in the post-learning stage are our focus here. So, our approach is independent of a specific theory of learning (see Fig.~\ref{f:circular_noRQM}C; see also Appendix~\ref{s:enactive} and Fig.~\ref{f:enactive} therein).

\

\subsection{As simple as possible, but not simpler}

Of course, the scientific process generally involves many scientists and technologies. However, much as the theory of relativity can be developed without modeling all types of realistic clocks, our approach aims at capturing some general underlying principles valid beyond the particular model investigated. For instance, we could also have a situation where, say, a scientist in the UK prepares a laser pulse to send to another scientist in the Netherlands who would then perform a measurement and send the result back to the former via email. Only after both scientists have communicated can they reach any scientific conclusions about any potential correlations between the initial and final states of the laser pulse. This would again be a circular process. Instead of two scientists we could have many and the fundamental process would still be circular. For simplicity, we focus here on a single scientist. However, experiments generally comprise the four stages above. So, ours can be considered as a model of a generic process of ``reliable'' {\em observations}---though ignoring relativistic considerations. This process is embodied because all scientists and technologies involved are so. 

\

\subsection{Experiments as circular processes}

Here we setup the mathematical framework. Science is fundamentally concerned with causation, not with mere correlation. So, in general, a scientist do not passively observe the system to determine its initial state. Rather, she actively intervenes it to prepare a fixed initial state, runs the experiment and observes the final state. She repeats this enough times to determine the probability that, given that the initial state prepared is $x=x_0$, the final state observed is $x_n$. Using Pearl's do-calculus, this probability can be denoted as $\overline{\mcP}(x_n| \textsc{do}[x=x_0])$, where $\textsc{do}[x=x_0]$ refers to the scientist's intervention (notice the bar on $\overline{\mcP}$). This notation emphasizes that the scientist is not passively {\it observing} the initial state to be $x=x_0$, but rather actively intervening the system to make sure the initial state is always $x=x_0$. Ideally, the scientist would prepare every possible initial state to compute the full probability distribution for any initial state, $x_0$ prepared---in practice this might be impossible, though. In principle, she can select each intervention with a given probability. 

Unlike Pearl's do-calculus, we explicitly model the scientist doing the causal intervention. So, instead of using the {\sc do} operator, we can deal with such an intervention in a more direct manner, as we are about to see. As we mentioned earlier, we are considering only the post-learning stage, when the scientist is just repeating the experiment a statistically significant number of times. We model this as the stationary state, $\widetilde{\mcP}(\widetilde{\boldsymbol{x}})$, of a stochastic process on a cycle, which includes deterministic systems as a particular case (see Fig.~\ref{f:circular_noRQM}C; notice the tildes on $\widetilde{\mcP}$ and $\widetilde{\boldsymbol{x}}$)---this allows us to establish {\em a posteriori} which is the case. Here $\widetilde{\bx} =(x_0,\dotsc , x_{k -1})$ denotes a closed path $x_0\to x_1\to\cdots\to x_{k-1}\to x_0$ which returns to $x_{k} = x_0$ due to the scientist's causal interventions---as we said, experiments are not mere passive observations. 
This path could be divided into two open paths $x_0\to\cdots\to x_{n}$ and $x_n\to\cdots\to x_{k}$, with $x_k = x_0$, corresponding to the experimental system and the scientist, respectively. Furthermore, $\widetilde{\mcP}(\widetilde{\boldsymbol{x}})$ denotes the probability to observe a path $\widetilde{\bx}$. As we said above, the scientist can in principle select each intervention with a given probability, so $\widetilde{\mcP}(\widetilde{\boldsymbol{x}})$ can be non-zero for paths with different values of $x_0$---again, causal interventions are reflected in the fact that paths are closed. 

In principle, instead of the single variable $x_0$ we should introduce two variables $x_0^{\rm e}$ describing the experimental system and $x_0^{\rm i}$ describing the corresponding Fabienne's physical correlate. However, since Fabienne intervenes the initial state of the system to be $x_0$, we have that $x_0^{\rm e} = x_0$ and $x_0^{\rm i} = x_0$, so we effectively have one variable $x_0$. Similar considerations apply to the final state $x_n$: there should be two variables $x_n^{\rm e}$ and $x_n^{\rm i}$. However, since Fabienne's measures the final state, which is $x_n$, we have that $x_n^{\rm e} = x_n$ and $x_n^{\rm i} = x_n$.

Since energy plays a key role in physics, we assume that the stationary state is characterized by an ``energy'' function $\mcH_\ell (x_{\ell + 1}, x_\ell)$, where $0\leq \ell\leq k$ denotes the time step. For the case of a particle in a non-relativistic potential $V$ we have 
\be\label{em:non-relativisticH}
\mathcal{H}_\ell(x_{\ell +1} ,x_\ell) = \frac{m}{2}\left(\frac{ x_{\ell +1} - x_\ell}{\epsilon}\right)^2 + \frac{1}{2}\left[V(x_\ell)+V(x_{\ell+1})\right].
\ee
for the external path ($\ell = 0,\dotsc , n-1$)---in principle, the internal path ($\ell=n,\dotsc , k$) can have a different functional form (but see below). More precisely, a part of the external path, say $\ell = 0\dotsc , \ell_{\rm prep}<n$, could be considered as the preparation process and have a different function $\mcH_{\ell}$, but we obviate this for simplicity. Unlike the traditional Hamiltonian function, $\mcH_\ell$ is written in terms of consecutive position variables, $x_\ell$ and $x_{\ell +1}$, rather than instantaneous position and momentum. The potential $V$ in Eq.~\eqref{em:non-relativisticH} is symmetrized for convenience. We will discuss later on the case of more general, complex-valued, and so ``non-stoquastic'' Hamiltonians (see Sec.~\ref{sm:obs-only-rel} and Appendix~\ref{s:EM}).

We derive $\widetilde{\mathcal{P}}$ using the principle of maximum path entropy~\cite{presse2013principles}, a general variational principle analogous to the free energy principle from which a wide variety of well-known stochastic models at, near, and far from equilibrium has been derived~\cite{presse2013principles} (see Appendix~\ref{s:MaxCal}). To do so, we use the assumption, common in statistical physics, that we only know the average energy on the cycle $E_{\rm av}=\bra\frac{\epsilon}{T}\sum_\ell\mathcal{H}_\ell\ket_{\widetilde{\mathcal{P}}}$ (see below). Here $\epsilon\to 0$ is the time step size and $T = (k + 1) \epsilon $ is the total duration of a cycle. This is known~\cite{presse2013principles} to yield a Boltzmann distribution (see Appendix~\ref{s:MaxCal})
\be\label{em:wholeP}
\widetilde{\mcP}\propto \exp\{-\epsilon\sum_\ell\mcH_\ell/\hobs\},
\ee
where $\hobs = T/\lambda$ and $\lambda$ is a Lagrange multiplier fixing the average energy $E_{\rm av}$ on the cycle (see Appendix~\ref{s:MaxCal}). We will investigate later on the potential sources of fluctuations characterized here by the temperature- or diffusion-like parameter $\Gamma$.

So, how can scientists escape their embodiment and obtain an observer-independent description of the world? Obviously, we cannot just forcefully neglect the scientist at this point. The proper way to ignore the scientist in our approach is by marginalizing $\widetilde{\mcP}$ over the degrees of freedom associated to her. So, following the tradition in physics, we now focus on the external system and ignore the scientist by marginalizing $\widetilde{\mcP}$ over the internal paths, i.e., over $(x_{n+1},\dotsc , x_{k-1})$. This yields (see Appendix~\ref{s:MaxCal}; notice the absence of tildes in the left-hand side)
\begin{equation}\label{em:circular}
\begin{split}
\mathcal{P}(\bx) & = \sum_{x_{n+1},\dotsc , x_{k-1}}\widetilde{\mathcal{P}}(\widetilde{\bx})\\
&= \frac{1}{Z} \widetilde{F}_n(x_{0}^\prime,x_n)\cdots F_1(x_2, x_1) F_0(x_{1},x_0),
\end{split}
\end{equation}
where $Z$ is the normalization constant and we have written $x_0^\prime= x_0$ for future convenience---here we use sums to indicate either sums or integrals depending on the context. The expression $\bx =(x_0,\dotsc , x_n)$ denotes a path $x_0\to x_1\to\cdots\to x_n\to x_0$ which returns to $x_0$ due to the scientist's causal intervention, but where we disregard how it does so.
Furthermore, 
\be\label{em:Ftilde}
\widetilde{F}_n(x_{0}^\prime,x_n)  = \sum_{x_{n+1},\dotsc , x_{k-1}}F_{k-1}(x_0^\prime, x_{k-1})\cdots  F_n(x_{n+1},x_n),
\ee
summarizes the dynamics internal to the scientist, whose details we have disregarded, and
\be\label{em:F}
F_\ell(x^\prime ,x) = e^{-\epsilon \mathcal{H}_\ell(x^\prime ,x)/\hobs}/Z_\epsilon,
\ee
for $\ell=0,\dotsc , k$, where the constant $Z_\epsilon = \sqrt{2\pi\Gamma\epsilon/m}$ is introduced for convenience.  

\

\subsection{Circular causality and imaginary-time quantum dynamics}

We now describe the relationship between our model of embodied scientists doing experiments (see Fig.~\ref{f:effective}B) and the typical view of experiments, and the world more generally (see Fig.~\ref{f:effective}A). We typically think of experiments and the world in terms of linear causality. That is, as external systems that have an observer-independent initial state that evolves forward in time according to some observer-independent dynamical law (see Fig.~\ref{f:effective}A). In contrast, the action-perception loop associated to an embodied scientist doing experiments is usually considered as an instance of circular causality (see Fig.~\ref{f:effective}B). Here we show that such a circular causality can be effectively described in terms of a kind of linear causality. That is, we will show that the circular dynamics entailed by the presence of the embodied scientist can be {\em effectively} described {\em as if} it were an observer-independent dynamics. The price to pay, however, is that the state of the system has to be described in terms of a probability {\em matrix} that follows a dynamics formally analogous to imaginary-time quantum dynamics (see Fig.~\ref{f:effective}). 

\

\subsubsection{Linear causality and Markov chains}

First, notice that if we neglect the scientist, i.e., if we neglect the ``energy'' function associated to the internal paths, then $\widetilde{F}_n(x_0, x_n)$ becomes a constant. In this case the cycle in Fig.~\ref{f:circular_noRQM}C turns into a chain and we recover the most parsimonious non-trivial dynamical model where the probability distribution in Eq.~\eqref{em:circular} is Markov with respect to a chain on variables $x_\ell$~\cite{pearl2009causality} (p. 16; see Appendix~\ref{s:chainMarkov} herein)---a more parsimonious dynamical model would be memoryless. 

In particular, by knowing only the initial marginal $p_0$ and the forward transition probabilities $\mcP_{\ell}^+$ from time step $\ell$ to $\ell+1$, for all $\ell$, we can readily obtain the probability for a path (see Appendix~\ref{s:chainMarkov})
\be\label{em:Pchain}
\mcP_{\rm ch}(\bx) = p_0(x_0)\mcP^+_{ 0}(x_1|x_0)\cdots\mcP^+_{n-1}(x_n|x_{n-1}).
\ee
This implies in particular that we can obtain the marginal $p_{\ell +1}$ from the previous marginal $p_\ell$ via a Markovian update 
\be\label{em:Markov_chain}
p_{\ell+1}(x_{\ell+1}) = \sum_{x_\ell} \mathcal{P}_{\ell}^+(x_{\ell +1}|x_\ell) p_{\ell}(x_\ell).
\ee
That is, via a linear transformation specified by kernels $\mathcal{P}_{\ell}^+(x_{\ell +1}|x_\ell)$ satisfying the Chapman-Kolmogorov equation---i.e., where the transition probability from $\ell$ to $\ell + 2$, for instance, can be written as 
\be
\mcP_{\ell+2 | \ell}^+(x_{\ell+2}| x_\ell) = \sum_{x_{\ell+1}}\mcP_{\ell + 1}^+(x_{\ell+2}| x_{\ell+1})\mcP_{\ell}^+(x_{\ell+1}| x_\ell).
\ee 

This Markov chain describes the external system in terms of an observer-independent initial state $p_\ell$ that evolves forward in time according to an observer-independent dynamical law $\mcP_{\ell}^+$. In this sense, it could be considered as a paradigmatic example of {\em linear} causality. 

\

\subsubsection{Circular causality and imaginary-time quantum dynamics}

In general, we cannot neglect the observer and we cannot write the probability of a closed path in terms of a Markov chain due to the loopy correlations. This implies in particular that we cannot obtain the marginal $p_{\ell+1}$ from the previous one $p_\ell$ via a Markovian update as above. Indeed, since conditioning on two variables, $x_0$ and $x_n$, turns the cycle into a chain on the remaining variables, $x_1,\dotsc , x_{n-1}$, it is possible to show that Eq.~\eqref{em:circular} can be written as (see Appendix~\ref{s:cycleBernstein})
\be\label{em:P_Bernstein}
\mcP(\bx) = p(x_0, x_n)\prod_{\ell = 0}^{n-2}\mcP_\ell^+(x_{\ell+1}|x_n,x_\ell),
\ee
which yields a Bernstein process where initial {\em and } final states must be specified~\cite{Zambrini-1987} (here the two-variable marginal $p$ and the transition probability $\mcP_\ell^+$, respectively, plays the role of $m$ and $h$ in Eq.~(2.7) therein). 

However, we can recover an {\em effective} Markovian-{\em like} update on configuration space if, instead of marginals, we consider (real) probability matrices. Indeed, if we relax the condition $x_0^\prime = x_0$ in Eq.~\eqref{em:circular} and marginalize all other variables, then we obtain a probability matrix $P_0(x_0^\prime , x_0) = \sum_{x_1,\dotsc , x_n}\mathcal{P}(\bx)$ whose diagonal $P_0(x_0 ,x_0)= p_0(x_0)$ yields the actual probabilities. So, interpreting factors as matrix elements, Eq.~\eqref{em:circular} yields $P_0 = \widetilde{F}_n\cdots F_1 F_0/Z$. Similarly, for $\ell=1$ we get $P_1=F_0\widetilde{F}_n\cdots F_1/Z$ and $P_1(x_1 , x_1)=p_1(x_1)$. Here we have removed the prime from $x_0$ in Eq.~\eqref{em:circular}, added a prime to $x_1$ in $F_0$, moved $F_0(x_1^\prime, x_0)$ to the beginning of Eq.~\eqref{em:circular}, and done the marginalization over all other variables, $P_1(x_1^\prime , x_1)=\sum_{x_0 , x_2,\dotsc , x_n}\mathcal{P}(\bx)$. 

So, we can obtain the probability matrix $P_1=F_0\widetilde{F}_n\cdots F_1/Z$ from the previous one, $P_0 = \widetilde{F}_n\cdots F_1 F_0/Z$, via the cyclic permutation of matrix $F_0$. Iterating this process $\ell$ times yields 
\be\label{em:matrix}
P_\ell = \frac{1}{Z} F_{\ell-1}\cdots F_1 F_0\widetilde{F}_n \cdots F_{\ell+1}F_\ell ,
\ee
where $P_\ell(x,x) = p_\ell(x)$. If $F_\ell$ is invertible we can write (for simplicity, we are assuming the case of mixed states in Eq.~\eqref{em:matrix}, since pure states would be associated to non-invertible matrices $F_\ell$---however, we can make a mixed state as close as we want to a pure state)
\be\label{em:P_l+1}
P_{\ell +1} = F_{\ell} P_{\ell} F_{\ell}^{-1},
\ee
for $\ell=0,\dotsc , n-1$. This is an effective Markovian-like update in that it yields $P_{\ell +1}$ via a linear transformation of $P_\ell$ alone, where the kernels $F_\ell$ satisfy the analogue of Chapman-Kolmogorov equation---i.e., the factor between time steps $\ell$ and $\ell +2$, for instance, can be written as $F_{\ell + 2 | \ell}\equiv F_{\ell +1} F_{\ell}$. In this sense, the Markovian-like update above could be considered a paradigmatic example circular causality.

This shows that we can effectively sidestep the circular causality entailed by the embodied scientist. In other words, we can indeed describe experiments in the traditional way, i.e., in terms of an external causal chain that seems to be independent of the observer (see Fig.~\ref{f:effective}A). However, the price to pay is that the state of such an external system has to be described in terms of probability matrices instead of probability vectors. The off-diagonal elements of such matrices contain relevant dynamical information since, if we neglect them, we cannot build $P_{\ell+1}$ from $P_\ell$ and $F_\ell$ alone. Much as in quantum physics, the diagonal elements of such probability matrices yield the actual probabilities to observe the system in a particular state. We will now see that such probability matrices follow an imaginary-time quantum dynamics, i.e., they satisfy von Neumann equation in imaginary-time. 

Indeed, when $\epsilon\to 0$, we can assume that variables $x_\ell$ and $x_{\ell +1}$ are typically close to each other. In other words, we can assume that 
\be\label{em:F=1+J}
F_\ell = \id + \epsilon J_\ell + O(\epsilon^2),
\ee
where $\id$ is the identity.  For discrete variables, the {\em dynamical matrix} $J_\ell$ has non-negative off-diagonal elements.  For continuous variables $J_\ell$ is actually an operator. For instance, for $\mcH_\ell$ in Eq.~\eqref{em:non-relativisticH} we have $J_\ell \to - H/\hobs$, when $\epsilon\to 0$, where
\be\label{em:H}
H = -\frac{\hobs^2}{2 m}\frac{\partial^2}{\partial x^2} +V(x) ,
\ee
is equivalent to the quantum Hamiltonian of a non-relativistic particle in a potential $V$, and $\hobs$ plays the role of Planck's constant. 

We can see this by applying the corresponding factor $F_\ell$ to a generic and well-behaved test function $g$, i.e.,
\be\label{em:test}
[F_\ell g](x) = \int F_\ell(x, x^\prime) g(x^\prime)\mathrm{d} x^\prime.
\ee
Introducing Eq.~\eqref{em:non-relativisticH} into Eq.~\eqref{em:F}, we have
\be\label{em:F_Gaussian}
F_\ell (x,x^\prime) = \tfrac{1}{\sqrt{2\pi\sigma^2}}e^{-\frac{1}{2\sigma^2}(x-x^\prime)^2}e^{-\frac{\epsilon}{2\Gamma} [V(x)+V(x^\prime)]},
\ee
where $\sigma^2 = \epsilon\hobs/m$ is the variance of the Gaussian factor.  
When $\epsilon\to 0$, this Gaussian factor 
is exponentially small except in the region where ${|x - x^\prime| = O(\sqrt{\hobs\epsilon/m})} $. So, we can estimate the integral in Eq.~\eqref{em:test} to first order in $\epsilon$ by expanding $g(x^\prime)$ around $x$ up to second order in ${x - x^\prime}$ and performing the corresponding Gaussian integrals. Consistent with this approximation to first order in $\epsilon$, we can also do $\exp{ [-V(y) \epsilon /2 \hobs]} = 1 - V(x)\epsilon /2 \hobs + O(\epsilon^2) $, for $y$ equal to either $x$ or $x^\prime$. This finally yields 
\be\label{em:g-Hg}
[F_\ell g](x) = g(x) -\epsilon H g(x)/\hobs + O(\epsilon^2),
\ee
i.e.,  $F_\ell = \id - \epsilon H/\hobs+ O(\epsilon^2)$, where $H$ is given by Eq.~\eqref{em:H} (see Appendix~\ref{s:derivationF=1+J}). So, for factors given by Eq.~\eqref{em:F_Gaussian} we have

Either way, whether the variables are discrete or continuous, introducing Eq.~\eqref{em:F=1+J} into Eq.~\eqref{em:P_l+1} yields
\be\label{em:real_vN}
\Delta {P}_{\ell} = \epsilon [J_\ell, P_\ell] + O(\epsilon^2),
\ee
where $\Delta P_{\ell} = P_{\ell + 1} - P_\ell$ and  
\be
[A, B] =  AB - BA,
\ee
is the commutator between operators $A$ and $B$. To obtain Eq.~\eqref{em:real_vN} we have taken into account that 
\be\label{em:F-1}
F_\ell^{-1} = \id -\epsilon J_\ell + O(\epsilon^2),
\ee
when $F_\ell$ is invertible. Dividing by $\epsilon$ and taking the continuous-time limit ($\epsilon\to 0$), Eq.~\eqref{em:real_vN} yields $\partial P/\partial t = [J, P]$, or 
\be\label{em:dPdt=[H,P]}
-\hobs\frac{\partial P}{\partial t} = [H, P],
\ee 
where $t=\ell\epsilon$. This is  von Neumann equation in imaginary time with $\Gamma$ playing the role of Planck's constant. Indeed, the von Neumann equation is given by 
\be\label{em:drhodt=[H,rho]}
i\hbar\frac{\partial \rho}{\partial t} = [H, \rho],
\ee
where $\rho$ is the density matrix and $i$ is the imaginary unit. Multiplying and dividing the left hand side of Eq.~\eqref{em:drhodt=[H,rho]} by $i$ yields $-\hbar\partial\rho/\partial (it) = [H,\rho]$, which is equivalent to Eq.~\eqref{em:dPdt=[H,P]} if we replace $i t$ by $t$ and $\hbar$ by $\hobs$. 


\

\subsubsection{Imaginary-time Schr\"odinger equation as belief propagation}\label{sm:imaginary-time-bp}
Here we show that the cavity method of statistical mechanics~\cite{Mezard-book-2009} (ch. 14) naturally leads to the imaginary-time versions of the wave function, $\psi(x)$, the Born rule, $p(x) = \psi(x)\psi^\ast(x)$,  and Schr\"odinger's equation
\be\label{em:idpsi/dt=Hpsi}
i\hbar\frac{\partial \psi}{\partial t} = H\psi,
\ee
which is equivalent to Von Neumann's equation, Eq.~\eqref{em:drhodt=[H,rho]}, for pure states, i.e. for $\rho(x, x^\prime)=\psi(x)\psi^\ast(x^\prime)$.

While the cavity method is not exact on cycles, here we will see that it remains exact when the particle is initially localized at position $x_0=x_0^\ast$, for some $x_0^\ast$, which is an instance of a pure state. This can be described by a factor of the form $F_0(x_1,x_0) = \delta(x_0-x_0^\ast) \overline{F}(x_1,x_0)$, where the Dirac delta function enforces the constraint $x_0=x_0^\ast$ and $\overline{F}(x_1,x_0)$ collects the remaining contributions to $F_0(x_1,x_0)$. In this case, the probability for the particle to be at position $x_\ell$ at time step $\ell$ is given by (cf. Eq.~(14.4) in Ref~\cite{Mezard-book-2009}; the factor $e^{\beta B \sigma_j}$ therein can be absorbed in the $\hat{\nu}$ messages)
\be\label{em:pl-bp}
\begin{split}
p_{\ell}(x_\ell) =& \int{\mcP}({\bx})\prod_{\ell^\prime\neq\ell}\mathrm{d}x_{\ell^\prime}
= \mu_{\to\ell}^\ast(x_\ell)\mu_{\ell\leftarrow}^{\ast}(x_\ell),
\end{split}
\ee
where $\mcP(\bx)$ is given by Eq.~\eqref{em:circular} and the integral is performed over variables $x_{\ell^\prime}$ for $\ell^\prime = 0, \dotsc \ell -1, \ell+1,\dotsc ,n$. Here the functions $\mu_{\to\ell}^\ast$ and $\mu_{\ell\leftarrow}^\ast$ collect the contributions for $\ell = 1,\dotsc , \ell -1$ and $\ell+1,\dotsc , n$, respectively, after integrating over $x_0$---the superindex $\ast$ reminds us that we have the constraint $x_0 = x_0^\ast$. They are given by (cf. Eq.~(14.2) in Ref.~\cite{Mezard-book-2009}; the messages therein are normalized differently)
\begin{widetext}
\BE
\mu_{\to\ell}^{\ast}(x_\ell) =& \frac{1}{\sqrt{Z}} \int F_{\ell-1}(x_\ell, x_{\ell-1})\cdots F_1(x_2,x_1) \overline{F}_{0}(x_1,x_0^\ast)\prod_{\ell^\prime = 1}^{\ell-1}\mathrm{d}x_{\ell^\prime},\label{em:it-bp->}\\
\mu_{\ell\leftarrow}^{\ast}(x_\ell) =& \frac{1}{\sqrt{Z}}\int\widetilde{F}_{n}(x_0^\ast, x_n) F_{n-1}(x_n, x_{n-1})\cdots F_{\ell}(x_{\ell + 1}, x_{\ell})\prod_{\ell^\prime=\ell+1}^{n}\mathrm{d}x_{\ell^\prime}, \label{em:it-bp<-}
\EE
\end{widetext}
for $\ell = 1, \dotsc , n$, with 
\BE
\mu_{\to 1}^{\ast}(x_1) =& \tfrac{1}{\sqrt{Z}}\overline{F}_0(x_1,x_0^\ast),\label{em:bp-init->}\\ 
\mu_{n\leftarrow}^{\ast}(x_n) =& \tfrac{1}{\sqrt{Z}}\widetilde{F}_n(x_0^\ast, x_n).\label{em:bp-init<-}
\EE
Here we have separated the normalization constant $Z$ into two contributions $\sqrt{Z}$, such that when we multiply the functions $\mu_{\to\ell}^{\ast}(x)$ and $\mu_{\ell\leftarrow}^{\ast}(x)$ we recover the full value of $Z$.

Equations~\eqref{em:it-bp->} and \eqref{em:it-bp<-} can be written recursively as
\BE
\mu_{\to\ell}^{\ast}(x_\ell) =& \int F_{\ell-1}(x_\ell, x_{\ell-1})\mu_{\to{\ell-1}}^{\ast}(x_{\ell-1})\mathrm{d}{x_{\ell-1}},\label{em:it-bp-recursion->}\\
\mu_{\ell\leftarrow}^{\ast}(x_\ell) =& \int \mu_{\ell+1\leftarrow}^{\ast}(x_{\ell+1}) F_{\ell}(x_{\ell + 1}, x_{\ell})\mathrm{d}x_{\ell+1},\label{em:it-bp-recursion<-}
\EE
for $\ell=2,\dotsc , n-1$, with initial conditions given by Eqs.~\eqref{em:bp-init->} and \eqref{em:bp-init<-}---although we can leave some factors $F_\ell$ as preparation of a more general initial probabilistic state. This can be seen by using Eqs.~\eqref{em:it-bp->} and \eqref{em:it-bp<-} to replace $\mu_{\to\ell-1}^\ast(x)$ and $\mu_{\ell +1\leftarrow}^\ast(x)$, respectively, in Eqs.~\eqref{em:it-bp-recursion->} and \eqref{em:it-bp-recursion<-}.
Equations~\eqref{em:it-bp-recursion->} and \eqref{em:it-bp-recursion<-} are an instance of the belief propagation algorithm derived via the cavity method (cf. Eq.~(14.5) in Ref~\cite{Mezard-book-2009}). 

Expanding the factors $F_{\ell-1}$ and $F_\ell$, respectively, in the integrals of Eqs.~\eqref{em:it-bp-recursion->} and \eqref{em:it-bp-recursion<-}---as we did for the factor in Eq.~\eqref{em:F_Gaussian}---we can obtain the analogue of Schr\"odinger's equation in imaginary time and its conjugate. Let us see this for the case of factors given by Eq.~\eqref{em:F_Gaussian}, which are symmetric. In this case, since the factors are symmetric, Eq.~\eqref{em:it-bp-recursion<-} can also be written as $\int  F_{\ell}(x_{\ell},x_{\ell + 1})\mu_{\ell+1\leftarrow}^{\ast}(x_{\ell+1})\mathrm{d}x_{\ell+1}$. So, both Eqs.~\eqref{em:it-bp-recursion->} and \eqref{em:it-bp-recursion<-} are instances of Eq.~\eqref{em:test} and, according to Eq.~\eqref{em:g-Hg}, we can write
\BE
\mu_{\to\ell}^{\ast}(x_\ell) =& \mu_{\to{\ell-1}}^{\ast}(x_{\ell})-\frac{\epsilon}{\hobs} H\mu_{\to{\ell-1}}^{\ast}(x_{\ell})\label{em:it-bp-test->}\\
\mu_{\ell\leftarrow}^{\ast}(x_\ell) =& \mu_{\ell+1\leftarrow}^{\ast}(x_{\ell}) -\frac{\epsilon}{\hobs}H\mu_{\ell+1\leftarrow}^{\ast}(x_{\ell}).\label{em:it-bp-test<-}
\EE

To take the continuous-time limit, $\epsilon\to 0$ with $t=\ell\epsilon$, let us write $\mu_\to^{\ast}(x,t)=\mu_{\to\ell}^{\ast}(x)$ and $\mu_\leftarrow^{\ast}(x,t)=\mu_{\ell\leftarrow}^{\ast}(x)$. Since the second terms in the right hand side of Eqs.~\eqref{em:it-bp-test->} and \eqref{em:it-bp-test<-} are already of order $\epsilon$ we can replace $\mu_{\to\ell-1}^{\ast}$ and $\mu_{\ell+1\leftarrow}^{\ast}$, respectively, by $\mu_{\to\ell}^{\ast}$ and $\mu_{\ell\leftarrow}^{\ast}$. Moving the first terms in the right hand side of Eqs.~\eqref{em:it-bp-test->} and \eqref{em:it-bp-test<-} to the left hand side, multiplying both sides by $-\hobs/\epsilon$, and taking the limit $\epsilon\to 0$ we get
\BE
-\hobs\frac{\partial\mu_{\to}^{\ast}(x,t)}{\partial t} &=& H\mu_{\to}^{\ast}(x,t),\label{em:itSch}\\
\hobs\frac{\partial\mu_{\leftarrow}^{\ast}(x,t)}{\partial t}  &=& H\mu_{\leftarrow}^{\ast}(x,t). \label{em:-itSch}
\EE

Equations~\eqref{em:itSch} and \eqref{em:-itSch} are the imaginary-time Schr\"odinger's equation and its adjoint, respectively, where $\hobs$, $\mu_{\to}^{\ast}$ and $\mu_{\leftarrow}^{\ast}$ play the role of Planck's constant, the imaginary-time wave function and its conjugate, respectively (cf. Eq.~\eqref{em:idpsi/dt=Hpsi}). Indeed, Eqs.~\eqref{em:itSch} and \eqref{em:-itSch} are formally analogous to Eqs. (2.1) and (2.17) in Ref.~\cite{Zambrini-1987}, where imaginary-time quantum dynamics is extensively discussed---the analogous of $\theta$ and $\theta^\ast$ therein are here $\mu_{\leftarrow}^\ast$ and $\mu_{\to}^\ast$, respectively. The imaginary-time analogue of the Born rule is naturally given by the continuous-time limit of Eq.~\eqref{em:pl-bp}, which yields $p(x,t) = \mu_\to^{\ast}(x,t)\mu_\leftarrow^{\ast}(x,t)$, where $p(x,t) = p_\ell(x)$ with $t=\ell\epsilon$.

\

\subsection{Example: Imaginary-time quantum interference in a {\em classical} two-slit experiment}\label{sm:two_slit}

\begin{figure}
\includegraphics[width=\columnwidth]{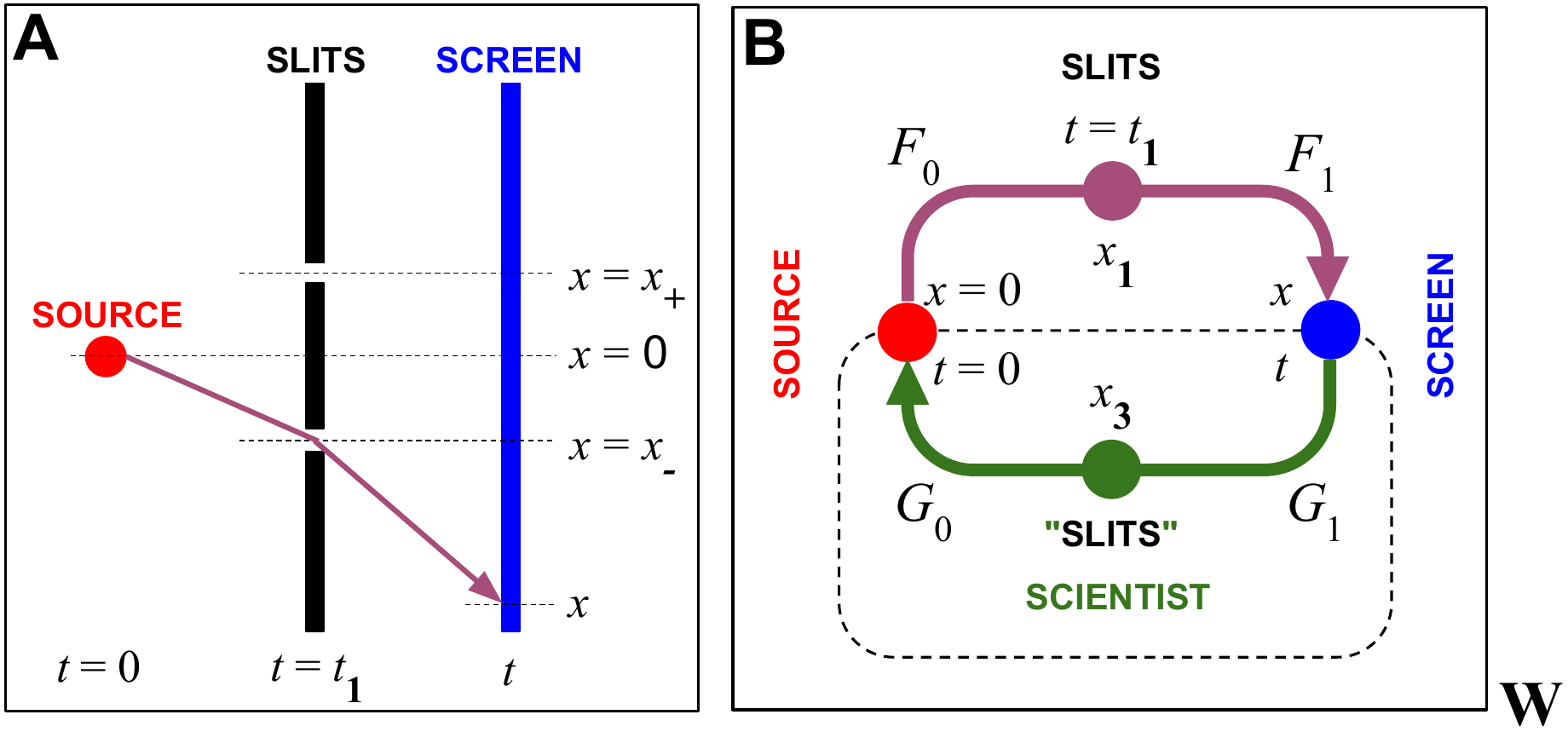}
\caption{{\em Two-slit experiment}: (A) A {\em classical} particle initially at (vertical) position $x=0$ at time $t=0$ goes at time $t = t_1$ through a barrier with two slits, located at positions $x_\pm = \pm x_{\rm slit}$, and hits a screen at a generic position $x$ at a generic time $t$. Slits can be open or closed. 
(B) Factor graph associated to a scientist performing a two-slit experiment. The real non-negative factors $F_0$ and $F_1$ capture the ``external'' dynamics between the source and the slits, and between the slits and the screen, respectively. Factors $G_0$ and $G_1$ captures the dynamics ``internal'' to the scientist. Here $x_3$ refers to the ``internal'' physical correlate of the position of the slits.  When only one slit is open, say the slit located at $x_+$, consistency requires that $x_1 = x_3 = x_+$. However, when both slits are open, $x_1$ does not have to equal $x_3$ anymore, which yields the imaginary-time version of quantum interference.  } \label{f:slits}
\end{figure}

Here we consider the specific instance of a scientist doing a {\em classical} two-slit experiment (see Fig.~\ref{f:slits}). We will see that it coincides with the imaginary-time version of the standard quantum two-slit experiment. 
Consider a path $\widetilde{\bx} = (x_0, x_1, x_2, x_3)$ that starts at the source located at $x_0=0$ at $t_0=0$ and goes through $x_1\in\{x_+ , x_-\}$, $x_2=x$, and $x_3$ at times $t_1$, $t_2=t$, and $t_3$, respectively, to return to $x_4 = x_0=0$. Here $x_1$ describes the slits and $x_3$ describes the physical correlates of the slits ``internal'' to the scientist in Fig.~\ref{f:slits}, so $x_3\in\{x_+, x_-\}$ too. With this notation, following Eq.~\eqref{em:circular}, the probability associated to the path $\widetilde{\bx}$ when both slits are open is
%
\be\label{e:Pfactor-slits}
\widetilde{\mcP}(\widetilde{\bx}) = f(\widetilde{\bx})\delta(x_0) b(x_1)b(x_3),
\ee
%
where 
\be\label{e:ffactors}
f(\widetilde{\bx})=G_0(x_0, x_3 )G_1(x_3 , x_2) F_1 (x_2, x_1) F_0 (x_1 , x_0) 
\ee
includes the factors describing the dynamics internal and external to the scientist. Furthermore, 
\be
b(x^\prime) = \delta(x^\prime - x_+) +\delta(x^\prime - x_-),
\ee
is included to enforce the constrains that $x_1,x_3\in\{x_+,x_-\}$. Additionally, the Dirac delta function $\delta (x_0)$ is included to enforce the constrain $x_0=0$.  In Eq.~\eqref{e:Pfactor-slits} the normalization constant has been absorbed in the factors $F_\ell$ and $G_\ell$, for $\ell = 0,1$.

If one of the two slits is closed, the particle can only go through one of them at time $t_1$. That is, $x_1 = x_\pm$ where the upper and lower sign denotes the situation where the particle goes through the upper and lower slit, respectively. In this case, the physical correlate of the slit internal to the scientist also has to equal $x_\pm$, i.e., $x_3=x_\pm$. To account for this, we have to incorporate two Dirac delta functions $\delta(x_1-x_\pm)$ and $\delta(x_3-x_\pm)$, respectively, instead of the terms $b(x_1)$ and $b(x_3)$ in Eq.~\eqref{e:Pfactor-slits}. Thus, the probability associated to the path $\widetilde{\bx}$ when only one slit is open is
\be
\widetilde{\mcP}_\pm(\widetilde{\bx})= \frac{1}{\mathcal{Z}_\pm}f(\widetilde{\bx})\delta(x_0)\delta(x_1-x_\pm)\delta(x_3-x_\pm), 
\ee
where $\mathcal{Z}_\pm$ is included to guarantee the normalization of $\widetilde{\mcP}_\pm(\widetilde{\bx})$.

According to this, the probability to find the particle at position $x_2 = x$ at time $t_2=t$ is given by
\be\label{e:im-both}
\begin{split}
p_\pm(x, t)\equiv \int\widetilde{\mcP}_\pm(\widetilde{\bx})\mathrm{d} x_0 \mathrm{d} x_1\mathrm{d} x_3= \tfrac{1}{Z_{\rm one}}f(0, x_\pm, x, x_\pm ), 
\end{split}
\ee
where 
\be
Z_{\rm one} = {Z_\pm}  \equiv \int f(0, x_\pm, x, x_\pm )\mathrm{d} x ,
\ee
ensures that $p_\pm$ is normalized. Due to the symmetry between slits, we have that $Z_\pm = Z_{\rm one} $ is the same for both $x_{\pm}$. 

Following Eqs.~\eqref{e:ffactors} and \eqref{e:im-both} we can write 
\BE
\tfrac{1}{\sqrt{Z_{\rm one}}}F_0 (x_\pm , 0) F_1(x, x_\pm)  &=& \sqrt{p_{\pm}(x, t)}e^{S_\pm(x,t)},\label{e:im-state}\\
\tfrac{1}{\sqrt{Z_{\rm one}}}G_1(x_\pm , x) G_0(0, x_\pm ) &=& \sqrt{p_{\pm}(x, t)}e^{-S_\pm(x,t)},\label{e:im-stateG}
\EE
which defines the imaginary-time phases $S_\pm$, i.e., 
\be\label{e:im-phase-def}
S_\pm(x,t)=\tfrac{1}{2}\log\frac{F_0 (x_\pm , 0) F_1(x, x_\pm)}{G_1(x_\pm , x) G_0(0, x_\pm )}.
\ee

If the two slits are open, instead, then the probability for the particle to be located at position $x_2 = x$ at time $t_2 = t$ is given by (see Eq.~\eqref{e:Pfactor-slits})
%
\be\label{e:im-interf}
\begin{split}
p_{\rm both}(x , t) &\equiv \int \widetilde{\mcP}(\widetilde{\bx})\mathrm{d}x_0\mathrm{d}x_1 \mathrm{d}x_3\\
&=\tfrac{1}{Z_{\rm both}}\sum_{x_1, x_3\in\{x_+ , x_-\}} f(0, x_1, x, x_3 ) \\ 
&= C \left\{\frac{1}{2}\left[p_+(x, t)  + p_-(x, t)\right] + \mathcal{I}\right\}\\
\end{split}
\ee
%
where $Z_{\rm both}$ is a normalization constant, $C = {2 Z_{\rm one}}/{Z_{\rm both}}$, and 
\be\label{em:Iim}
\mathcal{I} = \sqrt{p_+(x,t) p_-(x, t)}\cosh\left[\Delta S(x, t)\right] 
\ee
with
\be
\Delta S(x, t) = S_+(x,t) - S_-(x,t).
\ee
The first two terms in the second line of Eq.~\eqref{e:im-interf} come from the elements of the sum with $x_1 = x_3 = x_\pm$ (see Eq.~\eqref{e:im-both}).  The last term in the second line of Eq.~\eqref{e:im-interf}, which is the imaginary-time interference term, comes from the elements of the sum with $x_1\neq x_3$, i.e. (see Eq.~\eqref{e:ffactors}), 
%
\BE\label{e:im-interf-term}
f(0, x_+, x, x_- ) &=& Z_{\rm one}\sqrt{p_+(x, t)p_-(x,t)}e^{\Delta S} ,\\
f(0, x_-, x, x_+ )&=& Z_{\rm one}\sqrt{p_+(x, t)p_-(x,t)}e^{-\Delta S}.
\EE
%
The right hand side of these equations is obtained by using Eqs.~\eqref{e:im-state} and \eqref{e:im-stateG}. 

According to Eq.~\eqref{e:im-interf}, the sum rule of classical probability theory, i.e., 
\be\label{e:P!=P+P}
p_{\rm both}(x, t)\neq \frac{1}{2}\left[p_+(x, t) +p_-(x, t) \right].
\ee
is not satisfied in this context. Here the factor $1/2$ yields te probability for the particle to go through one of the slits. The terms $p_-(x)$ and $p_+(x)$, respectively, yields the conditional probability for the particle to hit the screen at position $x$, given that it passes through the slit at position $x_-$ or $x_+$, respectively.

So, this is analogous to the phenomenon of quantum interference. Indeed, if we change the hyperbolic cosine in Eq.~\eqref{em:Iim} for a cosine, i.e., doing $\cosh\to \cos$, Eq.~\eqref{e:im-interf} coincides exactly with the formula for the actual quantum two-slit experiment. 
In other words, Eq.~\eqref{e:im-interf} describes the imaginary-time version of quantum interference. 

Importantly, if there are no internal dynamics, i.e., if there are no factors $G_\ell$, there is no interference. This is because, in such a case, there would be no variables $x_3$ and so no contributions $x_3\neq x_1$.

\

\subsection{A few comments}\label{sm:fewcomm}
To recap, we have shown that scientists can obtain a seemingly observer-independent view of the world at the price of describing it in terms of a real probability matrix that follows an imaginary-time quantum dynamics. This suggests the following conjecture or observer's property:

\

\begin{quote}
\noindent {\bf O1. Embodiment:} The state of an observer explicitly described as a physical object, i.e., an ``observer-as-object,'' (see Sec.~\ref{sm:others...} below) interacting with an experimental system is given by a real probability matrix $P_\ell$. The diagonal elements of $P_\ell$ are the probabilities for the different outcomes of the experiment to occur. $P_\ell$ follows a dynamics given by (see Eq.~\eqref{em:real_vN})
\be\label{em:O1-embodiment}
\Delta P_\ell = \epsilon [J_\ell, P_\ell],
\ee
where $\Delta P_\ell = P_{\ell +1}-P_\ell$, $\ell$ denotes the time step and $\epsilon$ is the time step size.
\end{quote}

Alternatively, scientists can also describe the world in terms of forward and backward BP cavity messages, which are formally analogous to imaginary-time wave functions and their conjugates. These forward and backward BP cavity messages, respectively, follow a BP dynamics described by the imaginary-time Schr\"odiner equation and its conjugate. However, we here focus on probability matrices because these directly yield probabilistic information, unlike BP messages that must be multiplied by another object---its ``conjugate''---to do so. 

We take this as an conjecture here, not as a given, because it is not clear at this point that our analysis applies to generic initial states. Indeed, according to Eq.~\eqref{em:matrix}, given a set of factors characterizing the internal and external dynamics, the initial state is restricted to $P_0=\widetilde{F}_n F_{n-1}\cdots F_0$. Of course, there is some freedom to choose the initial state if some of the factors are considered as preparing the initial state, but it is not clear that this freedom is enough to cover all possible initial states.

Of course, scientists can also use standard probability distributions, if they prefer so. However, they cannot do it in terms of a single-variable marginal, on variables $x_\ell$, evolving according to a Markovian rule. Scientists would have to use a Bernstein process instead and would have to specify a two-variable marginal as initial state (see Eq.~\eqref{em:P_Bernstein})---e.g., the probability for the initial and final states to have a certain value. 

We have also shown how factor graph models with circular topology can be naturally described in terms of Markov-like chains formally analogous to imaginary-time quantum dynamics. These Markov-like chains are similar to the standard Markov chains, which naturally describe factor graph models with linear topology, except that the state of the system is described in terms of probability matrices rather than standard probability vectors. In this sense, Markov chains and imaginary-time quantum dynamics could be considered as instances of linear and circular causality, respectively.

Interestingly, imaginary-time quantum dynamics already displays some quantum-like features~\cite{Zambrini-1987}. For instance, using our framework we have shown that a {\em classical} two-slit experiment can entail constructive interference. 
This provides a fresh perspective to think about quantum interference. When an embodied observer has information about which slit the particle goes through---e.g., when only one slit is open---this has to be reflected in the physical correlates of the experiment ``internal'' to her. So, the variable $x_1$ describing the slit, which is external to the observer, and the variable $x_3$, which is the corresponding physical correlate internal to the observer, must be equal, i.e., $x_3 = x_1$ (see Fig.~\ref{f:slits}). In this view, the imaginary-time version of quantum interference arises because, when an embodied observer cannot access any information about which slit the particle goes through, the values of $x_1$ and $x_3$ do not have to coincide even though they refer to the same ``thing'' (i.e., the slits). This would be true no matter the reason for which the embodied observer lacks ``which-way'' information. Furthermore, our approach considers the whole experimental setup, or context, from beginning to end. So, it could also potentially take account of variations of the two-slit experiments, such as delay choice or quantum erasure experiments. 

Following the tradition in physics, we have focused on the dynamics external to the scientist. However, similar results can be obtained if, following the tradition in cognitive science, we focus on the dynamics internal to the scientist, instead, and consider the external system as hidden to her---this could be done by marginalizing $\widetilde{\mcP}$ in Eq.~\eqref{em:wholeP} over the external paths instead of marginalizing it over the internal paths as we have done. That is, we could obtain an equation from $\ell=n\dotsc , k-1$, similar to Eq.~\eqref{em:P_l+1}, if we focus on the physical correlates of the experimental system, which are internal to the observer, rather than on the experimental system itself, which is external to the observer.

Up to now we have focused on the case of a well-known so-called ``stoquastic'' Hamiltonian $H$ (see Eq.~\eqref{em:H}. 
However, we will describe more general Hamiltonians later on. 

It seems natural to wonder what about actual, real-time quantum dynamics. Up to now we have considered a scientist $F$ performing an experiment with a generic physical system $S$. However, $F$ is also a physical system. So, the combined system $S_{\rm coup} = F+S$ that we have modeled is also a physical system. In our previous analysis $S_{\rm coup}$ appears as an observer-independent physical system. Consistency would require that such an observer also be included in the action-perception loop. It appears that properly dealing with this situation manifests aspects of a real-time quantum dynamics. We will discuss this in Sec.~\ref{sm:first}.

\

\section{Reflexivity and real-time quantum dynamics}\label{sm:first}

According to the analysis in Sec.~\ref{sm:third}, summarized in observer property {\bf O1} (see Sec.~\ref{sm:fewcomm}), scientists can in principle ``escape'' the action-perception loop and describe classical systems in an {\em effective}, seemingly observer-independent way by describing such systems in terms of real-valued probability matrices that follow an imaginary-time quantum dynamics. This suggests that the observer might indeed be key to the quantum formalism, as emphasized in QBism~\cite{debrota2018faqbism,mermin2018making,fuchs2014introduction}. However, the actual quantum formalism does not take place in imaginary time, but in real time. 

Here we explore how real-time quantum dynamics might relate to our model of embodied scientists doing experiments. To do so, we first rewrite von Neumann equation, Eq.~\eqref{em:drhodt=[H,rho]}, which is complex valued, as a pair of real equations related to its imaginary and real parts. We then show that these equations are related to the equations associated to an embodied scientist doing an experiment, Eq.~\eqref{em:dPdt=[H,P]} and its transposed via a ``swap''operation. We also show how an analogous structure can naturally arise when two mirrors mutually reflect each other. 

This suggests that we might obtain real-time quantum dynamics by reflexively coupling two (sets of) observers mutually observing each other. In this case, observers are relative to each other rather than observer-independent or relative to an external, unacknowledged observer. We explore this in the rest of this section. Based on an analogy with reflexive systems, we conjecture some properties characterizing an observer, and we show that these lead to a dynamics with aspects of a genuine, real-time quantum dynamics---we refer to these conjectures collectively as the {\em reflexive coupling hypothesis}. 

We emphasize that this section is not to be considered as a rigorous derivation of real-time quantum dynamics from reflexivity. The reason is that, reflexivity being a subtle and scarcely studied subject, the connection between the ideas of reflexivity and the conjectures we introduce may not be completely clear. We introduce this section in this work because we find it conceptually plausible and we hope it can suggest future research on potential connections between the philosophy of mind and quantum physics. The literature on reflexive systems is rather scarce and we hope that an interdisciplinary approach to this topic could help further clarify or improve the conjectures we introduce here.

\

\subsection{Von Neumann equation as a pair of real equations}\label{sm:swap}

Actual quantum systems are generally described by a complex-valued density matrix $\rho$ satisfying the von Neumann equation, Eq.~\eqref{em:drhodt=[H,rho]}. In order to explore how the actual quantum dynamics relates to the imaginary-time quantum dynamics that we have obtained, which is real-valued, here we will separate its real and imaginary parts. To do so, we use the fact that the density matrix and the Hamiltonian are Hermitian operators, i.e., they are equal to their adjoints: $\rho^\dagger =\rho$ and $H^\dagger = H$. We will focus here on the common case where the adjoint operation $\dagger$ is given by the combination of transpose $T$ and complex conjugate $\ast$ operations, e.g., $[\rho^\dagger] (x, x^\prime) = [\rho(x^\prime , x)]^\ast$. In this case we can write 
\be\label{em:rhomapP}
\rho = P_s + P_a/ i,
\ee
where $P_s = P_s^T$ and $P_a = -P_a^T$ are, respectively, some generic real-valued symmetric and antisymmetric matrices. Without loss of generality, we can write $P_s = (P+P^T)/2$ and $P_a = (P-P^T)/2$ as the symmetric and antisymmetric parts of a generic real matrix $P$. Since the diagonal elements of $P$ and $P^T$ are the same, these are equal to the diagonal elements of $\rho$. That is, the diagonal elements of $P$ are the actual probabilities encoded in $\rho$. So, when the off-diagonal elements of $P$ are non-negative, $P$ is a probability matrix like those we use in Sec.~\ref{sm:third}.

We can write an equation for the Hamiltonian similar to Eq.~\eqref{em:rhomapP}, i.e., 
\be\label{em:HmapJ}
H = -\hbar J_s -\hbar J_a/ i,
\ee
where $J_s = (J+J^T)/2$ and $J_a = (J-J^T)/2$ are the symmetric and antisymmetric parts of a generic real matrix $J$. We write $H$ in terms of $J$ so we do not have to worry about $\hbar$ below. When $J$ can be interpreted in probabilistic terms as above, e.g., when $H$ is given by Eq.~\eqref{em:H}, it is a {\em dynamical matrix}, with non-negative off-diagonal entries, much like those we used in Sec.~\eqref{sm:third}. This suggests that $P$ and $J$ might actually be the most suitable objects to explore the potential relationship between genuine, real-time quantum dynamics and our framework. 

Introducing $\rho = P_s + P_a/ i$ and $H = -\hbar J_s -\hbar J_a/ i$ in von Neumann equation, Eq.~\eqref{em:drhodt=[H,rho]}, and separating the real and imaginary parts, we get a pair of real-valued equations for $P_s$ and $P_a$, i.e. 
\BE
\frac{\partial P_s}{\partial t} &=& [J_s, P_a] + [J_a, P_s],\label{em:vN-Ps}\\
\frac{\partial P_a}{\partial t} &=& -[J_s, P_s] + [J_a, P_a].\label{em:vN-Pa}
\EE 
Adding and subtracting Eqs.~\eqref{em:vN-Ps} and \eqref{em:vN-Pa} we get an equivalent pair of equations for $P$ and $P^T$, i.e.,
\BE
\frac{\partial P}{\partial t} &=& -[J_s, P^T] + [J_a, P],\label{em:vN-1}\\
\frac{\partial P^T}{\partial t} &=& [J_s, P] + [J_a, P^T].\label{em:vN-2}
\EE 
Notice that Eq.~\eqref{em:vN-2} is the transposed of Eq.~\eqref{em:vN-1}.

If the terms $-[J_s, P^T]$ in Eq.~\eqref{em:vN-1} and $[J_s, P]$ in Eq.~\eqref{em:vN-2} were swapped, Eqs.~\eqref{em:vN-1} and \eqref{em:vN-2} would become
\BE
\frac{\partial P_{\rm swap}}{\partial t} &=& [J, P_{\rm swap}] ,\label{em:no-vN-1}\\
\frac{\partial P^T_{\rm swap}}{\partial t} &=& -[J^T, P^T_{\rm swap}] ,\label{em:no-vN-2}
\EE 
since $J=J_s + J_a$, which are the imaginary-time von Neumann equation, Eq.~\eqref{em:dPdt=[H,P]}, and its transpose. According to our previous results, the former can be associated to an embodied scientist doing an experiment and the latter to the time reversal process.

So, it seems that this swap operation might help us bridge our approach with real-time quantum mechanics. For simplicity, since this swap operation only involves the terms with $J_s$, we will first focus on these terms by doing $J_a = 0$ for now. So, taking $J_a = 0$, discretizing Eqs.~\eqref{em:no-vN-1} and \eqref{em:no-vN-2} for convenience and dropping the subindex ``swap'', we have
\BE
\Delta P_\ell &=& \epsilon [J_{s,\ell}, P_\ell] ,\label{em:Pnoswap}\\
\Delta Q_\ell &=& -\epsilon [J_{s,\ell}, Q_\ell] ,\label{em:PTnoswap}
\EE 
where $Q_\ell = P_\ell^T$. Intuitively, Eqs.~\eqref{em:Pnoswap} and \eqref{em:PTnoswap} are related by a time-reversal, $\epsilon\to -\epsilon$---i.e., if Eq.~\eqref{em:Pnoswap} describes the circular process in Fig.~\ref{f:circular_noRQM}C in a clockwise direction, then Eq.~\eqref{em:PTnoswap} will describe it in a counter-clockwise direction. We will explore this further later on. 

\

\subsection{Reflexive coupling: An analogy with mirrors and video feedback}

\subsubsection{An analogy with mirrors}

\begin{figure*}
\includegraphics[width = \textwidth]{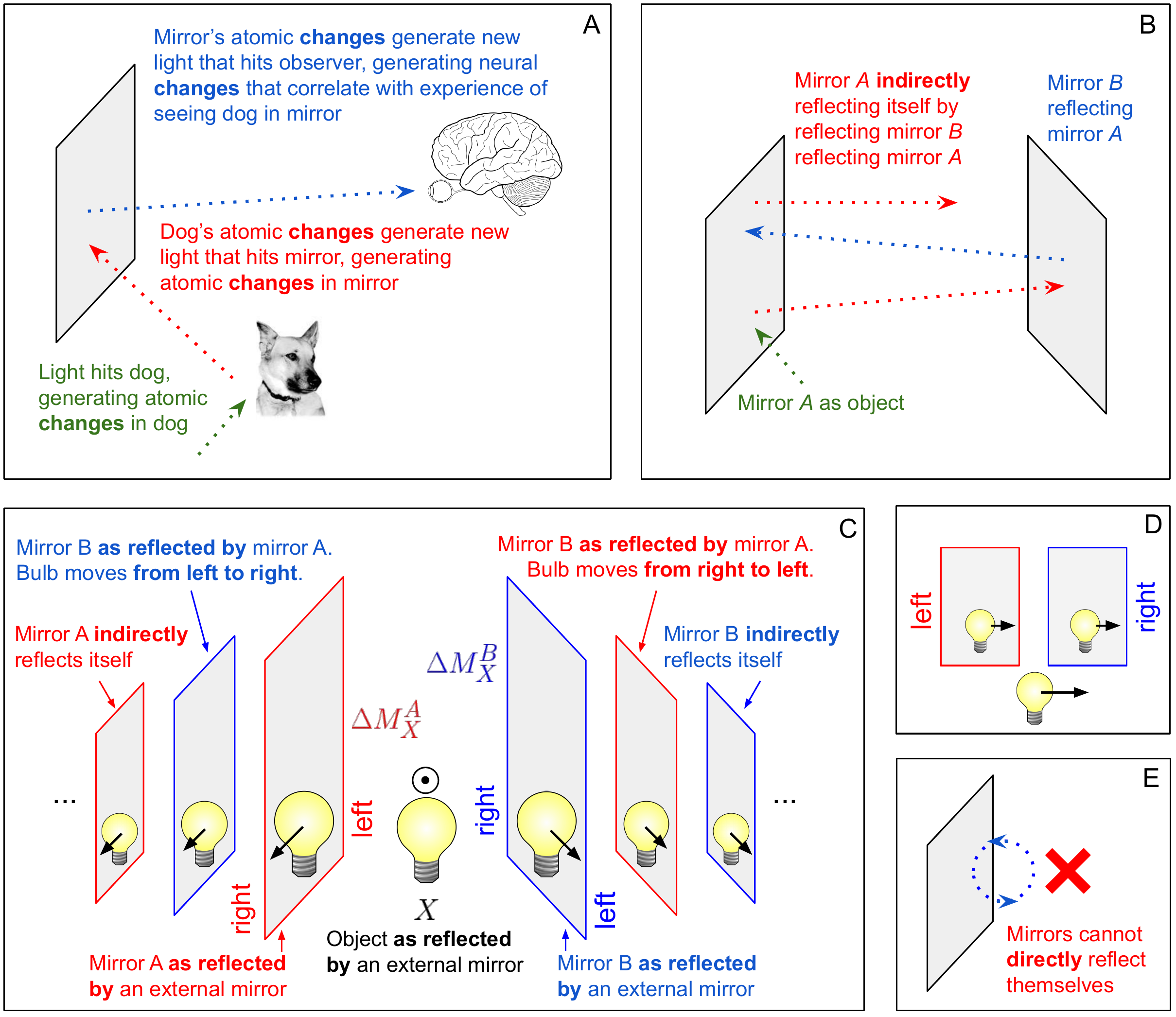}
\caption{{\em Some properties of mirrors:} (A) How do mirrors reflect objects?
When light hits an arbitrary object, say a dog, physical {\em changes} in the dog induce physical {\em changes} in the mirror which, through the observer's eyes, induces {\em changes} in the observer's neural patterns. These neural changes correlate with the experience of seeing the dog reflected in the mirror (see text). 
(B) Mirrors can {\em indirectly} reflect themselves. (C) If there is an object in between the two mutually reflecting mirrors (here a light bulb), they can reflect that object as well as each other reflecting that object. Mirrors $A$ and $B$ as well as the light bulb could also be a reflection in other mirror---say, the black rectangle in which the figure is framed. Importantly, if the object is moving, say, towards outside the page (denoted here by $\bigodot$) then it would appear to the (red) mirror $A$ as moving from left to right and to the (blue) mirror $B$ as moving from right to left instead. We refer to this as an ``apparent time reversal.'' (D) If there is no reflexive coupling between mirrors, i.e., if the mirrors are not facing each other, but are rather parallel to each other, there is no ``apparent time reversal'' as they reflect the object moving in the same direction. (E) Mirrors cannot {\em directly} reflect themselves. }
\label{f:mirror}
\end{figure*}

We have seen that a swap operation turns genuine real-time quantum dynamics, given by Eqs.~\eqref{em:vN-1} and \eqref{em:vN-2}, into two independent imaginary-time quantum dynamics, given by Eqs.~\eqref{em:no-vN-1} and Eq.~\eqref{em:no-vN-2} which are analogous to the equation describing an embodied observer, Eq.~\eqref{em:dPdt=[H,P]}, and its transposed. A similar structure appears when studying reflexive systems such as mirrors. It will then be useful to discuss some aspects of mirror reflection. 

In this analogy with mirrors, ``reflection'' is analogous to ``experience'' or ``observation''. That is, mirrors reflecting objects are analogous to observers experiencing or observing phenomena. In a sense, observers could also ``reflect'' phenomena by communicating them either through language at the conscious level, which is supported by physical processes such as air vibration patterns or ink patterns on paper, or through bioelectric signals at the unconscious level. Untrained observers and trained scientists, respectively, would be analogous to stained and stainless mirrors. We will focus on the latter.

Let us begin with some basic physics of mirror reflection. Consider the process of seeing an arbitrary object on a mirror---say, a dog (see Fig.~\ref{f:mirror}A). When light hits a dog, it induces some physical changes in the dog. Those changes can be energy changes---e.g., some atoms of the dog's body absorb photons and get excited. Such changes produce other physical changes---e.g., the dog's excited atoms can emit photons and relax. These changes in turn generates new light that hits the mirror inducing some physical changes---e.g., the mirror's silver atoms absorb photons, get excited, and then re-emit those photons. Those changes produce more light which can hit the observer's eyes producing some physical changes in the eye's atoms. Finally, these changes can induce changes in the patterns of neural activity, changes that correlate with the experience of seeing the dog reflected on the mirror.

In summary, physical {\em changes} in the dog induce physical {\em changes} in the mirror which, through the observer's eyes, induces {\em changes} in the observer's neural patterns. These neural changes correlate with the experience of seeing the dog reflected in the mirror. In line with this, the first aspect of mirror reflection that we want to consider is the following:

\begin{quote}
{\bf M1. Mirrors reflect objects via physical changes---internal and external.} A mirror reflects an object via the physical changes that light induces on the object. These changes, which are {\em external} to the mirror, generate light that induces physical changes on the mirror. These new changes, which are {\em internal} to the mirror, in turn generate light that induces physical changes on the observer. These latter changes correlate with the experience of seeing the object reflected on the mirror (see Fig.~\ref{f:mirror}A). In directly reflecting an object, the mirror engages two rays of light: one incoming and one outgoing (red and blue arrows in Fig.~\ref{f:mirror}A, respectively).
\end{quote}

Accordingly, we will denote an arbitrary mirror $A$ reflecting a generic object $X$ as 
\be\label{em:dMAdef}
\Delta M^A = \obs[A]{\Delta O^X}\equiv\boxed{\Delta O^X}_A,
\ee
to emphasize that mirror reflection takes place via {\em changes}. Here the rectangle with subindex $A$ denotes mirror $A$. To avoid cluttering the paragraphs with rectangles, we have also introduced the alternative notation of a vertical bar with subindex $A$ (i.e., $\obs[A]{}$) to denote reflection in mirror $A$. We will use the same vertical bar to denote the analogue of reflection in other analogies, i.e., experience or observation in the case of observers, and recording-and-displaying in the case of video-systems.

So, mirrors can take objects as ``input,'' so to speak, and output their reflection. Importantly, an arbitrary mirror $A$ can also be the object being reflected by another arbitrary mirror $B\neq A$---this is denoted as $\Delta M^B = \obs[B]{\Delta M^A}$. In this sense, mirrors are analogous to computer programs or Turing machines. Indeed, computer programs not only take data as input, process them, but they (or, more precisely, their code) can also be the data that other computer programs take as the input to process. The former is an {\em active} role in that mirrors and computer programs perform a function---i.e., to reflect objects and to process data, respectively. The latter is a {\em passive} role in that mirrors and computer programs are just objects and data, respectively, that do {\em not} perform any function at all---rather, a function is performed on them. We can summarize this property of mirrors as follows:

\begin{quote}
{\bf M2. Mirrors play both active and passive roles.} A mirror can both reflect other objects and be the object reflected by other mirrors. The latter is a {\em passive} role: mirror as object. The former is an {\em active} role: mirror as object-reflecting ``subject,'' so to speak---here ``subject'' is used in a strict technical sense as the opposite of object. 
\end{quote}

It is well known that two mirrors can recursively reflect each other, producing a so-called ``infinite mirror'' effect (see Fig.~\ref{f:mirror}B). If there is an object in between the two mirrors, each mirror will reflect both the object and the other mirror reflecting the object (see Fig.~\ref{f:mirror}C). In this way, mirrors can {\em indirectly} reflect themselves, even though they cannot {\em directly} do so (see Fig.~\ref{f:mirror}E). Importantly, the infinite mirror effect appears due to the size distortion, each image appearing smaller than the previous one in the recursion. We do not expect a similar distortion to occur in the case of observers. 

This bring us to the next two properties of mirrors that we want to highlight:

\begin{quote}
{\bf M3. Mirrors can indirectly reflect themselves by reflexively coupling to other mirrors.} According to {\bf M1}, physical changes in mirror $A$ can induce changes in mirror $B\neq A$, which can in turn induce further changes in mirror $A$. At this point, mirror $A$ is reflecting an image of itself. However, in contrast to the {\em direct} reflection described in {\bf M1}, here mirror $A$ engages {\em four} rays of light instead of two: two incoming and two outgoing (see Fig.~\ref{f:mirror}B). In this sense, mirror $A$ {\em indirectly} reflects an image of itself. Moreover, this process of mutually reflecting the physical changes of each other continues {\em ad infinitum}, producing the well-known ``infinite mirror'' effect---this effect is due to size distortion and is not expected to occur in the case of observers. Importantly, each of these changes is external to one mirror and internal to the other; none of these changes can be both internal and external to the very same mirror (see Fig.~\ref{f:mirror}B). If there is an object in between the two mirrors, each mirror will reflect both the object and the other mirror reflecting that same object (see Fig.~\ref{f:mirror}C).
\end{quote}

\begin{quote}
{\bf M4. Mirrors cannot directly reflect themselves.} This would imply that the associated physical changes can simultaneously be both internal and external to the very same mirror, contradicting the situation described in {\bf M3} (see Fig.~\ref{f:mirror}E). Similarly, a mirror cannot directly reflect its own internal changes because these are the very changes that allow the mirror to reflect any external changes at all---i.e., changes induced by light on objects external to the mirror. There is nothing mystical about this. This does not imply that the mirror is not a physical object; of course, it is. It only implies that the mirror's internal changes are fundamentally inaccessible {\em to the mirror itself}, they will always remain implicit---of course, these changes can be accessed and reflected by {\em another} mirror (see Fig.~\ref{f:mirror} B).  
\end{quote}

Property {\bf M4} suggests that every perspective has a blindspot. The situation is analogous to that of an eye that can directly ``see'' any objects in its visual field, but it cannot {\em directly} see itself---of course, it can ``see'' an image of itself, say, in a mirror.

To formalize the situation illustrated in Fig.~\ref{f:mirror}C, consider two mirrors $A$ and $B$ which are reflected by other mirrors, say $W$ and $W^\prime$, respectively. That is,
\BE
\Delta M^{W} &=&\obs[W]{\Delta M^B}\equiv\boxed{\Delta M^B}_{W},\label{em:MW<-MB}\\
\Delta M^{W^\prime} &=& \obs[W^\prime]{\Delta M^A}\equiv\boxed{\Delta M^{A}}_{W^\prime}.\label{em:MW'<-MA}
\EE 
In the particular case in which $W=A$ and $W^\prime = B$, Eqs.~\eqref{em:MW<-MB} and \eqref{em:MW'<-MA} describe the situation illustrated in Fig.~\ref{f:mirror}C when mirrors $A$ and $B$ mutually reflect each other, i.e.,  
\BE
\Delta M^{A} &=&\obs[A]{\Delta M^B}\equiv \boxed{\Delta M^B}_A,\label{em:MA<-MB}\\
\Delta M^B &=& \obs[B]{\Delta M^A}\equiv \boxed{\Delta M^A}_{B}.\label{em:MB<-MA}
\EE 
Equation~\eqref{em:MB<-MA} indicates that mirror $B$ is reflecting mirror $A$. Introducing this into Eq.~\eqref{em:MA<-MB} yields $\Delta M^{A}= \obs[A]{\obs[B]{\Delta M^{A}}}$. Similarly, Eq.~\eqref{em:MA<-MB} indicates that mirror $A$ is reflecting mirror $B$. Introducing this into Eq.~\eqref{em:MB<-MA} yields $\Delta M^{B}= \obs[B]{\obs[A]{\Delta M^{B}}}$. Continuing this process recursively yields the infinite-mirror effect, i.e., 
\BE
\Delta M^A &=& \boxed{\boxed{\boxed{\dots}_A}_B}_A ,\label{em:MB<-MA<-}\\
\Delta M^B &=&\boxed{\boxed{\boxed{\dots}_B}_A}_B .\label{em:MA<-MB<-}
\EE 

Importantly, Eqs.~\eqref{em:MB<-MA<-} and \eqref{em:MA<-MB<-} seems to describe a situation of two empty mirrors reflecting each other, but no other generic objects like the light bulb in Fig.~\ref{f:mirror}C. However, this is not so. It appears to be so because  we have left the other generic object implicit to avoid cluttering the notation. 

It will be useful to explicitly consider the case shown in Fig.~\ref{f:mirror}C of two mirrors reflecting an object moving parallel to the mirrors towards outside the page (denoted in the figure by $\bigodot$). If mirror $A$ reflects the object as moving from its left to its right (see Fig.~\ref{f:mirror}C), then mirror $B$ reflects the object as moving in an opposite direction, i.e., from its right to its left. In this case, Eqs.~\eqref{em:MB<-MA<-} and \eqref{em:MA<-MB<-} can be written more precisely as
\BE
\Delta M^{A} &=&\obs[A]{\Delta M^B}^\leftarrow 
,\label{em:mirror<-}\\
\Delta M^B &=& \obs[B]{\Delta M^A}^\to 
.\label{em:mirror->}
\EE 
Here the arrows $\leftarrow$ and $\to$ are used to indicate that if an object appears moving in one direction in one of the mirrors, it would appear moving in the opposite direction in the other mirror. If we unleash the recursion of Eqs.~\eqref{em:MA<-MB} and \eqref{em:MB<-MA} we get
\BE
\Delta M^A &=& \boxed{\overleftarrow{\boxed{\overleftarrow{{\boxed{\dotsc}_A}}}_B}}_A,\label{em:mirror->unleashed}\\
\Delta M^B &=& \boxed{\overrightarrow{\boxed{\overrightarrow{{\boxed{\dotsc}_B}}}_A}}_B. \label{em:mirror<-unleashed}
\EE

Notice that if there is no reflexive coupling between mirrors, i.e., if the mirrors do not face each other, there is no arrow inversion (see Fig.~\ref{f:mirror}D). This brings us to the next mirror property we want to highlight:
\begin{quote}
{\bf M5. Mirrors' reflexive coupling involves an ``apparent time-reversal'':} Consider two reflexively coupled mirrors, $A$ and $B$, reflecting an external object. If the reflected object appears to move from left to right for mirror $A$, it would appear to move from right to left for mirror $B$. So, if from $A$'s perspective the object has velocity $v$, then from $B$'s perspective the object is moving with velocity $-v$. We will refer to this as an ``apparent time-reversal.'' Alternatively, if mirror $A$ registers a change in position $\Delta x^A = \epsilon v$ then mirror $B$ registers a change in position $\Delta x^B = -\epsilon v$---in this case we could also say that the ``apparent time reversal'' is manifested in the change $\epsilon\to -\epsilon$.
\end{quote}

Of course, the object could move in a diagonal direction and only the horizontal component would be inverted. However, we have to keep in mind that this is only an analogy. We are taking into account only one dimension, the horizontal one, because in the reflexive coupling of observers there would also be only one dimension that is inverted, i.e., the temporal labels.

\

\subsubsection{An analogy with video feedback}\label{sm:video}
\begin{figure}
\includegraphics[width=\columnwidth]{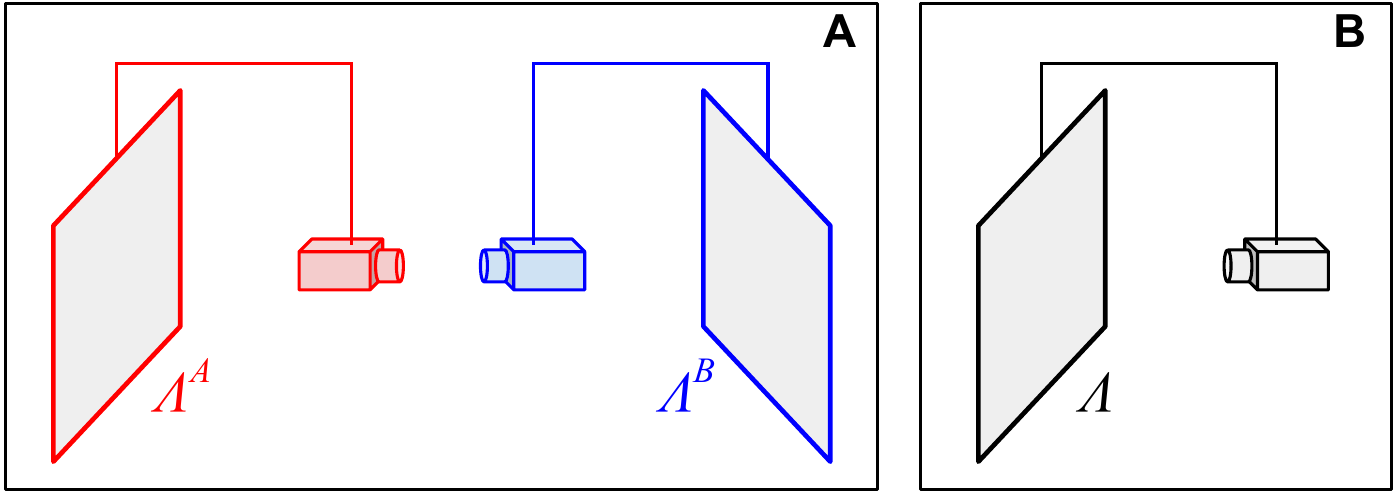}
\caption{{\em Video-feedback analogue of mirror-mirror reflection:} (A) Two video-systems, $A$ and $B$, mutually recording-and-displaying each other. The video screens of systems $A$ and $B$ project images $\Lambda^A_\ell$ and $\Lambda^B_\ell$. (B) A single video-system can partially record-and-display itself. It cannot fully display itself, though; in particular, the camera cannot record itself.}
\label{f:vfeedback}
\end{figure}

To better understand the reflexive coupling between observers, we now resort to a more formal analogy with video-systems (see Fig.~\ref{f:vfeedback}A). In this analogy, each mirror is analogous to a system composed of a video camera connected to a video screen. A mirror receiving light emitted by an object is analogous to the video camera of a video-system recording an object, and that same mirror reflecting the light received is analogous to the video screen displaying the object recorded by the video camera attached to it. 
In Fig.~\ref{f:vfeedback}A the video camera of system $A$ records the video screen of system $B$ and vice versa. While the video camera of a single system can record the video screen of that same system (see Fig.~\ref{f:vfeedback}B), the video camera cannot record itself. So, a video-system composed of both video camera and screen, which together are the analogue of a mirror, cannot completely record-and-display itself either. 

Crutchfield~\cite{crutchfield1984space} analyzed the case shown in Fig.~\ref{f:vfeedback}B, where only one video camera records the very same screen to which it is connected. He formalized the situation along the following lines: Let $\Lambda_\ell$ be the image displayed in the screen at time step $\ell$; that is, $\Lambda_\ell$ is a squared array of pixels, where each pixel can be any gray color, from white to black. 
The image at the next time step, $\ell + 1$ is given by $\Lambda_{\ell+1} = \alpha\Lambda_\ell + \epsilon\mcD\left(\Lambda_\ell\right)$. The first term in the right hand side characterizes the memory of the system---when $\alpha = 0$ the system does not remember the previous image. Here we are interested only in the case with $\alpha = 1$. For this reason we will omit $\alpha$ from now on. 

So, the video-system in Fig.~\ref{f:vfeedback}B satisfies the equation
\be\label{em:JC}
\Delta\Lambda_\ell = \epsilon\mcD\left(\Lambda_\ell\right).
\ee
Here $\Delta\Lambda_\ell\equiv\Lambda_{\ell+1} - \Lambda_\ell$, $\epsilon$ is the time step size and the function $\mcD$, characterizing the recording-and-displaying relation, can include, e.g., a rotation or a scaling of the image at time step $\ell$---i.e., by rotating or zooming in or out the video camera, respectively. The time step size, $\epsilon$, is introduced because we are interested in the limit $\epsilon\to 0$ when this difference equation becomes a differential equation. 

As we already mentioned, the single video-system shown in Fig.~\ref{f:vfeedback}B and described by Eq.~\eqref{em:JC} cannot record-and-display an image of itself since the video camera cannot record itself. A reflexive coupling between two video-systems, like two mirrors reflecting each other, can record-and-display an image of itself.

In analogy with mirrors, we can formalize the situation in Fig.~\ref{f:vfeedback}A by first considering the situation in which a video-system $B$ record-and-display another video-system $A$, respectively. Let $\Lambda_\ell^X$ denote the state of video-system $X$ at time step $\ell$---here $X\in\{A,B\}$. At time step $\ell$ the states of the video-systems $A$ and $B$ are $\Lambda_\ell^A$ and $\Lambda_\ell^B$, respectively. Since the video-system $B$ is recording-and-displaying the video-system $A$, at time step $\ell+1$ the state of the video-system $B$ is 
\be\label{em:video-system}
\Lambda^B_{\ell+1} = \Lambda^B_\ell + \epsilon\mcD_B\left(\Lambda_\ell^A\right)
\ee
Importantly, the first term in the right hand side of this equation is $\Lambda_\ell^B$, not $\Lambda_\ell^A$, because this is a memory term---the video-system $B$ has memory about its {\em own} previous state, not about $A$'s. Similarly, the function $\mcD_B$ has the subindex $B$, not $A$, because it is the camera of video-system $B$ that captures and can transform (e.g., rotate or scale) an image of the video-system $A$. We can write this equation as 
\be
\Delta^{\rm i} \Lambda_{\ell}^B = \obs[B]{\Delta^{\rm e}\Lambda_\ell^A}^{\to}\equiv \epsilon\mcD_B\left(\Lambda_\ell^A\right),\label{em:VFRA}
\ee
where $\Delta^{\rm i}\Lambda_\ell^B = \Lambda_{\ell+1}^B - \Lambda_\ell^B $ are changes internal to the video-system $B$, as emphasized by the superindex ``i.'' In analogy with mirrors, we here denote by $\obs[B]{\Delta^{\rm e}\Lambda_\ell^A}^{\to}$ the changes in $A$ as recorded-and-displayed by the video-system $B$---these changes are external to the video-system $B$, as emphasized by the superindex ``e.'' In this case $\obs[B]{\Delta^{\rm e}\Lambda_\ell^A}$ is given by $\epsilon\mcD_B\left(\Lambda_\ell^A\right)$. Furthermore, the arrow $\to$ plays a role analogous to the arrow in Eq.~\eqref{em:mirror->}, i.e., it indicates that Eq.~\eqref{em:VFRA} is one of a pair describing the reflexive coupling between video-systems $A$ and $B$. So, a reverse arrow $\leftarrow$ is analogous to the arrow in Eq.~\eqref{em:mirror<-}. 

To complete the reflexive coupling we have to consider the case where $A$ and $B$ play the complementary roles of the ``subject'' that records-and-displays and the object being recorded-and-displayed. Again, here the word ``subject'' is used in a strict sense to denote the opposite of object. In analogy with Eq.~\eqref{em:VFRA}, this yields the equation
\be
\Delta^{\rm i}\Lambda_{\ell}^{A} =\obs[A]{\Delta^{\rm e}\Lambda_{\ell}^B}^{\leftarrow}\equiv \epsilon\mcD_{A}\left(\Lambda_\ell^B\right). \label{em:VFR'B}
\ee
Equations~\eqref{em:VFRA} and \eqref{em:VFR'B} are analogous to Eqs.~\eqref{em:mirror->} and \eqref{em:mirror<-} for mirrors. They implement the reflexive coupling between the video-systems $A$ and $B$.

\subsection{Reflexive coupling: Observers mutually observing each other}\label{sm:reflexive-observers}
\subsubsection{Scientists doing experiments are observed by other scientists}\label{sm:others...}
\begin{figure*}
\includegraphics[width = \textwidth]{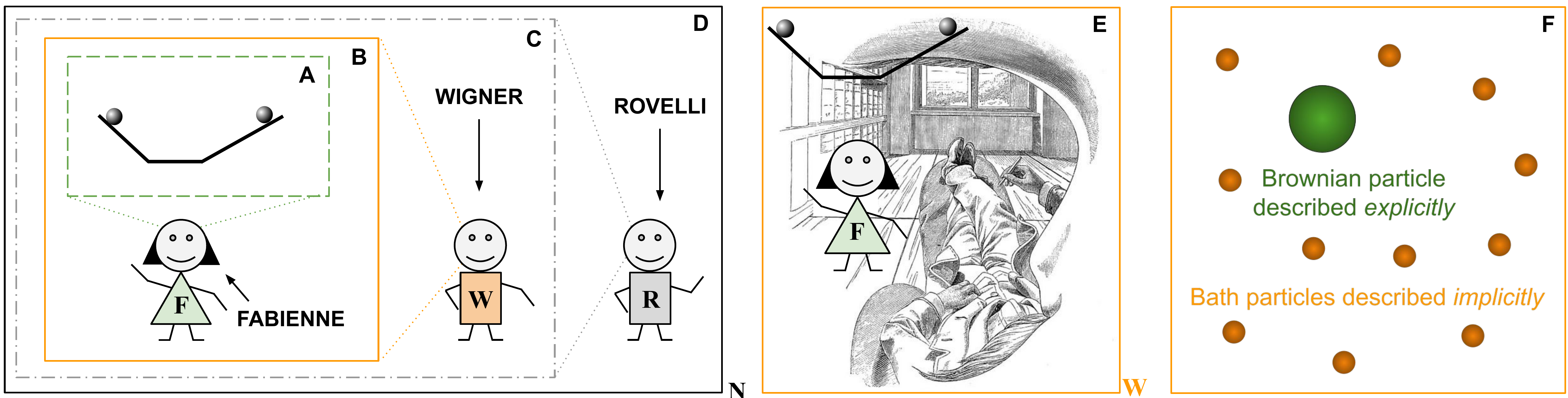}
\caption{{\it Escaping the infinite regress via the first-person perspective (figures {\rm A-D} are nested):} (A) Traditionally, physics focuses on the experimental system alone (inner green dashed box), effectively ignoring the observer. (B) Traditionally, cognitive science focuses on the observer (Fabbiene, Wigner's friend) interacting with an experimental system (middle orange solid box), ignoring the external observer (Wigner) that observes the former. (C) Rovelli's RQM~\cite{Rovelli-1996} hints at the need to take explicit account of Wigner, as we do here. (D) However, this more relational perspective would still ignore the observer who is looking at Wigner. By adding another observer (Rovelli) we are headed to an infinite regress (cf. Fig. 5.1 in Ref.~\cite{rovelli2007quantum}). One way to escape the infinite regress is by distinguishing between the first- and third-person perspectives, between observer-as-subject and observer-as-object (see text). (E) Please imagine that you take the role of Wigner in (C), so you are observing the situation depicted within the solid orange box from your own first-person perspective (1PP). The way you experience the situation is analogous to that of the gray man lying down in Mach self-portrait in (E). You can see Fabienne (and her experimental system), but you cannot see {\em relevant} aspects of yourself. Of course, you can see some aspects of yourself, e.g., parts of your body, but not the physical correlates associated to your experience of seeing Fabbiene, which are our focus here. We will represent the observer-as-subject with the rectangle within which the figure is framed; the letter in the bottom right of this rectangle indicates the observer-as-subject we are referring to (here $W$ for Wigner). (F) The relevant degrees of freedom of an observer-as-subject, represented by the man lying down in (E), are analogous to the degrees of freedom of a thermal bath in that they are inaccessible---however, the former are inaccessible in principle while the latter are inaccessible only in practice (see text). The degrees of freedom of an observer-as-object immersed in the field of experience of an observer-as-subject, so to speak, are analogous to those of a Brownian particle immersed in a thermal bath (see text).  
}
\label{f:homunculus}
\end{figure*}


Physical experiments are usually described as observer-independent systems (see Fig.~\ref{f:homunculus}A). This work aims at understanding how can scientists escape the action-perception loop to establish such a seemingly observer-independent science. To do so, we have been investigating the circular interaction between a scientist, Fabienne ($F$), and an external physical system, $S$ (see Fig.~\ref{f:homunculus}B; see also Figs.~\ref{f:effective}B and \ref{f:circular_noRQM}A). 
Now, not only $S$, but also Fabienne are physical systems. So, the coupled system, $S_{\rm coup} = F+S$, constituted by the scientist, $F$, interacting with the experimental system, $S$, can be considered as a physical system too. 

This brings us back to the beginning of this work: we should not {\em a priori} describe any physical system---in particular $S_{\rm coup}$---as an observer-independent system. 
For our approach to be consistent we have to take into account the observer that observes $S_{\rm coup}$---let us call this observer Wigner ($W$; see Fig.~\ref{f:homunculus}C). That is, Wigner and $S_{\rm coup}$ are in principle also part of an action-perception loop. The role of Wigner is typically played by cognitive scientists who observe other human beings interact with the world---indeed, figures like Fig.~\ref{f:homunculus}B are common in the cognitive science literature. So, we need to ask: How can cognitive scientists too escape the action-perception loop and properly describe the coupled systems of humans and their surroundings as observer-independent systems? Indeed, much like experimental systems in physics, cognitive systems are often described in the literature as observer-independent.  

Interestingly, this approach is in line with Rovelli's relational interpretation of quantum mechanics (RQM), which posits that ``the fact that a certain quantity $q$ has a value with respect to [an observer $F$] is a physical fact; as a physical fact, its being true, or not true, must be understood as relative to an observer, say $[W]$.''~\cite{Rovelli-1996} (see Sec.~II~D therein). However, in RQM observers are considered as generic {\em quantum} system---i.e., as far as physics is concerned, there are no relevant differences between electrons and observers. In contrast, here we are treating observers as {\em classical} cognitive systems. 

Figure~\ref{f:homunculus}C illustrates this. It shows a cognitive scientist (Wigner) looking at another scientist (Fabienne, Wigner's friend) doing an experiment. In principle, we should build a model of Wigner interacting with his friend and the experimental system. But then again the coupled system $S^\prime_{\rm coup} = W + S_{\rm coup}$ of Wigner, his friend and the experimental system is also a physical system and would therefore be relative to another unacknowledged external observer. In principle, we could also make explicit such an additional external observer (Rovelli in Fig.~\ref{f:homunculus}D), but then we would be headed to an infinite regress. That is, we would have to keep on adding external observers {\em ad infinitum}.

This brings us to a potentially subtle point. We conjecture that one way to escape such an infinite regress involves two steps. First, we have to distinguish that observers, similar to mirrors, can play complementary roles as objects and as object-experiencing ``subjects''---here the word ``subject'' is used in the strict technical sense of the opposite of object. Second, we have to implement a reflexive coupling between two (sets of) observers mutually observing each other. 

Let us now describe the first of these two steps in more detail. Please look at Fig.~\ref{f:homunculus}E and imagine that you are experiencing the situation depicted in it---that is, imagine that you play the role of the man in Mach's self-portrait. From a cognitive science perspective, when you look at the system $S_{\rm coup} = F+S$ depicted in Fig.~\ref{f:homunculus}E there is a physical interaction between you and $S_{\rm coup}$---e.g., light reflected from $S_{\rm coup}$ interacts with your eyes. Such a physical interaction generates physical processes inside you that correlate with your experience of observing $S_{\rm coup}$---e.g., the corresponding neural correlates of that experience. If there were no interaction between $S_{\rm coup}$ and you, it would not be possible for neural processes inside you to correlate with the experience of $S_{\rm coup}$. Such correlations are built through (direct or indirect) physical interactions.

Notice that your physical interaction with $S_{\rm coup}$ and the physical processes generated inside you remain unobservable or implicit to you, even though they play a key role in your ability to experience $S_{\rm coup}$. In other words, you cannot simultaneously observe both $S_{\rm coup}$ and the physical correlates associated to your experience of observing $S_{\rm coup}$. Let $\{S_{\rm coup}\}$ denote the latter. If you were to simultaneously observe both $S_{\rm coup}$ and $\{S_{\rm coup}\}$ you would not be observing $S_{\rm coup}$ anymore, but a different object, i.e., $S_{\rm coup}$+$\{S_{\rm coup}\}$, and there would be new physical correlates associated to this new experience, i.e., $\{S_{\rm coup}$+$\{S_{\rm coup}\}\}$. Please remember that here we do not want to jump ahead with assumptions, but to try to model things as explicitly as possible. So, we do not want to {\em a priori} neglect any physical processes associated to you in this example. Rather, we want to understand {\em a posteriori} how is it that we can do so, if indeed we can.

Thus, observers can play two different roles. One is the role that Fabienne, $F$, plays for you, dear reader, when you look at Fig.~\ref{f:homunculus}E: she appears to you as an {\em explicit} physical system, {\em external} to you, which therefore you can in principle model in full mechanical detail. In this sense, Fabienne plays the role of an {\em object} of observation for you, or from the perspective of any other observer different from Fabienne. This is analogous to the role a mirror plays when it is the object being reflected by another mirror---i.e. mirror as object (see mirror's property {\bf M2}). 

Instead, like mirrors that cannot directly reflect themselves, observers cannot directly, simultaneously, fully observe themselves. When you observe an object, there are key physical processes associated to you, which enable you to experience the observed object, and yet remain unobservable or {\em implicit} to you. Those physical processes do not appear to you as an object of observation. So, we will technically say that the role you play for yourself is that of a ``subject''---in the strict technical sense of the opposite of object, or not an object of observation {\it for you}. Of course, those physical processes can in principle appear to others as objects of observation, but not to you. There is nothing mysterious here. We have already described the mirror analogue of this in mirror properties {\bf M1} and {\bf M2}.

In sum, from a third-person perspective (3PP), observers appear as objects of observation. From a first-person perspective (1PP), instead, observers cannot fully appear as objects of observation. There are physical processes that cannot be observed from a 1PP because they are the very processes that enable observers to experience any object at all. Let us summarize this in the following conjectures or observer's properties:

\begin{quote}
{\bf O2. Observer-as-object:} this is how an observer appears to other observers---from a 3PP---i.e., as an {\em explicit} physical system or ``object'' that can be directly experienced by other observers. So, an observer-as-object can be modeled with an explicit mechanical model. When referring to a particular observer-as-object Wigner we can say ``Wigner-as-object'' for short and denote it $W$, as usual.
\end{quote}

\begin{quote}
{\bf O3. Observer-as-subject:} this is how an observer appears to herself---from a 1PP---i.e., as an {\em implicit} physical system or ``subject'' that can directly experience objects, including other observers-as-objects, but cannot {\em directly} experience key physical processes associated to herself---e.g., she cannot directly experience both a dog and the physical correlates associated to her experience of that dog. However, in principle she can {\em indirectly} experience those key aspects of herself via, e.g., a picture of them. When referring to a particular observer-as-subject, say, Wigner, we can say ``Wigner-as-subject'' and denote it as $\I{W}$---we will explain this notation below. An observer-as-subject cannot be modelled in terms of objects following cause-effect mechanisms (see below and observer's property~{\bf O4}).
\end{quote}

Let us now explain the notation $\I{W}$ for an observer-as-subject $W$. Please imagine again that you, dear reader, take the role of observer-as-subject and look at $F$ in Fig.~\ref{f:homunculus}E from your own 1PP. You can refer to that experience as ``I observe $F$''---here you are describing yourself from a 1PP as a subject, as an ``I.'' In contrast, if Wigner looks at you, while you observe $F$, you can be referred to by Wigner as ``he observes $F$''---here you are being described from a 3PP as an ``object,'' as a ``he.''

In analogy with this, we will use the word $I_W$ to refer to observer-as-subject $W$ or, more precisely, to $W$'s 1PP. However, to emphasize that an observer-as-subject cannot directly, fully appear to himself as an object of observation, we will cancel this expression, i.e., $\text{\sout{$I$}}_W$. This notation is inspired in Heidegger's {\em sous erasure}. We can use the convention that the black rectangle on which Fig.~\ref{f:homunculus}E is framed, i.e., $\boxed{\textcolor{white}{\cdot}}_W$, refers to the observer-as-subject and neglect Mach's self-portrait for simplicity. The letter $W$ at the bottom right of that rectangle makes explicit to which observer-as-subject we are referring to---so, $\text{\sout{$I$}}_W = \boxed{\textcolor{white}{\cdot}}_W$. 
Again, we should not fall into the trap of thinking that the rectangle {\em is} the observer-as-subject. Doing so would immediately turn the observer-as-subject into an object of observation, and we would fall back into the {\em infinite regress} of Figs.~\ref{f:homunculus}A-D. In this respect, the symbol $\text{\sout{$I$}}_W$, or the rectangle framing the figure, play a role analogous to that of the number $0$ in that {\em it does not denote a thing but an absence of thing} (cf.~Ref.~\cite{deacon2011incomplete} ch. 0). 
Finally, we will use $\left[\text{\sout{$I$}}_W\right]$ to denote those physical processes that remain unobservable or implicit to observer-as-subject $\text{\sout{$I$}}_W$.

An analogy with consciousness neuroscience would help us be more precise. According to consciousness neuroscience a person, say Wigner, can process information about an external system, $S$, either consciously or unconsciously. In both cases Wigner has information about $S$ in the sense that $S$ ``can be invariantly recognized and influence motor, semantic, and decision levels of processing''~\cite{dehaene2017consciousness}. For this to be possible, in both cases there must be physical correlates that allows Wigner to identify the system $S$, i.e., to discriminate $S$ as different from any other system $S^\prime\neq S$---let us denote these physical correlates as $[S]$. However, in the case of conscious processing Wigner additionally notices the presence of $S$ and can reliably report its identity--in the case of unconscious processing this does not happen. For this to be possible, besides $[S]$ there must be other physical correlates that allows Wigner to notice and report the identity of $S$. These are the so-called neural correlates of consciousness and are somewhat analogous to $[\I{W}]$ here.  

Perhaps a better analogy is with Thompson's recent neurophenomenological perspective that, instead of the conscious/unconscious taxonomy, proposes to distinguish between experience-as-such---the mere potential to experience something---and the contents of experience~\cite{thompson2014waking}. Here $\I{W}$ might be taken as the analogue of experience-as-such. Importantly, for Thompson experience-as-such cannot become a content of experience because it is the very precondition for experiencing any content at all. In a sense, it is always in the background. This is somewhat analogous to mirror's property {\bf M4} and to the fundamental unobservability of $[\I{W}]$. 

In analogy with this, we assume that the physical correlates of an experience, $\{S\}=[S] + [\I{W}]$ are of two kinds, those that allows us to discriminate a content of experience from any other content of experience, $[S]$, and those who allows us to experience any content at all, $[\I{W}]$. In our approach what specifies the system $S$ and the corresponding physical correlates, $[S]$,  is the Hamiltonian function, $\mcH$. In contrast, $[\I{W}]$, which is present in any conscious experience, cannot be turned into an object of observation but, as we will conjecture, could be modeled implicitly as all-pervasive, irreducible fluctuations characterized by the analogue of Planck's constant. Again, we do not want to jump ahead with assumptions and neglect $[\I{W}]$ {\em a priori}. Of course, we can neglect it {\em a posteriori} if the analysis suggests so. 

To summarize, when an observer-as-subject $\text{\sout{$I$}}_W$ observes an observer-as-object $F$ interact with an experimental system $S$, the degrees of freedom associated to $F$ and $S$ are in principle accessible to $\text{\sout{$I$}}_W$, but the degrees of freedom associated to $\left[\text{\sout{$I$}}_W\right]$ are not. The situation is somewhat analogous to that of a mass attached to a mechanical spring which is immersed in a gas. We can usually build an explicit mechanical model of the mass attached to the spring because the corresponding degrees of freedom are easily accessible. In this sense, the mass attached to the spring is analogous to $S_{\rm coup} = F+S$, which is {\em directly} accessible to an observer-as-subject. In contrast, the degrees of freedom associated to the gas are in practice inaccessible because there are too many molecules to track. In this sense, the gas is analogous to $\left[\text{\sout{$I$}}_W\right]$, which is also inaccessible to an observer-as-subject. 

Physicists usually deal with this situation by modeling the gas as a thermal bath whose exchanges of energy with the mass and the spring are modelled statistically as random fluctuations, or noise, characterized by a temperature parameter. We could follow a similar strategy and model $\left[\text{\sout{$I$}}_W\right]$ as a kind of thermal bath whose exchanges of energy with $F$+$S$ are modeled statistically as random fluctuations, or noise, characterized by a temperature-like parameter. This suggests that the parameter $\Gamma$ in Eq.~\eqref{em:F} can be interpreted precisely as this temperature-like parameter (see $W$ in Fig.~\ref{f:homunculus}E; cf. Fig.~\ref{f:homunculus}F). 

The analogy is even closer if instead of a mass-spring system in a gas, we consider a Brownian particle immersed in a fluid. Here, again, we can build an explicit mechanical model of the Brownian particle---in this case, the model is simply that of a free particle that would move with constant velocity, except for its interactions with the molecules of the fluid. In this sense, the Brownian particle is analogous to $F$+$S$. The fluid, like the gas, can be modeled as random fluctuations, or noise, that change the velocity of the Brownian particle and are characterized by a diffusion constant, $D$. 

If the Brownian particle is at position $x$ at a given time step $t$, at the next time step, $t+\epsilon$, it would be in position $x^\prime$ with probability $\mathcal{P}_{\rm Brown} = e^{-\left(x^\prime - x\right)^2/4 D\epsilon}/Z$, 
where $\epsilon$ is the time step size and $Z$ is the normalization constant. Notice that $\mathcal{P}_{\rm Brown}$ have the exact same mathematical form as Eq.~\eqref{em:F} for the case of a free particle---i.e., when $V(x) = 0$ in Eq.~\eqref{em:non-relativisticH}. The role of $D$ is played by $\hobs/2m$, where $m$ is the mass of the free particle. The analogy can be taken further if we notice that the diffusion coefficient can be written as $D = \kappa (T/\eta)/r$, where $r$ is the radius of the Brownian particle, while $T$ and $\eta$ are the temperature and viscosity of the fluid---$\kappa$ is a constant. So, $m$ is analogous to $r$ in that they both characterize the accessible degrees of freedom, and $\hobs$ is analogous to $T/\eta$ in that they both characterize the inaccessible degrees of freedom. 

There is an important difference, though, between the inaccessibility of the degrees of freedom of an observer-as-subject $W$, i.e., $\left[\text{\sout{$I$}}_W\right]$, and the inaccessibility of the molecular degrees of freedom of a fluid or a gas. Indeed, the latter are inaccessible only {\em in practice}, i.e., we could access them if we have powerful enough technologies. In contrast, Wigner's own degrees of freedom are always inaccessible to him, i.e., they are {\em fundamentally} inaccessible to Wigner. Interestingly, this is in line with the idea that the randomness associated to quantum physics is irreducible, i.e., that it cannot be described in terms of lower level mechanisms---like the fluid's molecular collisions in the case of a Brownian particle.  

Importantly, in spite of their irreducibility we can still in principle infer the influence of the inaccessible degrees of freedom on the accessible ones by measuring the random fluctuation of the latter. This is analogous to our ability to measure the temperature characterizing the influence of a thermal bath on a spring-mass system, or the diffusion coefficient characterizing the influence of a fluid on a Brownian particle, by measuring the random fluctuations on the spring, or on the Brownian particle.

Let us summarize this in the following conjecture or observer's property:

\begin{quote}
{\bf O4. Observers-as-subjects can be implicitly modeled as pervasive, irreducible fluctuations:} The in-principle accessible degrees of freedom associated to objects, including observers-as-objects, can be modeled via explicit dynamical models. In contrast, the in-principle inaccessible degrees of freedom associated to an observer-as-subject can be modeled as a kind of thermal bath, or fluid, that induces noise in {\em any} object being observed by the observer-as-subject
Such pervasive, irreducible fluctuations are characterized by the temperature- or diffusion-like parameter $\Gamma$ (see Eq.~\eqref{em:wholeP}).
\end{quote}

\begin{figure}
\includegraphics[width=\columnwidth]{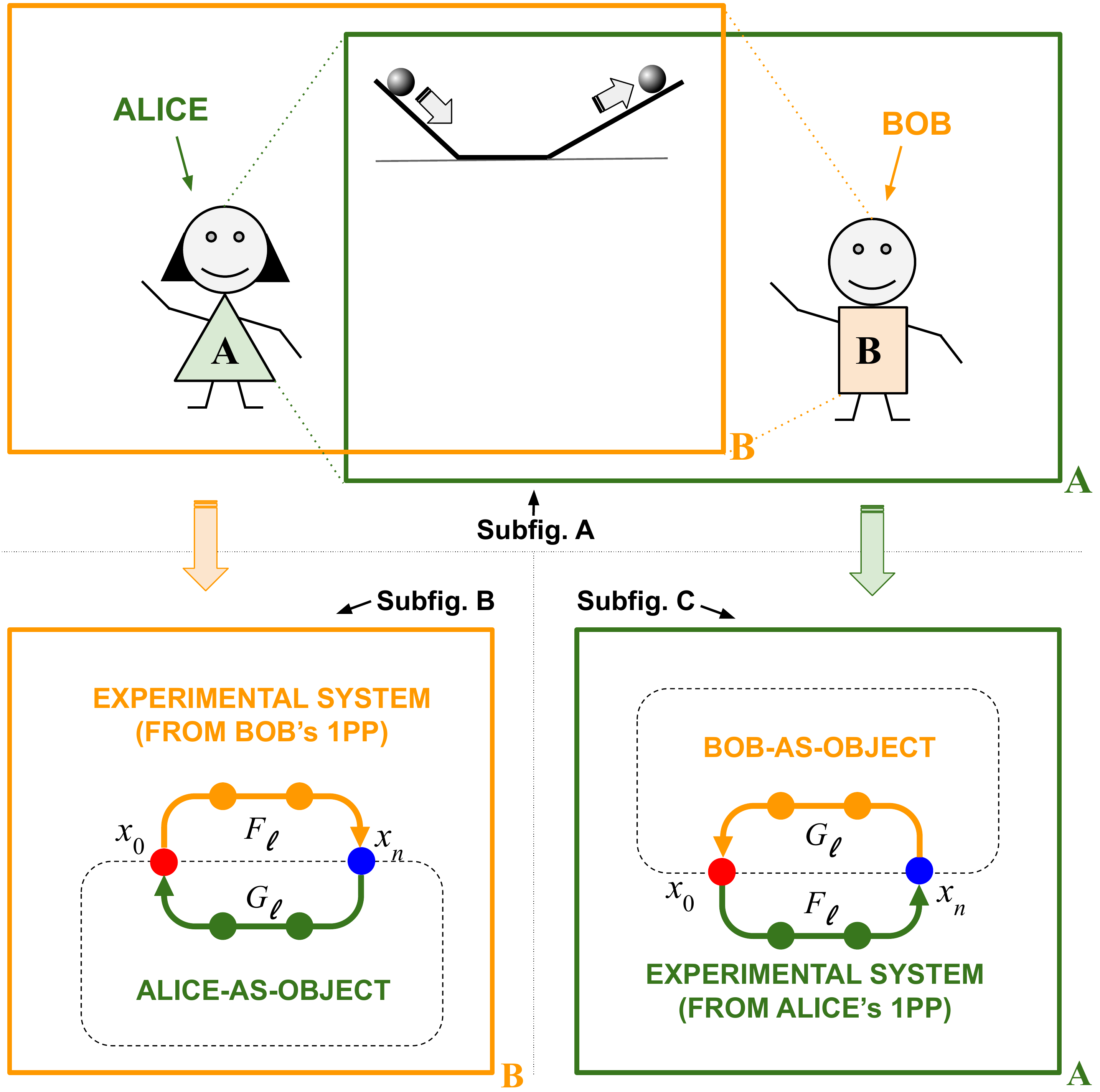}
\caption{{\it Reflexive coupling between observers:} (A; top half) \textsc{Alice} and \textsc{Bob} mutually experience each other, while observing the same experimental system. In this case, \textsc{Alice} and \textsc{Bob} play complementary roles as both the objects being experienced by and the subjects who experience each other. We remind the reader that a rectangle with a letter in the bottom right corner represents an observer-as-subject. (B,C) The bottom figures show the corresponding graphical models. Importantly, a change in perspective inverts the sense in which the dynamics flow. Indeed, from \textsc{Bob}'s first-person perspective (1PP), the dynamics of \textsc{Alice}-as-object interacting with the experimental system flow clockwise (B; bottom left). In contrast, from \textsc{Alice}'s 1PP, the dynamics of \textsc{Bob}-as-object interacting with the experimental system flows counter-clockwise (C; bottom right). We refer to this as an ``apparent time reversal'' (see observer's property {\bf O5} and mirror's property {\bf M5}) }\label{f:mind-body}
\end{figure}

So, distinguishing between observer-as-object, $F$, and observer-as-subject, $\I{W}$, allows us to stop the infinite regress because we do not need to keep on adding observers-as-objects {\em ad infinitum}. However, this step requires us to assume that observers play complementary roles as {\it both} subject and object. So, our description is still incomplete because we are taking account of $F$ as an object but not as a subject and of $W$ as a subject but not as an object (see Fig.~\ref{f:homunculus}E). We can take full account of these complementary roles that observers play by implementing a reflexive coupling where $F$ and $W$ mutually observe each other interact with the same experimental system $S$ (see Fig.~\ref{f:mind-body} where we use names Alice and Bob instead to emphasize this more symmetric situation).

This is analogous to the situation of two mirrors mutually reflecting both each other and a light bulb. As described in {\bf M3}, in this situation each mirror also plays complementary roles as the object being reflected by the other mirror and the object-reflecting ``subject'' that reflects the other mirror. As described in {\bf M5}, if for one of the two mirrors engaged in a reflexive coupling the object reflected moves with (horizontal) velocity $v$, for the other mirror it moves with velocity $-v$. We have referred to this as an ``apparent time reversal.'' 
In the reflexive coupling between two observers we conjecture there is also an ``apparent time-reversal.'' 
Before discussing the implementation of the reflexive coupling, let us discuss in more detail this ``apparent time reversal.'' 

\subsubsection{Constraints for a reliable science}\label{sm:constraints-2}
Here we start discussing the reflexive coupling between two observers, $\textsc{Alice}=(\I{A}, A)$ and $\textsc{Bob}=(\I{B}, B)$ (see Fig.~\ref{f:mind-body}). These new names and new notation highlight that we are considering the more symmetric situation wherein each observer plays the two complementary roles as both subject and object. We here discuss how the perspectives of \textsc{Alice} and \textsc{Bob}, which are involved in a reflexive coupling, are related by an ``apparent time reversal.'' In doing so, we also point out how the reliability conditions {\bf R1}-{\bf R3} in Sec.~\ref{sm:intro} might be incorporated in our framework. In Sec.~\ref{sm:acoplamiento} we continue the discussion of the reflexive coupling.

In our relational approach what we consider the experimental system external to an observer-as-object, say {\sc Alice} (or $A$) in Fig.~\ref{f:mind-body}B, does not exist in an absolute sense. Rather it is external to {\sc Alice} in the sense that it can be experienced by {\em other} observer, say {\sc Bob} (or $B$), different from {\sc Alice}---i.e., $B\neq A$. Figure~\ref{f:mind-body}B describes the situation we have been considering till now, wherein {\sc Bob} and his friend {\sc Alice} play the role of subject and object, respectively. From {\sc Bob}'s 1PP, factors $F_\ell$ and $G_\ell$, respectively, correspond to {\sc Bob}'s experience of the experimental system, which we here denote as $[S]_B$, and {\sc Bob}'s experience of the physical correlates of {\sc Alice}'s experience of the experimental system, which we denote as $[[S]_A]_B$---let us denote this situation as $[S]_B \sim  F_\ell $ and $[[S]_A]_B \sim  G_\ell $; here $\ell = 0,\dotsc , n$.

Figure~\ref{f:mind-body}C shows the complementary situation wherein the roles of {\sc Bob} and his friend {\sc Alice} are reversed---i.e., now {\sc Bob} is the object being observed by {\sc Alice}-as-subject. The reliability condition {\bf R1}---standardization---requires that this situation be described in the same manner. That is, we should use the same factors $F_\ell$ and $G_\ell$. Now, from {\sc Alice}'s 1PP, factors $F_\ell$ and $G_\ell$, respectively, correspond to {\sc Alice}'s experience of the experimental system, $[S]_A$, and {\sc Alice}'s experience of the physical correlates of {\sc Bob}'s experience of the experimental system, $[[S]_B]_A$---in short, $[S]_A \sim  F_\ell $ and $ [[S]_B]_A \sim G_\ell $. 

The reliability condition {\bf R2}---intersubjectivity---requires that, at each time step, what {\sc Bob}-as-subject considers as the experimental system, $[S]_B\sim F_\ell$, which is external to his friend {\sc Alice}, be the same as the experimental system that {\sc Bob}-as-object observes via the corresponding physical correlates $[[S]_B]_A\sim G_\ell$. Under an exchange of roles, the initial factor $F_0(x_1, x_0)$, for instance, characterizing the experimental system observed by {\sc Bob}-as-subject (see Fig.~\ref{f:mind-body}B), corresponds to the final factor $G_0 \equiv F_{2n-1}(x_0, x_{2n-1})$ characterizing the experimental system observed by {\sc Bob}-as-object performing the experiment (see Fig.~\ref{f:mind-body}C). However, the latter has to be transposed because it corresponds to paths traversed in the reversed direction. More generally, we have
\be\label{em:F=FT2}
G_\ell = F_{2n-\ell- 1} = F^T_\ell;
\ee
Here $\ell=0,\dotsc , n$. So, here again the factors $F_\ell$ and $G_\ell$ are the same, except for the fact that they correspond to paths traversed in reverse directions and this manifests in a transpose operation. 

This implies in particular that
\be\label{em:P_l=FFT}
\begin{split}
P_0 &= \frac{1}{Z} G_{0}\cdots G_{n-1} F_{n-1} \cdots F_0\\
       &= \frac{1}{Z} F_{0}^T\cdots F_{n-1}^T F_{n-1} \cdots F_0 .
       \end{split}
       \ee

       Reliability condition {\bf R3}---truthfulness---requires that if the internal and external dynamics coincide, no observer reports to the contrary. 

       So, under a change in perspective there is an ``apparent time reversal'' in that the processes referring to one and the same observer change direction. For instance, when {\sc Bob} plays the role of subject, the dynamics of the experimental system experienced by him, $[S]_B\sim F_\ell$, flows from $x_0$ to $x_n$ (see Fig.~\ref{f:mind-body}B). In contrast, when {\sc Bob} plays the role of object, the physical processes that refers to him are the physical correlates of his experience, $[[S]_B]_A\sim G_\ell $, which flow in the opposite direction (see Fig.~\ref{f:mind-body}C). This ``apparent time-reversal'' is characterized by the transposed operation; so, $G_\ell = F_\ell^T$. Let us summarize this in the following conjecture or observer's property:

       \

       \begin{quote}{\bf O5. Observers' reflexive coupling involves an ``apparent time-reversal'':} 
       Consider two (sets of) of observers, say $A$ and $B$, engaged in a reflexive coupling. That is, $B$ and $A$ mutually observe each other being engaged in an action-perception loop with an experimental system (see Fig.~\ref{f:mind-body}). Assume that $B$ describes the loop associated to $A$ in a clockwise direction, via changes $\Delta P_\ell =\epsilon [J_s, P_\ell]$. Then $A$ describes the loop associated to $B$ in a counter-clockwise direction, via the transposed changes $\Delta P^T_\ell = -\epsilon[J_s, P^T_\ell]$. Alternatively, if $\Delta P^A_\ell = \epsilon [J_s, P^A_\ell]$ are the changes associated to observer $A$, then $\Delta P^B_\ell = -\epsilon [J_s, P^B_\ell]$ are the changes associated to observer $B$. The latter is a less restrictive assumption since the former also requires that $P^B_\ell = (P_\ell^A)^T$. The latter assumption only requires the change $\epsilon\to - \epsilon$ due to the inversion of the temporal labels. Any one of these two assumptions will be enough for our purposes. \end{quote}

       \

       \subsubsection{Reflexive coupling and quantum dynamics with stoquastic Hamiltonians}\label{sm:acoplamiento}
       Here we finally describe the reflexive coupling between two observers, $\textsc{Alice}=(\I{A}, A)$ and $\textsc{Bob}=(\I{B}, B)$. Again, these new names and new notation highlight that we are considering the more symmetric situation wherein each observer plays the two complementary roles as both subject and object. As I mentioned in Sec.~\ref{sm:brief}, the results presented in this section should not be read as a derivation of genuine, real-time quantum dynamics from reflexivity. Rather, the intention of this section is to stress the similarity between Eqs.~\eqref{em:mirror<-} and \eqref{em:mirror->} for mirrors---or Eqs.~\eqref{em:VFR'B} and \eqref{em:VFRA} for video-systems---and the von Neumann equations, Eqs.~\eqref{em:vN-1} and \eqref{em:vN-2}, for $J_a=0$. We here speculate on the possibility of giving this similarity a simple physical meaning, although this is not completely successful. Furthermore, in Sec.~\ref{sm:discussion} we briefly discuss the conceptual plausibility of the reflexive coupling between observers, as illustrated in Fig.~\ref{f:mind-body}, in the light of some developments in the philosophy of mind.

       Intuitively, in line with observer's property {\bf O5}, the dynamics of $\textsc{Bob}$-as-object, as experienced by $\textsc{Alice}$-as-subject, is the transpose of the dynamics of $\textsc{Alice}$-as-object, as experienced by $\textsc{Bob}$-as-subject. The reflexive coupling between \textsc{Alice} and \textsc{Bob} couples these two transposed dynamics via a swap operation analogous to the one that connects Eqs.~\eqref{em:vN-1} and \eqref{em:vN-2}, which are equivalent to von Neumann equation, into Eqs.~\eqref{em:no-vN-1} and \eqref{em:no-vN-2}, which are formally analogous to the equations describing an embodied observer and its transposed. In this way, the reflexive coupling leads to a dynamics that manifests aspects of a genuine, real-time quantum dynamics. However, in this section we are considering the case of symmetric dynamical matrices---i.e., $J_a = 0$ so $J = J_s = J^T_s$. These correspond to so-called {\em stoquastic} Hamiltonians, $H=-\hbar J_s$---i.e., Hamiltonians with real, non-positive off-diagonal entries, which can be naturally interpreted in probabilistic terms. We discuss more general Hamiltonians in Sec.~\ref{sm:generalH} below. We now formalize these ideas.

       Embodiment---observer's property {\bf O1}---entails that an observer-as-object interacting with an experimental system is described by a real probability matrix $P_\ell$ following a dynamics given by Eq.~\eqref{em:O1-embodiment}. However, in our relational approach the dynamical changes in one observer-as-object, say \textsc{Alice}-as-object (or $A$), are relative to, or cognized by, another observer-as-subject, say \textsc{Bob}-as-subject (or $\I{B}$). To make this explicit, in analogy with video-systems, we will denote this as $\left.\Delta^{\rm e} P_\ell^A\right|_B^\to = \epsilon [J_s,P_\ell^A]$ (cf. middle and right hand sides of Eq.~\eqref{em:VFRA}). Here the superindex ``e'' emphasizes that these changes are external to {\textsc{Bob}}---i.e., these are the changes of \textsc{Alice}-as-object described from a 3PP. Furthermore, the arrow $\to$ plays a role analogous to the arrow in Eq.~\eqref{em:mirror->}, i.e., it emphasizes that this equation is one of a pair describing the reflexive coupling between \textsc{Alice} and \textsc{Bob}. In particular, the arrow $\to$ indicates the situation illustrated in Fig.~\ref{f:mind-body}B, where the circular dynamics flows clockwise. A reversed arrow $\leftarrow$ is analogous to the arrow in Eq.~\eqref{em:mirror<-}. That is, it indicates that, like in the case of mirrors (see property {\bf M5}), there is an ``apparent time reversal'' when \textsc{Alice} and \textsc{Bob} exchange perspectives (see observer's property {\bf O5})---in other words, it refers to the situation illustrated in Fig.~\ref{f:mind-body}C, where the dynamics flows counter-clockwise instead. 

       Since \textsc{Bob} is also a physical system, there are physical changes internal to him that correlate with his experience of the changes of \textsc{Alice}-as-object, $\obs[B]{\Delta^{\rm e} P_\ell^A}^\to$, which are external to him. Such internal physical changes should reflect the external ones. So, these changes should also be described in terms of a real probability matrix $P_\ell^B$ that describes \textsc{Bob}'s state at time step $\ell$. After experiencing these changes \textsc{Bob}'s state changes to $P_{\ell+1}^B = P_\ell^B +\obs[B]{\Delta^{\rm e} P_\ell^A}^\to$ (cf. Eq.~\eqref{em:video-system}). Importantly, like in the case of video-systems, the first term in the right hand side of this equation is $P_\ell^B$, not $P_\ell^A$, because \textsc{Bob} remembers his {\em own} previous state not \textsc{Alice}'s. In other words, before \textsc{Bob} experiences the changes $\obs[B]{\Delta^{\rm e} P_\ell^A}^\to$, he is in state $P^B_\ell$. \textsc{Bob}'s experience of the changes of \textsc{Alice}-as-object is supported by physical correlates $\Delta^{\rm i} P_\ell^B = P_{\ell+1}^B - P_\ell^B$---here the superindex ``i'' emphasizes that these changes are internal to \textsc{Bob}-as-subject. These internal changes reflect the external ones, so $\Delta^{\rm i} P_\ell^B= \obs[B]{\Delta^{\rm e} P_\ell^A}^\to$.
       So, in analogy with mirror property {\bf M1}, we introduce the following conjecture or observer's property:

       \begin{quote}
       \noindent {\bf O6. Embodied observers experience objects via physical changes---internal and external:} In between any two time steps, $\ell$ and $\ell+1$, the change in state of an \oo, $A$, interacting with an experimental system, as experienced by an \os,\, $\I{B}$, is given by 
       \be\label{em:boxedPAR}
       \left.{\Delta^{\rm e} P_\ell^A}\right|_B^{\to} 
       =\epsilon [J_s, P_\ell^A].
       \ee
       These changes, which are {\em external} to $\I{B}$ (as emphasized by superindex ``e''), induce changes $\Delta^{\rm i} P_\ell^B = P_{\ell+1}^B - P_\ell^B$, which are {\em internal} to $\I{B}$ (as emphasized by superindex ``i'') and correlate with $\I{B}$'s experience of observing those external changes. That is (cf. \eqref{em:VFRA}) 
       \be\label{em:inPR}
       \Delta^{\rm i} P_\ell^B = \left.{\Delta^{\rm e} P_\ell^A}\right|_B^\to,
       \ee 
       In other words, since observer $\textsc{Bob} = (\I{B}, B)$ is itself a physical system there are physical changes, $\Delta^{\rm i} P^B_\ell$, implicit or unobservable to $\textsc{Bob}$ that allows him to experience the \oo\, $A$. 
       \end{quote}

       Equation~\eqref{em:inPR} ignores the role of \textsc{Alice}-as-subject as well as the role of \textsc{Bob}-as-object. We can take account of \textsc{Alice}'s and \textsc{Bob}'s complementary roles as both subject and object, without introducing any new observers, by implementing a reflexive coupling between them. In this way, \textsc{Alice}'s and \textsc{Bob}'s descriptions are relative to each other, rather than to an unacknowledged external observer. As discussed in Sec.~\ref{sm:others...}, introducing such an external observer as an object would require us to introduce yet another observer that observes the former, which would result in an {\em infinite regress}. 

       To implement the reflexive coupling, we have to consider the case where \textsc{Alice} and \textsc{Bob} play the complementary roles of the subject who observes and the object being observed, respectively. In analogy with Eq.~\eqref{em:inPR}, this yields the equation $\Delta^{\rm i} P_\ell^A = \left.{\Delta^{\rm e} P_\ell^B}\right|_{A}^{\leftarrow}$, with $\Delta^{\rm i} P_\ell^A =P_{\ell +1}^A - P_\ell^A$. This equation is the analogue of Eq.~\eqref{em:mirror<-} for mirrors. Again, the reverse arrow $\leftarrow$ indicates that there is an ``apparent time reversal'' (see observer's property {\bf O5} and Figs.~\ref{f:mind-body}C; cf. the analogous mirror's property {\bf M5}). So, we introduce the following conjecture or observer's property:

       \begin{quote}
       \noindent {\bf O7. Reflexivity:} To describe the world from within, without any reference to external observers, two (sets of) observers must implement a reflexive coupling where they mutually observe each other (see Fig.~\ref{f:mind-body}). In analogy with Eqs.~\eqref{em:mirror<-} and \eqref{em:mirror->}, this is described by the pair of equations: 
       \BE
       \Delta^{\rm i} P^A_\ell &=& \left.{\Delta^{\rm e} P^B_\ell}\right|_A^{\leftarrow},\label{em:finally!B}\\
	       \Delta^{\rm i} P^{B}_\ell &=& \left.{\Delta^{\rm e} P^A_\ell}\right|_{B}^{\to}. \label{em:finally!A}
	       \EE
	       \end{quote}

	       So, the reflexive coupling implements a swap operation analogous to the one that turned Eqs.~\eqref{em:vN-1} and \eqref{em:vN-2}, which are equivalent to von Neumann equation, into Eqs.~\eqref{em:no-vN-1} and \eqref{em:no-vN-2}, which are formally analogous to the equations describing an embodied observer and its transposed---again, in this section we are considering the case of symmetric dynamical matrices, i.e., $J_a = 0$ so $J = J_s = J^T_s$. 

	       We will now explore in what sense Eqs.~\eqref{em:finally!B} and \eqref{em:finally!A} are indeed formally analogous to von Neumann equation. To begin, the right hand side of Eq.~\eqref{em:finally!A}, which is given by Eq.~\eqref{em:boxedPAR}, describes a situation analogous to that illustrated in Fig.~\ref{f:mind-body}B, where \textsc{Alice} and \textsc{Bob}, respectively, play the role of object (like $F$ in Fig.~\ref{f:homunculus}E), and subject (like $W$ in Fig.~\ref{f:homunculus}E). With this convention, the right hand side of Eq.~\eqref{em:finally!B}, $\obs[A]{\Delta P^B_\ell}^{\leftarrow}$, describes the complementary situation shown in Fig.~\ref{f:mind-body}C, where the roles of \textsc{Alice} and \textsc{Bob} are reversed. According to observer's property {\bf O5}, an exchange of roles entails an ``apparent time reversal'', which manifests in the change $\epsilon\to -\epsilon$. That is, 
	       \be\label{em:apparent-time-rev}
	       \obs[A]{\Delta^{\rm e} P^B_\ell}^{\leftarrow} = -\epsilon\left[J_s, P_\ell^B\right].
	       \ee
	       So, Eqs.~\eqref{em:finally!B} and \eqref{em:finally!A} can be written as
	       \BE
	       P^A_{\ell+1} &=& P^A_\ell- \epsilon\left[J_s, P^B\right],\label{em:finally!Aite}\\
		       P^{B}_{\ell + 1} &=& P^B_\ell + \epsilon\left[J_s , P^A_\ell\right].\label{em:finally!Bite}
		       \EE

		       If we can show that $P^A_\ell = (P^B_\ell)^T = P_\ell$, for all $\ell$, then Eqs.~\eqref{em:finally!Aite} and  \eqref{em:finally!Bite} become $\Delta P_\ell = -\epsilon [J_s,P_\ell^T]$ and $\Delta P_\ell^T = \epsilon [J_s,P_\ell]$, respectively, which are equivalent to Eqs.~\eqref{em:vN-1} and \eqref{em:vN-2} (again, with $J_a = 0$), which are in turn formally analogous to von Neumann equation, Eq.~\eqref{em:drhodt=[H,rho]}. So, all that we need to show is that $P^A_\ell = (P^B_\ell)^T$, for all $\ell$.  

		       We will start with the base case, i.e., 
		       \be\label{em:P0P0Tinit}
		       P_0^A = (P_0^B)^T. 
		       \ee
		       Indeed, according to Eqs.~\eqref{em:P_l=FFT}, the initial state of the dynamics of an embodied scientist is given by $P_0 = \widetilde{F}_n\widetilde{F}_n^T = P_0^T$, which is symmetric---here $\widetilde{F}_n = F_0^T\cdots F_{n-1}^T$ (see Eqs.~\eqref{em:Ftilde} and \eqref{em:F=FT2}). So, it seems natural to assume that \textsc{Alice} and \textsc{Bob} have the same initial symmetric state $P_0^A = P_0 = P_0^T = P_0^B$, in which case they are trivially related by a transpose operation. We summarize this in the following conjecture or observer's property:

		       \begin{quote}
{\bf O8. The initial states of two reflexively coupled observers is the same and symmetric:} This is given by $P_0 = \widetilde{F}_n\widetilde{F}_n^T = P_0^T$, which implies that the states of the two observers, \textsc{Alice} and \textsc{Bob}, are trivially related by a transposed operation, i.e., $P_0^A = P_0 = P_0^T = P_0^B$---so, Eq.~\eqref{em:P0P0Tinit} is naturally satisfied. We will discuss the case of more general initial states in Sec.~\ref{sm:genericStates} below.
\end{quote}

We now show that with the base case expressed in Eq.~\eqref{em:P0P0Tinit} we have that $P_\ell^A=(P_\ell^B)^T$, for all $\ell$, and so Eqs.~\eqref{em:finally!B} and \eqref{em:finally!A} are formally analogous to von Neumann equation, Eq.~\eqref{em:drhodt=[H,rho]}, for symmetric initial states and $J_a=0$. We will show this by induction, i.e., by showing that if $P_\ell^A = (P_\ell^B)^T$ at time step $\ell$ then $P_{\ell+1}^A = (P_{\ell+1}^B)^T$ at time step $\ell + 1$. We already have the base case for $\ell=0$, as expressed in Eq.~\eqref{em:P0P0Tinit}. So, let us assume that $P_{\ell}^A = (P_{\ell}^B)^T$ at a generic time step $\ell$. Taking the transpose of Eq.~\eqref{em:finally!Aite} yields $(P^A_{\ell+1})^T = (P^A_\ell)^T+ [J_s, (P^B_\ell)^T]$ since $J_s = J_s^T$ is symmetric. This equation can be written as 
\be\label{em:latter}
(P^A_{\ell+1})^T = P^B_\ell + [J_s, P^A_\ell],
	\ee
	since $(P_\ell^A)^T = P_\ell^B$. Now, the right hand side of Eq.~\eqref{em:latter} equals the right hand side of Eq.~\eqref{em:finally!Bite}. Therefore, the left hand sides of those same equations should be equal, i.e., $P_{\ell+1}^B = (P_{\ell+1}^A)^T$. Thus, Eq.~\eqref{em:finally!Aite} is the transposed of Eq.~\eqref{em:finally!Bite} as we wanted to show. 

	Alternatively, we can obtain the same result if we assume that an ``apparent time reversal'' manifests as a transposed operation as discussed in observer's property {\bf O5}. In this way, Eq.~\eqref{em:apparent-time-rev} can be replaced with the equation 
	\be\label{em:apparent-alt}
	\obs[A]{\Delta^{\rm e} P^B_\ell}^{\leftarrow} = \left(\obs[B]{\Delta^{\rm e} P^A_\ell}^\to\right)^T = -\epsilon\left[J_s, \left(P_\ell^A\right)^T\right].
	\ee
	It can be shown along the same lines that, under the condition expressed in Eq.~\eqref{em:P0P0Tinit}, Eqs.~\eqref{em:apparent-alt} and \eqref{em:finally!A} are also equivalent to von Neumann equation, Eq.~\eqref{em:drhodt=[H,rho]}. 

	So, this suggests that the reflexive coupling between two embodied observers can entail a dynamics that manifests aspects of a genuine, real-time quantum dynamics with an initial state given by a density matrix that is real, symmetric and has non-negative entries. We will discuss in Sec.~\ref{sm:genericStates} below the situation of more general density matrices. 

	However, there is a technical assumption in the reflexive coupling that we now make more explicit. Equations~\eqref{em:finally!Aite} and \eqref{em:finally!Bite} build on Eq.~\eqref{em:O1-embodiment}, describing the dynamics of an observer-as-object from the perspective of an observer-as-subject. However, Eqs.~\eqref{em:finally!Aite} and \eqref{em:finally!Bite} entail a dynamics which is different, in general, to the dynamics entailed by Eq.~\eqref{em:O1-embodiment}. So, at time step $\ell$, the probability matrix $P_\ell$ entailed by Eqs.~\eqref{em:finally!Aite} and \eqref{em:finally!Bite} can in principle be outside the domain of the dynamics described by Eq.~\eqref{em:O1-embodiment}---e.g., $P_\ell$ might have negative off-diagonal entries due to the minus sign associated to the ``apparent time reversal.'' 

	So, we are implicitly assuming that Eq.~\eqref{em:O1-embodiment} connects the dynamics in between any two consecutive time steps, even in such more general cases. 
	That is, we have to relax the constrain that the off-diagonal elements have to be non-negative, while keeping the constraint that the diagonal elements are probabilities. This is somewhat analogous to one of the postulates Feynman used to derive his path-integral formulation of quantum mechanics. That is, that in between two consecutive time steps all paths are characterized by the corresponding {\em classical} action, even though most of those paths may {\em not} be the ones that a classical particle would follow. In our case, in between two consecutive time steps observers are characterized by the dynamics that an observer-as-object would follow, as described by another observer-as-subject, even though this is not the actual dynamics that observers follow after the reflexive coupling is made.
	The fact that, after a reflexive coupling, the diagonal of the probability matrices involved still represent probabilistic information suggests that 
	an integrated approach that couples from the start both embodiment and reflexivity may help better formalize these ideas. 

	%
	\subsection{Quantum dynamics with more general Hamiltonians}\label{sm:generalH}

	Our discussion up to now has been restricted to symmetric dynamical matrices $J_s$ with non-negative off-diagonal entries, which can be interpreted in probabilistic terms via the corresponding factors $F_\ell = \id + \epsilon J_s + O(\epsilon^2)$. So, we have only considered real Hamiltonians, $H_s = -\hobs J_s$, with non-positive off-diagonal entries (see Eq.~\eqref{em:HmapJ})---here $\hobs$ is the analogue of Planck constant. Genuine quantum dynamics does not seem to have this restriction, though, since general Hamiltonians, $H=H_s + H_a/i$, can also have positive and complex entries which, according to Eq.~\eqref{em:HmapJ}, can entail asymmetric dynamical matrices with negative entries, $J = -H_s/\hobs - H_a/\hobs$---here $H_s$ and $H_a$ are symmetric and anti-symmetric operators, respectively. 

	We now discuss how our approach can accommodate this more general situation. We do this from two different perspectives.  First, we discuss a couple of examples to show that effective asymmetric dynamical matrices with off-diagonal negative entries can arise from approximations of systems described by symmetric dynamical matrices with non-negative off-diagonal entries. Second, we motivate the introduction of the last observer's property which enables us to extend our results to dynamical matrices with a non-zero anti-symmetric part. 
	%

	\subsubsection{Effective non-stoquastic Hamiltonians as approximations to stoquastic ones}
	To begin, consider the well-known case of a (two-dimensional) two-level atom which arises after truncating the full (infinite dimensional) model of an atom interacting with a coherent radiation field~\cite{haken2005physics} (see Sec.~15.3 therein; see Appendix~\ref{s:two_atom} herein). The full model consists of an electron, described by the momentum operator $i\hbar\nabla_\bx $, moving in a potential $V (\bx)$ produced by the nucleus---here $\mathbf{x}$ and $\nabla_\bx$ are the three-dimensional position vector and differential operator, respectively. So, the free atom---without the coherent radiation field---is characterized by the Hamiltonian
	\be\label{em:H0atom}
	H_0 = -\frac{\hbar^2}{2m}\nabla^2_\mathbf{x} + V(\mathbf{x}) = \sum_{n=0}^\infty E_n\left|n\rangle\langle n\right|
	\ee
	The expression after the first equality is the three-dimensional version of Eq.~\eqref{em:H}. So, this example can be handled completely with the tools we have developed up to now. The coherent radiation field is characterized by a time-dependent potential energy $U(\mathbf{x}, t) = \mathbf{r}\cdot\mathbf{D_0} \cos(\omega t)$, considered as a perturbation to $H_0$---here $\mathbf{D_0}$ is a suitable constant vector and $\omega$ is the frequency at which the coherent radiation field oscillates. The perturbed model can still be handled completely with the tools we have developed up to now by simply replacing $V(\mathbf{x})$ in Eq.~\eqref{em:H0atom} with $V(\mathbf{x}) +U(\mathbf{x}, t) $. 

	The expression after the second equality in Eq.~\eqref{em:H0atom} is the expansion in the eigenbasis of $H_0$. The two-level atom is obtained by assuming that the coherent radiation field is near resonance with two relevant energy levels---say, $E_0$ and $E_1$. 
	Under this assumption we can keep only two terms in the series after the second equality in Eq.~\eqref{em:H0atom} and handle the coherent radiation field as a perturbation. The two-level atom so obtained is described by an effective Hamiltonian 
	\be\label{em:Heff2atom}
	H_{\rm eff} = \overline{E}\id_2 - C\sigma_Z + D\cos(\omega t)\sigma_X,
	\ee
	where $\overline{E}$, $C$, and $D$ are suitable constants, and $\id_2 = \left|0\rangle\langle 0\right| + \left|1\rangle\langle 1\right|$, $\sigma_Z=\left|0\rangle\langle 0\right| - \left|1\rangle\langle 1\right|$ and $\sigma_X = \left|0\rangle\langle 1\right| + \left|1\rangle\langle 0\right|$, respectively, are the two-dimensional identity matrix, the Pauli matrix in the $Z$ direction, which is diagonal, and the Pauli matrix in the $X$ direction, which has zeros in the diagonal and ones in the off-diagonal. The factor $\cos(\omega t)$ multiplying $\sigma_X$ can have positive and negative values, depending on the value of $t$. So, the effective Hamiltonian described in Eq.~\eqref{em:Heff2atom} can have both negative and positive off-diagonal entries, even though the original Hamiltonian $H_0 + U$ from which it is derived only has negative off-diagonal entries, given by the first term after the first equality in Eq.~\eqref{em:H0atom}.

	Now consider the case of more general complex Hamiltonians, $H_\ell = H_{s,\ell}+H_{a,\ell}/i$, which are associated to asymmetric dynamical matrices $J_\ell = J_{s,\ell} + J_{a,\ell}$ (see Eq.~\eqref{em:HmapJ}).
	There are also examples where these kinds of Hamiltonians with complex entries can be obtained as approximations to real Hamiltonians with non-positive entries. For instance, Vinci and Lidar discuss the case of superconducting flux qubits described by a real Hamiltonian with non-positive off-diagonal entries characterized by a kinetic term of the form $-(E_C/2)\sum_j\partial^2/\partial\phi^2_j$, where $\phi_j$ refer to the magnetic fluxes trapped by the flux qubits, and $E_C$ represents a charging energy~\cite{vinci2017non} (see Eq.~(1) therein). After some approximations Vinci and Lidar obtain an effective Hamiltonian that has a purely imaginary term, even though the original Hamiltonian is real and has non-positive off-diagonal entries~\cite{vinci2017non} (see Eq.~(11) therein). 

	This invites the question of whether general Hamiltonians can be obtained as approximations of real Hamiltonians with non-positive entries, which can be accommodated in our approach. Trying to answer this question might help us better understand what is fundamental in quantum theory and what is simply approximation methods. Anyways, we now discuss the final observer's property, which will enable us to obtain a class of Hamiltonians with complex entries.  

	\subsubsection{Observers only experience relative changes}\label{sm:obs-only-rel}
	Now consider the case of a free particle given by the real Hamiltonian in Eq.~\eqref{em:H} with $V = 0$. This yields the Schr\"odinger equation 
	\be\label{em:diffSchrodinger}
	i\hbar\frac{\partial\psi}{\partial t} = -\frac{\hbar^2}{2m}\frac{\partial^2\psi}{\partial x^2},
	\ee
	whose classical analogue is the diffusion equation with zero drift, i.e., $\partial p/\partial t = D\partial^2p/\partial x^2$. 
	Equation~\eqref{em:diffSchrodinger} acquires an imaginary part if we change to a reference frame moving with velocity $v$~\cite{padmanabhan2011nonrelativistic} (see Eqs.~(14)-(16) therein) 
	\be\label{em:diffSchrodingerDrift}
	i\hbar\frac{\partial\psi}{\partial t} = -\frac{\hbar^2}{2m}\frac{\partial^2\psi}{\partial x^2} + i \hbar v\frac{\partial\psi}{\partial x}.
	\ee
	This is because the time derivative accordingly changes as $\partial/\partial t\to\partial /\partial t + v\partial/\partial x$, and the time derivative is multiplied by $i$.
	The classical analogue of this is a diffusion equation with non-zero drift $\partial p/\partial t = D\partial^2p/\partial x^2 - v\partial p/\partial x$. Of course, the new term in Eq.~\eqref{em:diffSchrodingerDrift} can be canceled by introducing the rule that the wave function also has to be multiplied by a suitable phase under a change of reference frame~\cite{padmanabhan2011nonrelativistic} (see Eqs.~(21) and (22) therein). However, refraining from introducing this rule will help us illustrate the point we want to make. 

	So, the new Hamiltonian in the moving reference frame is $H_\ell = H_{s, \ell}+H_{a,\ell}/i$, with $H_{s,\ell} = -(\hbar^2/2m)\partial^2/\partial x^2$ and $H_{a,\ell} = \hbar v \partial / \partial x$, which according to Eqs.~\eqref{em:vN-1} and \eqref{em:vN-2} yields
	\BE
	\Delta P_\ell - \epsilon [J_{a, \ell} , P_\ell] &=& -\epsilon[J_{s, \ell}, P^T_\ell], \label{em:free1PPa}\\
		\Delta P^T_\ell - \epsilon [J_{a, \ell}, P^T_\ell] &=& \epsilon [J_{s, \ell}, P_\ell],\label{em:free1PPb}
		\EE
		where we have moved the term containing $J_{a,\ell}$ to the left hand side. Here $J_{s, \ell}=-H_{s,\ell}/\hbar$, $J_{a, \ell} = -H_{a,\ell}/\hbar $, and $\rho_\ell = (P_\ell + P_\ell^T)/2 + (P_\ell - P_\ell^T)/2i$.

		The commutators in the left-hand side of Eqs.~\eqref{em:free1PPa} and \eqref{em:free1PPb} here arise from a change of reference frame. The similarity between Eqs.~\eqref{em:free1PPa} and \eqref{em:free1PPb} and the general von Neumann equations, Eqs.~\eqref{em:vN-1} and \eqref{em:vN-2}, suggests that the latter have the structure of a reflexive coupling between two (sets of) observers, combined with something that looks like a reference frame that moves relative to the observers involved in the reflexive coupling. 

		Now, consider two mirrors involved in a reflexive coupling. If in some frame of reference $L$ the two mirrors are moving with the same velocity, $\bv$, they will reflect each other as being static---this is a simple instance of Galilean relativity. So, to determine what the mirrors reflect from the perspective of the reference frame $L$, we first have to remove the motion of the mirrors relative to $L$. The two mirrors cannot reflect the motion associated to $\bv$ because that motion is common to both mirrors in the sense that, from the perspective of $L$, they both move with the same velocity $\bv$. In other words, mirrors involved in a reflexive coupling only reflect changes relative to each other. 

		Similarly, we conjecture that the commutators in the left-hand side of Eqs.~\eqref{em:free1PPa} and \eqref{em:free1PPb} could be considered as removing the changes that are common to the two observers in the reflexive coupling, changes that, as  in the example of the quantum free particle discussed here, could arise due to the motion of the frame of reference. We can capture this intuition in the following conjecture or observer's property: 
		\begin{quote}
{\bf O9. Observers involved in a reflexive coupling can only experience relative changes:}
Two (sets of) observers involved in a reflexive coupling cannot observe the part characterized by $J_{a,\ell}$, since this refers to changes that are common to both observers---these changes could arise due to a change of frame of reference. So, when performing the reflexive coupling associated to a general dynamical matrix $J_\ell= J_{s, \ell} + J_{a, \ell}$, we first have to subtract the anti-symmetric part, precisely as in Eqs.~\eqref{em:free1PPa} and \eqref{em:free1PPb}. 
\end{quote}

Observer's property {\bf O9} leads to the analogue of von Neumann equation with a class of non-stoquastic Hamiltonians with complex entries. 
For instance, consider the real Hamiltonian-like function
\begin{widetext}
\be\label{em:HEM}
\mathcal{H}_{\rm EM}(\bx , \bx^\prime) = \frac{m}{2}\left(\frac{\bx - \bx^\prime}{\epsilon}\right)^2 + V\left(\frac{\bx+\bx^\prime}{2}, t\right) + \frac{e}{c}\left(\frac{\bx-\bx^\prime}{\epsilon}\right)\cdot\mathbf{A}\left(\frac{\bx+\bx^\prime}{2}, t\right) +\frac{e^2}{m c^2} \left[\mathbf{A}\left(\frac{\bx+\bx^\prime}{2}, t\right)\right]^2 . 
\ee
\end{widetext}
This leads to factors with an anti-symmetric component due to the third term in the right hand side, which is linear in $\bx-\bx^\prime$. Dealing with this anti-symmetric component according to observer's property {\bf O9} leads to the Schr\"odinger equation of a quantum particle in an electromagnetic field $\mathbf{A}$, i.e., 

     \be\label{em:SchEM}
     \begin{split}
     i\hbar\frac{\partial \psi(\bx,t)}{\partial t}=& -\frac{\hbar^2}{2m}\left(\nabla - i \frac{e}{\hbar c}\mathbf{A}\right)^2\psi(\bx,t)\\
	     &+ e V(\bx,t)\psi(\bx,t).
	     \end{split}
	     \ee
	     We show this in Appendix~\ref{s:EM}.

	     \

	     \subsection{More general initial quantum states and observables}\label{sm:genericStates}

	     Here we show how the previous results can be extended to more general initial density matrices and observables. In our approach the initial density matrix is given by $\rho_0 = P_0= \widetilde{F}_n\widetilde{F}_n^T$, which is real and so symmetric. This is not necessarily a restriction, though. More general ``initial'' quantum states $\rho_{\rm prep}= U_{\rm prep}\rho_0 U_{\rm prep}^\dag$ can be prepared after applying a suitable quantum operation $U_{\rm prep}$ to $\rho_0$. In principle, we can write $U_{\rm prep} = U_m\cdots U_0$, for a suitable number of time steps $m$, where each $U_\ell = \id - i\epsilon H_\ell/\hobs$ is obtained from a factor $F_\ell = \id + \epsilon J_\ell$ via $H_\ell = -\hobs (J_{s, \ell} + J_{a,\ell}/i)$---here $\ell=0,\dotsc , m$.

	     Furthermore, we have mostly focused on one observable, i.e., position. However, according to Feynmann ``all measurements of quantum-mechanical systems could be made to reduce eventually to position and time measurements (e.g., the position of the needle on a meter or the time of flight of a particle). Because of this possibility a theory formulated in terms of position measurements is complete enough in principle to describe all phenomena''~\cite{feynman2010quantum} (p. 96).

	     Indeed, this view aligns with our focus on modeling how science is actually performed {\em in practice} and how concepts are constructed out of this. For instance, the concept of momentum $p=-i\hbar\partial/\partial x$ of a particle is usually taken as existing in an abstract space. However, to actually measure momentum in practice we need to make the focus system interact with another system that serves as a measuring device. The measuring device can be the position $X$ of a probe particle that interacts with the focus system. 

	     The concept of momentum can emerge out of the description of what actually happens in practice while measuring it. Indeed, let the initial state of the probe be a Gaussian, $\psi_{\rm dev}(X)\propto e^{-X^2/4\sigma^2}$, centered around zero, with $0<\sigma\ll 1$ characterizing the initial uncertainty on the probe's position. Since $e^{i k x}$ are eigenfunctions of the momentum operator, with eigenvalues $\hbar k$, it is convenient to write the initial state of the system as $\psi_{\rm sys}(x)\propto \int c_k e^{i k x}\mathrm{d} k$, where $c_k$ are suitable coefficients. After a suitable interaction between the two particles, we can obtain the joint state $\Psi(x,X)\propto \int c_k e^{i k x} e^{-(X- g_0\hbar k)^2/4\sigma^2}\mathrm{d}x$, where $g_0$ is a constant, and the corresponding joint probability $\mcP(x, X) = |\Psi(x,X)|^2$~\cite{svensson2013pedagogical} (see Sec.~4.1 therein; see also Appendix~\ref{a:momentum} herein). 

	     However, in practice we are {\em not} interested in observing $x$, but on {\em inferring} the momentum of the system by observing {\em only} the probe. Marginalizing $x$ yields $\mcP_{\rm dev}(X) \approx\int {|c_k|^2} \delta(X- g_0 \hbar k )\mathrm{d}k$ for $\sigma\ll 1$, where $\delta$ is the Dirac delta function. Thus with probability $\propto\, |c_k|^2$ the position of the probe, $X= g_0\hbar k$, is proportional to the momentum eigenstate $\hbar k$. This is the (projective) measurement postulate of quantum theory. 

	     \

	     \section{Discussion}\label{sm:discussion}
	     %
	     As we have mentioned in Sec.~\ref{sm:brief}, this work is to be read as a set of conjectures informed by the philosophy of mind and a reverse-engineering of science and quantum physics---see Sec.~5 in Ref.~\cite{realpe2} for a brief conceptual presentation of the main ideas involved. Here we summarize our work up to now and place it in the landscape of the philosophy of mind. This can be interpreted as a potential answer to the question posed in the title of this work. Of course, such an answer will not be conclusive as we have only presented a set of conjectures. Importantly, what we conjecture is not that the mind somehow ``emerges'' out of quantum physics, but the other way around. We ask: could quantum physics emerge from modelling classical observers, with both an objective and a subjective aspect, interacting with classical experimental systems?

	     \subsection{Embodiment and imaginary-time quantum dynamics}
	     Our results on embodiment suggest that scientists do {\em not} actually escape the action-perception loop but only {\em appear} to do so by describing the world with some aspects of quantum theory. More precisely, the circular causality associated to embodiment, as traditionally understood, can manifest aspects of imaginary-time quantum dynamics (see Sec.~\ref{sm:third}), which already manifests genuinely quantum phenomena, like (constructive) interference (see Sec.~\ref{sm:two_slit}). This gives the impression that scientists can actually obtain an observer-independent description of the world---a world that just happens to manifests these quantum aspects for no reason. However, our results suggest that there is indeed a reason for why scientists would have to use these quantum-like tools: to effectively account for the observer and the experimental context. 

	     As we said, imaginary-time quantum dynamics already displays some quantum-like features~\cite{Zambrini-1987}. For instance, using our framework we have shown that a classical two-slit experiment can entail constructive interference (see Sec.~\ref{sm:two_slit}). This provides a fresh perspective to think about quantum interference. When an embodied observer has information about which slit the particle goes through---e.g., when only one slit is open---this has to be reflected in the physical correlates of the experiment ``internal'' to her. So, the variable $x_1$ describing the slit, which is external to the observer, and the variable $x_3$ associated to the corresponding physical correlate, which is internal to the observer, must be equal, i.e., $x_3 = x_1$ (see Fig.~\ref{f:slits}). In this view, the imaginary-time version of quantum interference arises because, when
	     an embodied observer cannot access any information about which slit the particle goes through, the values of $x_1$ and $x_3$ do not have to coincide even though they refer to the same ``thing'' (i.e., the slits).

	     More technically, we have shown that an embodied scientist performing an experiment can be described in terms of a real probability matrix that satisfies an imaginary-time von Neumann equation. Alternatively, using the cavity method of statistical mechanics, we have shown that, when the experimental system is initially located at a given position $x_0=x_0^\ast$, which is an instance of a pure state, embodied scientists can describe the world in terms of forward and backward cavity messages, $\mu_{\to}^\ast$ and $\mu_{\leftarrow}^\ast$. These forward and backward cavity messages are formally analogous to the imaginary-time wave function and its conjugate, respectively, and follow a belief propagation dynamics described by the imaginary-time Schr\"odinger equation and its conjugate. The imaginary-time analogue of the Born rule for the probability $p(x, t)$ for being at location $x$ at time $t$ is the standard rule of the cavity method: $p (x,t) = \mu_{\to}^\ast(x,t)\mu_\leftarrow^\ast(x,t)$.

	     However, according to Eq.~\eqref{em:matrix}, given a set of factors, $F_0, \dotsc , F_{k-1}$, characterizing the internal and external dynamics of an embodied scientist performing an experiment, the initial state is restricted to $P_0=F_{k-1}\cdots F_0$. Although there is some freedom to choose the initial state if some of the factors are considered as preparing the initial probabilistic state---say $F_\ell$ for $\ell = 0,\dotsc ,\ell_{\rm prep}$, it is not clear that this freedom is enough to cover all possible initial states. This situation is not alien to genuine, real-time quantum mechanics, though. Indeed, Aharonov {\em et al.}~\cite{aharonov2023conservation} argued recently that not any quantum state can be prepared in nature.

	     \subsection{Reflexivity and real-time quantum dynamics}
	     Why imaginary-time and not genuine, real-time quantum dynamics? We discussed this extensively based on ideas of reflexive systems (see Sec.~\ref{sm:first}). We introduced some conjectures characterizing an observer, and we showed that these lead to a dynamics with aspects of a genuine, real-time quantum dynamics---we refer to these conjectures collectively as the {\em reflexive coupling hypothesis}. 

	     However, Sec.~\ref{sm:first} should not be considered as a rigorous derivation of real-time quantum dynamics from reflexivity. The reason is that, reflexivity being a rather subtle and scarcely studied subject, the connection between the ideas of reflexivity and the conjectures we introduce may not be completely clear. We introduce this section in this work because we find it conceptually plausible and we hope it can suggest future research on potential connections between the philosophy of mind and quantum physics. The literature on reflexive systems is rather scarce and we hope that an interdisciplinary approach to this topic could help further clarify or improve the conjectures we introduce here. So, with this in mind, let us summarize our work and discuss its potential connections with the philosophy of mind.

	     Our approach leads us to conjecture that the core message of quantum theory might be encoded in two core principles: (i) Observers are embodied; this is associated to observer's property {\bf O1}. (ii) The world must be described from within---that is, without any reference to external observers; this is effectively associated to observer's properties {\bf O2}-{\bf O9}. Principles (i) and (ii) entail embodiment and reflexivity, respectively. Here embodiment refers to the idea that when scientists interact with an experimental system they are involved in an instrumentally-mediated action-perception loop (see Figs.~\ref{f:effective}B, \ref{f:circular_noRQM}A). Reflexivity refers to the idea that for scientists to describe the world without any reference to external observers, two (sets of) observers should mutually describe each other (see Fig.~\ref{f:mind-body})---in this sense, observers are relative to each other rather than to an external, unacknowledged observer. Embodiment and reflexivity are rather generic concepts, and our model is rather minimal. So, our approach should not be restricted to a specific kind of observer. 

	     We assume that a scientist, say Fabienne, doing an experiment does not exist in an absolute sense, but she is relative to another external observer, say Wigner. In this case Fabienne plays the role of an object being observed by Wigner---this additional observer is typically neglected in cognitive science. Now, to be consistent, we should also take into account the observer that observes Wigner. Otherwise, Wigner would exist in an absolute, non-relational way. So, who does observe Wigner? If we added another observer, say Rovelli, so that Wigner becomes an object of observation for Rovelli, we would be headed into an {\it infinite regress}. Indeed, we can now ask: who does observe Rovelli? And so on (see Figs.~\ref{f:homunculus}A-D). 

	     We escaped this infinite regress in two steps. First, we assumed that observers play two complementary roles: as objects being observed by other observers and as ``subjects'' that can observe other objects, including {\it other} observers (see Fig.~\ref{f:homunculus}E). Here ``subject'' is used in the strict technical sense of the opposite of object, or not-an-object. For instance, in the example above Fabienne plays the role of an object for Wigner, but Wigner plays the role of a subject for himself in the sense that Wigner cannot completely become an object of observation for himself (see Fig.~\ref{f:homunculus}E). This is analogous to a mirror that cannot directly reflect itself---though it can do so indirectly with the help of another mirror. Acknowledging that Wigner cannot become an object of observation for himself stops the chain of infinite regress since we do not need to keep on adding observers-as-objects {\it ad infinitum}. We acknowledge that there is something that cannot become an object of observation, a ``subject.'' This reminds us of Thompson's words~\cite{thompson2014waking} (p. 100):
		     \begin{quote}
		     ``Consciousness is our way of being, and it cannot be objectified, that is, treated as just another kind of object out there in the world, because it is that by which any object shows up for us at all.''
		     \end{quote}

		     This may be a subtle point, so let us discuss it a bit more. Please imagine, for instance, that you observe a physical system $S$. Besides $S$, which is external to you, there are physical processes inside you that allow you to experience $S$, e.g., the neural correlates associated to your experience of $S$---denoted here as $\{S\}$. For $S$ and $\{S\}$ to be correlated, there has to be a physical interaction between them. Such a physical interaction may or may not be negligible; however, we do not want to {\it a priori} neglect it, but to determine {\it a posteriori} whether we can do so. While $\{S\}$ is key in allowing you to experience $S$, it is absent for you in the sense that you cannot simultaneously and directly experience both $S$ and $\{S\}$ as objects of observation. If you simultaneously and directly observe both $S$ and $\{S\}$ you are observing a different physical object, i.e., $S+\{S\}$, with associated physical correlates $\{S+\{S\}\}$, which you cannot directly observe at the same time that you observe $S+\{S\}$. Again, this is the analogue of a mirror that cannot directly reflect itself.

		     Now, when you play the role of a subject that observes a system $S$, you describe $S$ from your own first-person perspective (1PP)---you can refer to this as ``I observe $S$.'' In contrast, when somebody else observes you, while you observe $S$, you play the role of an object, you are being described from a third-person perspective (3PP)---the other observer can refer to you as ``he observes $S$.'' In analogy with this, we denoted an observer-as-subject, say Wigner ($W$), as $\I{W}$. Here the symbol $I_W$ refers to Wigner's 1PP, while the cancellation refers to its ``absential nature'' to use an expression by Deacon~\cite{deacon2011incomplete} (ch. 0)---this cancellation is inspired in Heidegger {\it sous erasure}. Here $\I{W}$ is similar to the number zero in that it is a placeholder to indicate the absence of something.

		     Due to this ``absential nature'' we cannot directly access the degrees of freedom of the physical correlates of $\I{W}$. So, we model such physical correlates indirectly in the same way that we model the inaccessible degrees of freedom of a thermal bath, i.e., as fluctuations or noise characterized by a temperature- or diffusion-like constant, $\hobs$, which in this approach plays a role similar to that of Planck's constant, $\hbar$. Such fluctuations are all-pervasive and irreducible since $\I{W}$ is always present whenever an object is being observed and cannot be reduced to an object of observation like, e.g., the atoms of a thermal bath. These properties are analogous to the properties of actual quantum fluctuations, which are also all-pervasive and irreducible. If it turns out that $\hobs = \hbar$, the bar in $\hbar$ could remind us of the ``absential nature'' of observers-as-subjects.

		     Here we built on an {\em analogy} with Thompson's neurophenomenological perspective~\cite{thompson2014waking} which distinguish between contents of experience and experience as such---i.e., the mere capacity to experience any content at all. For Thompson, while contents of experience are objects of observation, experience as such cannot become an object of observation because it is the very precondition to experience any object at all. In line with this, we distinguish between the physical correlates $[S]$, that allows a subject, say Wigner, to discriminate between $S$ and any other system $S^\prime\neq S$, and the physical correlates associated to the mere capacity to experience any system at all, which we refer to as $[\I{W}]$ (see below). So, here $\{S \} = [S]+[\I{W}]$. In principle, $[S]$ can be {\it inferred} from $S$ because they both refer to the same object. In contrast $[\I{W}]$ refers to the mere capacity to experience from a 1PP, not to any specific system. In line with Thompson, we assumed $[\I{W}]$ to be fundamentally inaccessible to Wigner, and thus we modeled it as noise.

		     Anyways, coming back to our previous discussion of Fabienne and Wigner, assuming that Wigner plays the role of a subject and Fabienne that of an object allows us to stop the infinite regress. However, the description is still incomplete because we are neglecting the role of Wigner as object and that of Fabienne as subject. This brings us to the second step to escape the infinite regress. We have to implement a reflexive coupling between Fabienne and Wigner where they play both roles as the subjects that observe each other from a 1PP, and the objects being observed by each other from a 3PP. So, while embodiment implements a kind of circular relationalism between observers-as-object and experimental system, reflexivity implements a kind of circular relationalism between subject and object, between the 1PP and the 3PP (see Fig.~\ref{f:mind-body}). We could summarize this by saying that every experience has a physical correlate and that every physical phenomenon is an experience for someone---this does not imply that rocks have experiences, but that rocks are rocks for someone who experience them as such. This is {\em analogous} to what is sometimes referred to in the literature as ``subject-object non-duality''~\cite{thompson2014waking,tang2015neuroscience}---i.e., the interdependence between subject and object. 

		     This perspective is analogous to a view emerging in the philosophy of mind. Indeed, Fig.~\ref{f:mind-body} resonates with Sharf's words~\cite{deguchi2021can} (p. 156 and 160):

			     \begin{quote}
			     ``Philosophers sometimes speak of the two opposing positions as the first-person or subjective point of view, and the third-person or objective point of view [...] [T]he subjective and objective perspectives constitute two poles of an antinomy; they are not merely interdependent but also subsume and enfold one another [...] The world is within me, and I am within the world. And it is impossible to specify where one perspective ends and the other begins; they fold back upon one another seamlessly, like the two sides of a M\"obious strip.''
			     \end{quote}

			     Or Merleau-Ponty's~\cite{merleau1962phenomenology} (p. 430):

				     \begin{quote}
				     ``The world is inseparable from the subject, but from a subject which is nothing but a project of the world, and the subject is inseparable from the world, but from a world which the subject itself projects.''
				     \end{quote}

				     Our work aligns with growing evidence suggesting that in quantum theory facts are relative~\cite{brukner2020facts}. Indeed, our approach shares some elements with the main interpretations of quantum theory that endorse relative facts~\cite{Rovelli-1996,mermin2014physics,debrota2018faqbism,fuchs2013quantum,pienaar2021qbism,brukner2017quantum}: Like QBism, our approach explicitly acknowledges the role of scientists in science~\cite{mermin2014physics,fuchs2013quantum,debrota2018faqbism,pienaar2021qbism}.
				     Like the (neo-)Copenhagen interpretation of quantum theory~\cite{brukner2017quantum}, our approach explicitly acknowledges the experimental context. These elements already appear when we take account of an embodied scientist, i.e. they do not need reflexivity (see Fig.~\ref{f:circular_noRQM}). 
				     Like relational quantum mechanics (RQM)~\cite{Rovelli-1996,rovelli2021helgoland}, our approach assumes that the relation ``an observer observes a phenomenon'' is itself relative to another observer. This element appears when dealing with reflexivity (see Fig.~\ref{f:homunculus}).

				     However, there are important differences too. To begin, QBism treats scientists as rather abstract agents immersed in a publicly shared physical universe. The quantum formalism is seen as a normative criterion~\cite{debrota2018faqbism} (see Sec. 18 therein) setting ``the standard to which agents should strive to hold their expectations''. Such agents are betting, implicitly or explicitly, on their subsequent experiences, based on earlier ones, and the quantum formalism is a tool to help them place better bets~\cite{mermin2018making} (p. 8). Even though agents are acknowledged to be physical systems and also part of the world, QBism's emphasis on agent's subjective beliefs seems to move it a bit away from the objective side of things and towards the subjective. As DeBrota and Stacey said recently ``subjective judgments [...] comprise much of the quantum machinery''~\cite{debrota2018faqbism}. There is no mention of the quantum formalism emerging out of the dynamics of agent and world, as we suggest here. 

				     On the other hand, regarding RQM, unlike us, Rovelli treats observers as generic quantum systems---that is, as far as physics is concerned, for Rovelli there are no relevant differences between an electron and an observer. Furthermore, for Rovelli to ``think that a human being, their mind [...] plays any special role in the grammar of nature is nonsense''~\cite{rovelli2021helgoland} (p.~140).

				     Our approach suggests a middle way between QBism and RQM since we treat observers as physical systems with a dynamical role to play, but not as any kind of physical system. Instead of working with an {\em abstract} notion of what physicists might {\em assume} an observer is, however, our approach attempts to build on general insights gained by the areas of science dedicated to investigate {\em actual} observers. Moreover, we do not treat observers as quantum systems, but as {\em classical} cognitive systems. Importantly, non-trivial aspects of the quantum formalism seem to emerge out of two key concepts: embodiment and reflexivity. To clarify, observers can be described by {\em other} observers as quantum systems since they are physical systems too. However, our approach suggests that this is due to the relationalism between {\em classical} observers and classical experimental systems, which manifests as embodiment and reflexivity. 

				     In our approach the kind of relationalism associated to quantum theory seem to be the {\it circular} co-dependence between perceiver and world, which, according to Varela {\it et al.}~\cite{varela2017embodied} (p. 172), allows embodied cognition to find a middle way between materialism and idealism. Indeed, in this approach neither the perceiver nor the world are primary, but they depend on each other like ``two sheaves of reeds propping each other up.'' More precisely, the kind of relationalism associated to our approach suggests a double circularity, one associated to embodiment and another to reflexivity---a kind of ``strange loop''~\cite{hofstadter2013strange}. In contrast, Rovelli conceives relationalism in terms of a generic network of relations, not in terms of circularity. A network of relations that does not seem to include the mind, as if the mind were an island independent of the physical world. This contrasts with the Buddhist tradition, a tradition that Rovelli tries to relate to his relational interpretation~\cite{rovelli2021helgoland} (ch.~5)---more specifically, the particular tradition known as the Middle Way or Madhyamaka~\cite{garfield1995fundamental,westerhoff2009nagarjuna}. However, the tradition considered to be the best interpretation of Madhyamaka is Madhyamaka-Pr\a=asa\.ngika, which includes the mind in the network of relations. Indeed, according to Westerhoff, this tradition maintains that~\cite{westerhoff2024candrakirti} (p. 114):

					     \begin{quote}
					     ``[T]he mind is part of the network of dependent origination like everything else, [hence] its existence is thereby regarded as only dependent, but cannot be fundamental.''
					     \end{quote}

					     Here ``dependent origination'' refers, na\"ively, to ``relations'' (see, e.g., Ref.~\cite{garfield2014engaging}, p. 25-36, for a more precise description of this expression). Rovelli seems to see his relational interpretation mostly in terms of causal relations. However, according to Westerhoff, the Madhyamaka-Pr\a=asa\.ngika tradition maintains that~\cite{westerhoff2009nagarjuna} (p.~124):

					     \begin{quote}
					     ``[T]he causal relation does not exist from its own side, is conceptually constructed... it follows that each material object must be conceptually constructed.''
					     \end{quote}

					     So, for the Madhyamaka-Pr\a=asa\.ngika tradition neither the mind nor the material world is fundamental, as they depend on each other; there is no ground to which we can grasp.

					     The prospects of a science that depends on the observer may be seen as something negative. However, we believe this possibility could also be seen in a positive light. To begin, if taking explicit account of the observer indeed happens to entail the formalism of quantum theory, this means that an observer-dependent science does not have to violate current scientific knowledge.
					     It would only violate our {\em assumption} that we have the special status of understanding the world from a disembodied perspective, independent of our capacity to experience it, as if we were not part of the world.  Indeed, we might learn something useful by rigorously investigating whether indeed we have this special status in the same way that our understanding of the universe advanced when we doubted our special status of being at the center of it. 

					     Furthermore, the prospects of an observer-dependent science suggest a potential relationship between quantum physics and the areas of science that rigorously investigate observers, such as cognitive science and neurophenomenology. In particular, if the quantum formalism---the foundation on which the skyscraper of science stands---already integrates subject and object, the 1PP and the 3PP, the question of how subjective experience ``emerges'' out of physics might become more tractable. What would emerge is not experience as such but increasingly complex contents of experience. 
					     This would parallel the emergence of increasingly complex physical phenomena
					     from the ``basic constituents of matter.'' 

					     Importantly, this also suggests new kinds of experiments, where the objective and subjective aspects of the observer can be part of the experimental setup. For instance, rigorous mind training techniques may allow scientists to directly experience the relational nature of the world or the random fluctuations that we associate to the observer-as-subject. There have been reports about this in the Buddhist tradition even before the advent of science as we know it: the former seems related to the experience of emptiness (see, e.g., Ref.~\cite{garfield2014engaging}, ch.~3, Ref.~\cite{bitbol2019two} or Ref.~\cite{rovelli2021helgoland}, ch.~5) and the latter seems related to the experience of a ``subtle energy'' associated to what is considered the most fundamental aspect of consciousness, i.e., ``pure awareness.'' On the latter Thompson says~\cite{thompson2014waking} (p. 342-344):

						     \begin{quote}
						     ``The Dalai Lama [said] that the physical basis for pure awareness is a subtle energy whose presence can be felt in the body. This energy [...] is said to carry all excitation and movement, including at the level of cells. [...] The Dalai Lama suggested that the scientific concept of matter may need to be modified in order to appreciate this energy.'' 
						     \end{quote}

						     Notice that, in our approach, there is a kind of modification to the concept of matter when dealing with the physics of the observer-as-subject. Indeed, we cannot model it in terms of objects following cause-effect mechanisms, as traditionally done in science. Rather, as we said, we have to model it as an all-pervasive noise irreducible to lower-level mechanisms, which looks similar to the irreducible, all-pervasive fluctuations in quantum physics.

						     \

						     \noindent {\bf Data availability:} Data sharing not applicable to this article as no datasets were generated or analysed during the current study.

						     \acknowledgements

						     I thank Marcela Certuche, Harold Certuche and Blanca Dominguez for their support, financial and otherwise, during a substantial part of this project. I thank Tobias Galla, Alan J. McKane, and the University of Manchester for their support at the beginning of this project. I thank Guen Kelsang Sangton for insightful discussions on Buddhist philosophy. 
						     I thank Marcin Dziubi\'nski for his brief but useful lessons on recursion and self-reference. I thank Shailesh Date and Jose Jaramillo for useful comments. I thank Max Velmans, Michel Bitbol, Jerome Busemeyer, Diana Chapman Walsh, Alejandro Perdomo-Ortiz, Addishiwot Woldesenbet Girma, Delfina Garc\'ia Pintos, Marcello Benedetti, Kenneth Augustyn, John Myers, Marcus Appleby, Nathan Killoran, Markus M\"uller, Michael R. Sheehy, Nathan Berkovitz, Eduardo Pont\'on, Roberto Kraenkel, Camila Sardeto Deolindo, Cerys Tramontini, Hernan Ocampo, Oscar Bedoya, Gonzalo Ordo\~nez, Maria Schuld and Robinson F. Alvarez for comments and constructive criticism. I thank Christopher A. Fuchs for clarifying comments regarding QBism. I thank Mariela G\'omez Ram\'irez and Nelson Jaramillo G\'omez for bringing my attention to these ideas. This research is funded in part by the Gordon and Betty Moore Foundation (Grant GBMF7617) and by the John Templeton Foundation as part of the Boundaries of Life Initiative (Grant 60973). I thank FAPESP grant 2016/01343-7 for funding my visit to ICTP-SAIFR from 20-27 January 2019 where part of this work was done.

						     \

						     \appendix

						     \appendix

						     \section{Particle in an electromagnetic field via real non-negative kernels}\label{s:EM}

						     Here we discuss the case of a quantum particle in a classical electromagnetic field, which is associated to a complex (and so non-stoquastic) Hamiltonian operator. We show that this can also be written in terms of non-negative real kernels. This adds further evidence that the non-negativity of the factors in our approach does not necessarily restrict it to stoquastic Hamiltonians. 

						     The Schr\"odinger equation of a particle of charge $e$ interacting with an electromagnetic field can be written as 
						     \be\label{e:Main_SchrodingerEM}
						     \begin{split}
						     i\hbar\frac{\partial \psi(\bx,t)}{\partial t}=& -\frac{\hbar^2}{2m}\left(\nabla - i \frac{e}{\hbar c}\mathbf{A}\right)^2\psi(\bx,t)\\
							     &+ e V(\bx,t)\psi(\bx,t),\\
							     \end{split}
							     \ee
							     where $\bx$ denotes the position vector in three dimensional space, while $V$ and $\mathbf{A}$ denote the scalar and vector fields respectively. Notice that the Hamiltonian associated to Eq.~\eqref{e:Main_SchrodingerEM} now contains an imaginary part given by the terms linear in $\mathbf{A}$ arising from the expansion of ${(\nabla - i e\mathbf{A} /\hbar c)^2\psi(\bx,t)}$.

							     We will show in a series of three theorems and a corollary that Eq.~\eqref{e:Main_SchrodingerEM} can be written as a pair of equations analogous to Eqs.~\eqref{em:vN-1} and \eqref{em:vN-2} in the main text with a non-negative real kernel. As we discussed in the main text, these pair of equations are equivalent to von Neumann equation.

							     In the first theorem and corollary, we show that Eq.~\eqref{e:Main_SchrodingerEM} and the corresponding von Neumann equation can be written in terms of convolutions with a complex-valued kernel $\mathcal{C}\propto e^{-\epsilon\widetilde{\mathcal{H}}_{EM}/\hbar}$, where $\widetilde{\mathcal{H}}_{EM}$ is a complex-valued Hamiltonian-like function. Afterwards, in the second theorem, we show that $\mathcal{C}$ can be replaced by a real-valued kernel $\mathcal{W}\propto e^{-\epsilon\mathcal{Q}_{EM}/\hbar}$ in the limit when $\epsilon\to 0$. However, $\mathcal{Q}_{EM}$ depends on $\hbar$, unlike the Hamiltonian-like function in Eq.~\eqref{em:non-relativisticH} in the main text, which is independent of $\hbar$. In the last theorem we show that in the limit $\epsilon\to 0$ it is possible to replace $\mathcal{W}$ by another real-valued kernel $\mathcal{K}\propto e^{-\epsilon\mathcal{H}_{EM}/\hbar}$, where $\mathcal{H}_{EM}$ is independent of $\hbar$. In this way we show that the von Neumann equation associated to Eq.~\eqref{e:Main_SchrodingerEM} can be written as a pair of real-valued matrix equations, like Eqs.~\eqref{em:vN-1} and \eqref{em:vN-2} in the main text, in terms of $\mathcal{K}$, which is real-valued and non-negative. 

							     We begin by showing that Eq.~\eqref{e:Main_SchrodingerEM} can be written in terms of a convolution with a complex-valued kernel in the following

							     \begin{theorem}
							     The Schr\"odinger equation for a charged particle in an electromagnetic field, Eq.~\eqref{e:Main_SchrodingerEM}, is equivalent to
							     \be\label{e:partial}
							     \epsilon\frac{\partial\psi}{\partial t} = i [\mathcal{C}\ast\psi - \psi],
							     \ee
							     in the limit $\epsilon\to 0$. Here 
							     \be\label{e:C*psi=}
							     [\mathcal{C}\ast\psi](\bx) = \int \mathcal{C}(\bx-\bx^\prime)\psi(\bx^\prime)\mathrm{d}^3\bx
							     \ee
							     denotes the convolution between the wave function $\psi$ and the kernel
							     \be\label{e:C}
							     \mathcal{C}(\bx , \bx^\prime) = \frac{1}{\mathcal{Z}_{EM}}\exp{\left[-\frac{\epsilon}{\hbar}\widetilde{\mathcal{H}}_{EM}(\bx, \bx^\prime) \right]} ,
							     \ee
							     where
							     %
							     \be\label{e:classicalH_EM}
							     \begin{split}
							     \widetilde{\mathcal{H}}_{EM}(\bx, \bx^\prime) =& \frac{m}{2}\left(\frac{\bx - \bx^\prime}{\epsilon}\right)^2 + V\left(\frac{\bx+\bx^\prime}{2}, t\right)\\
								     &- i\frac{e}{c}\left(\frac{\bx-\bx^\prime}{\epsilon}\right)\cdot\mathbf{A}\left(\frac{\bx+\bx^\prime}{2}, t\right),
							     \end{split}
							     \ee
							     %
							     and $\mathcal{Z}_{EM} = (2\pi\hbar\epsilon/m)^{3/2}$ is a normalization constant.
							     \end{theorem}
							     \begin{proof}
							     We have to show that Eq.~\eqref{e:partial} is equivalent to Eq.~\eqref{e:Main_SchrodingerEM} in the limit $\epsilon\to 0$. To do so, notice that the Gaussian factor in the complex kernel $\mathcal{C}$ defined in Eq.~\eqref{e:C} associated to the kinetic term in Eq.~\eqref{e:classicalH_EM} has a variance proportional to $\epsilon$, which allows us to expand the other factors in the integral in Eq.~\eqref{e:C*psi=} around $\bx$ up to second order in $|\bx - \bx^\prime|$ or to first order in $\epsilon$, since $\epsilon\to 0$. More precisely, by introducing the variable ${\bu = \bx-\bx^\prime}$, so ${(\bx+\bx^\prime)/2 = \bx - \bu/2}$ as well as ${\bx^\prime = \bx - \bu}$, we can write

							     \begin{widetext}
							     \be\label{e:C*psi}
							     [\mathcal{C}\ast\psi](\bx , t) =\bra f(\bx, \bu, t)\left[\psi(\bx ,t) - \bu\cdot\nabla\psi(\bx ,t) + \frac{1}{2}\bu\cdot \mathbf{H}\psi(\bx, t)\cdot\bu \right]\ket_\bu + O(\epsilon^2),
							     \ee
							     \end{widetext}
							     where 
							     \be\label{e:Gauss...}
							     \bra\cdots\ket_\bu = \frac{1}{|\mathcal{Z}_{EM}|}\int \exp\left(-\frac{m \bu^2}{2 \hbar\epsilon } \right)(\cdots),
							     \ee
							     denotes the Gaussian average associated to the kinetic term in Eq.~\eqref{e:classicalH_EM}, $\mathbf{H}\psi$ stands for the Hessian or matrix of second derivatives of $\psi$. Furthermore, the function
							     \begin{widetext}
							     \be\label{e:f=}
							     \begin{split}
							     f(\bx, \bu, t) &= \exp\left[-\frac{\epsilon}{\hbar} V\left(\bx-\bu/2, t\right) + i\frac{\epsilon e}{\hbar c} \frac{\bu}{\epsilon}\cdot\mathbf{A}\left(\bx-\bu/2, t\right)\right]\\
									       &= 1-\frac{\epsilon}{\hbar}V(\bx, t) + i\frac{e}{\hbar c}\bu\cdot\mathbf{A}(\bx,t) - i\frac{e}{2\hbar c}\bu\cdot\nabla\mathbf{A}(\bx,t)\cdot\bu -\frac{1}{2}\left[\frac{e}{\hbar c}\bu\cdot\mathbf{A}(\bx, t)\right]^2  + O(\epsilon^2, \epsilon |\bu| , |\bu|^3),
							     \end{split}
							     \ee
							     \end{widetext}
							     gathers all the interaction terms in $\mathcal{C}$, i.e., those containing $V$ and $\mathbf{A}$, but not the kinetic term. The expansion in the right hand side of Eq.~\eqref{e:f=} contains only those terms that give contribution up to first order in $\epsilon$ in the convolution $\mathcal{C}\ast\psi$, since the remaining terms vanish in the limit $\epsilon\to 0$.

							     Taking into account that the first two moments of $\bu$ are 
							     \BE
							     \bra u_j\ket_\bu &=& 0,\label{e:<u>}\\
										\bra u_j u_k\ket_\bu &=& \delta_{j k}\hbar\epsilon/m ,\label{e:<uu>}
										\EE
										where $\langle\cdots\rangle_\bu$ refers to the average taken with the Gaussian ${\exp(-m\bu^2/2\hbar\epsilon)/\mathcal{Z}_{EM}}$ (see Eq.~\eqref{e:Gauss...}), and that terms containing $\epsilon |\bu|$ and $|\bu|^3$ or higher can be neglected, the Gaussian average in Eq.~\eqref{e:C*psi} yields 
										\begin{widetext}
										\be\label{e:C*psiSeries}
										[\mathcal{C}\ast\psi](\bx, t) = \left(1-\frac{\epsilon}{\hbar} V\right)\psi +\frac{\hbar\epsilon}{2 m}\nabla^2\psi - i\frac{e\epsilon}{mc}\mathbf{A}\cdot\nabla\psi - i\frac{e \epsilon}{2 m c}\nabla\cdot\mathbf{A} \psi - \frac{e^2\epsilon}{2 \hbar m c^2 }\mathbf{A}^2\psi
										\ee
										\end{widetext}

										Furthermore, taking into account that
										\be 
										\begin{split}
										\left(\nabla - i\frac{e}{\hbar c}\mathbf{A}\right)^2\psi =& \nabla^2\psi - \left(\frac{e}{\hbar c}\right)^2\mathbf{A}^2\psi - \\
											& i \frac{e}{\hbar c}\left[2 \mathbf{A}\cdot\nabla\psi + (\nabla\cdot\mathbf{A})\psi\right],
										\end{split}
										\ee
										we can replace the last four terms in the right hand side of Eq.~\eqref{e:C*psiSeries} by $(\hbar\epsilon/2m)\left(\nabla - i\frac{e}{\hbar c}\mathbf{A}\right)^2\psi$. So, Eq.~\eqref{e:C*psiSeries} becomes
										\begin{widetext}
										\be\label{e:C*psi_final}
										[\mathcal{C}\ast\psi](\bx , t) =\psi(\bx ,t) + \frac{\epsilon}{\hbar}\left[\frac{\hbar^2}{2 m}\left(\nabla - i\frac{e}{\hbar c}\mathbf{A}\right)^2 \psi(\bx ,t) - V(\bx, t)\psi(\bx ,t)\right] + O(\epsilon^2).
										\ee
										\end{widetext}

										Finally, introducing Eq.~\eqref{e:C*psi_final} into Eq.~\eqref{e:partial}, multiplying by $i\hbar/\epsilon$, and taking the limit $\epsilon\to 0$ yields Eq.~\eqref{e:Main_SchrodingerEM} as we wanted to prove.
										\end{proof}

										We now use this result to prove that the von Neumann equation of a charged particle in an electromagnetic field can be written in the usual way, replacing the Hamiltonian operator by the complex-valued kernel above. We do this in the following
										\begin{corollary}
										The von Neumann equation of a charged particle in an electromagnetic field can be written as
										\be\label{e:C*rho}
										\frac{\partial\rho}{\partial t} = \frac{i}{\epsilon} (\mathcal{C}\ast\rho - \rho\ast\mathcal{C})\equiv \frac{i}{\epsilon}[\mathcal{C}, \rho] ,
										\ee
										where $\mathcal{C}$ is given in Eq.~\eqref{e:C}.
										\end{corollary}

										\begin{proof}
										For simplicity, we show this corollary for a pure density matrix ${\rho(\bx, \bx^\prime , t) = \psi(\bx, t)\psi^\ast(\bx^\prime, t)}$. The extension to more general density matrices is straightforward. Taking the time derivative of this density matrix yields
										\be\label{e:Psi_time} 
										\frac{\partial \rho(\bx, \bx^\prime , t)}{\partial t} = \frac{\partial\psi(\bx, t)}{\partial t}\psi^\ast(\bx^\prime, t) + \psi(\bx, t)\frac{\partial\psi^\ast(\bx^\prime, t)}{\partial t};
\ee 
now, replacing the time derivatives of the wave function $\psi$ and its conjugate $\psi^\ast$ in Eq.~\eqref{e:Psi_time}, respectively, by the right hand side of Eq.~\eqref{e:partial} and its conjugate we obtain
\be\label{e:Gaussian_vN}
\frac{\partial \rho}{\partial t} = \frac{i}{\epsilon}\left(\mathcal{C}\ast\rho - \rho\right)  - \frac{i}{\epsilon}\left(\mathcal{C}\ast\rho - \rho \right).
\ee
Clearly, the terms $\rho$ in the right hand side cancel out, which yields Eq.~\eqref{e:C*rho} as we wanted to prove.
\end{proof}

We now show that the von Neumann equation above can be written as a pair of 
real matrix equations like Eqs.~\eqref{em:vN-1} and \eqref{em:vN-2} in the main text. We do this in the following
\begin{theorem}
Equation~\eqref{e:C*rho}, which is equivalent to the von Neumann equation for a charged particle in an electromagnetic field, is equivalent to the following pair of real equations (cf. Eqs.~\eqref{em:vN-1} and \eqref{em:vN-2} in the main text):
	\BE
	\frac{\partial P}{\partial t} &=& - \frac{1}{\epsilon}[\mathcal{W}_s, P^T] +\frac{1}{\epsilon}[\mathcal{W}_a, P]  ,\\
		\frac{\partial P^T}{\partial t} &=& \frac{1}{\epsilon}[\mathcal{W}_s, P]  + \frac{1}{\epsilon}[\mathcal{W}_a, P^T]  ,
	\EE
	where 
	\BE
	\mathcal{W}_s(\bx,\bx^\prime) &=& \frac{1}{2}\left[\mathcal{W}(\bx, \bx^\prime) + \mathcal{W}(\bx^\prime , \bx )\right],\\
		\mathcal{W}_a(\bx,\bx^\prime) &=& \frac{1}{2}\left[\mathcal{W}(\bx, \bx^\prime) - \mathcal{W}(\bx^\prime, \bx )\right],
	\EE
	are the symmetric and antisymmetric parts of a real kernel
	\be\label{e:WrealEM}
	\mathcal{W}(\bx , \bx^\prime) = \frac{1}{\mathcal{Z}_{EM}}\exp{\left[-\frac{\epsilon}{\hbar}\mathcal{Q}_{\rm EM}(\bx , \bx^\prime) \right]}.
	\ee
	Here
	%
	\be
	\begin{split}
	\mathcal{Q}_{\rm EM}(\bx , \bx^\prime) =& \frac{m({\bx - \bx^\prime})^2}{2\epsilon^2} + V\left(\frac{\bx+\bx^\prime}{2}, t\right)\\
		&+ \frac{e}{ c}\left(\frac{\bx-\bx^\prime}{\epsilon}\right)\cdot\mathbf{A}\left(\frac{\bx+\bx^\prime}{2}, t\right)\\
		&+ \frac{\epsilon }{\hbar }\left[\frac{e}{ c}\left(\frac{\bx-\bx^\prime}{\epsilon}\right)\cdot\mathbf{A}\left(\frac{\bx+\bx^\prime}{2}, t\right)\right]^2 .\label{e:Qpre}
		\end{split}
		\ee
		%
		\end{theorem}
		\begin{proof}
		Equation~\eqref{e:C*rho} has the same form of von Neumann equation in the main text (see Eq.~\eqref{em:drhodt=[H,rho]}). Furthermore, by separating its real and imaginary parts, the kernel $\mathcal{C}$ defined in Eq.~\eqref{e:C} can be written as $\mathcal{C} = \widetilde{\mathcal{W}}_s + \widetilde{\mathcal{W}}_a / i$, where 
		%
		\BE
		\widetilde{\mathcal{W}}_s(\bx , \bx^\prime) &=& \frac{1}{\mathcal{Z}_{EM}}e^{-{\epsilon}\mcH^0\left(\bx, \bx^\prime\right)/{\hbar}  }\cos(z),\label{e:Ws=} \\
			\widetilde{\mathcal{W}}_a(\bx , \bx^\prime) &=& -\frac{1}{\mathcal{Z}_{EM}}e^{-{\epsilon} \mcH^0\left(\bx, \bx^\prime\right)/{\hbar}  }\sin(z).\label{e:Wa=} 
			\EE
			%
			Here 
	\be
\mcH^0(\bx,\bx^\prime) = \frac{m}{2}\frac{({\bx - \bx^\prime})^2}{\epsilon^2} +V\left(\frac{\bx+\bx^\prime}{2}, t\right) 
	\ee
	and 
	\be\label{e:z}
	z = \frac{\epsilon }{\hbar }\frac{e}{ c}\left(\frac{\bx-\bx^\prime}{\epsilon}\right)\cdot\mathbf{A}\left(\frac{\bx+\bx^\prime}{2}, t\right).
	\ee
	The expressions $\widetilde{\mathcal{W}}_s$ and $\widetilde{\mathcal{W}}_a$ defined in Eqs.~\eqref{e:Ws=} and \eqref{e:Wa=} are clearly symmetric and antisymmetric, respectively, under an exchange of $\bx$ and $\bx^\prime$ since $\cos (-z) = \cos z$ and $\sin(-z) = -\sin z$, where $z$ is given by Eq.~\eqref{e:z}. So, $\widetilde{\mathcal{W}}_s$ and $\widetilde{\mathcal{W}}_a$ can be considered the symmetric and antisymmetric parts of a kernel 
	%
	\be
	\begin{split}
	\widetilde{\mathcal{W}}(\bx , \bx^\prime) &= \widetilde{\mathcal{W}}_s + \widetilde{\mathcal{W}}_a\\ 
	& = \tfrac{1}{\mathcal{Z}_{EM}}e^{-{\epsilon}\mcH^0({\bx,\bx^\prime})/{\hbar} }[\cos z - \sin z],\label{e:Kpre}
	\end{split}
	\ee
	%

	Due to the very sharp Gaussian factor (since ${\epsilon\to 0}$), we can expand the sine and cosine functions up to second order in their argument since the rest gives contributions of order higher than $\epsilon$. Now, up to second order we have 
	\be\label{e:cossinexp}
	\begin{split}
	\cos z - \sin z =& \exp(-z-z^2) + O(z^3)\\
			 &= 1 - z -\frac{z^2}{2} +O(z^3),
	\end{split}
	\ee
	So, we can safely replace $\cos z - \sin z$ by $\exp(-z-z^2) $ in Eq.~\eqref{e:Kpre}. That is, we can replace $\widetilde{\mathcal{W}}$ by the kernel
	\be
	\mathcal{W}(\bx , \bx^\prime) = \tfrac{1}{\mathcal{Z}_{EM}}e^{-{\epsilon}\mcH^0({\bx,\bx^\prime})/{\hbar} -z - z^2 }.\label{e:W=}
	\ee
	Replacing $z$ in this equation by the right hand side of Eq.~\eqref{e:z} we obtain Eq.~\eqref{e:WrealEM} as we wanted to prove.
	\end{proof}
	However, the function $\mathcal{Q}_{EM}$ defined in Eq.~\eqref{e:Qpre} is not a standard Hamiltonian-like function. Indeed, the last term in the right hand side of Eq.~\eqref{e:Qpre} is proportional to $1/\hbar$. However, it is possible to replace $\mathcal{Q}_{EM}$ by a proper Hamiltonian-like function that does not depend on $\hbar$ according to the following
	\begin{theorem}
	In the limit $\epsilon\to 0$, the kernel $\mathcal{W}$ defined in Eq.~\eqref{e:WrealEM} can be replaced by the kernel
	\be\label{e:KrealEM}
	\mathcal{K}(\bx , \bx^\prime) = \frac{1}{\mathcal{Z}_{EM}}\exp{\left[-\frac{\epsilon}{\hbar}\mathcal{H}_{\rm EM}(\bx , \bx^\prime) \right]}
	\ee
	where the Hamiltonian-like function (with no tilde) is given by
	%
	\be\label{e:HrealEM}
	\begin{split}
	\mathcal{H}_{\rm EM}(\bx , \bx^\prime) &= \frac{m}{2}\left(\frac{\bx - \bx^\prime}{\epsilon}\right)^2 + V\left(\frac{\bx+\bx^\prime}{2}, t\right)\\
		&+ \frac{e}{c}\left(\frac{\bx-\bx^\prime}{\epsilon}\right)\cdot\mathbf{A}\left(\frac{\bx+\bx^\prime}{2}, t\right) \\
		&+\frac{e^2}{ m c^2} \left[\mathbf{A}\left(\frac{\bx+\bx^\prime}{2}, t\right)\right]^2 
		\end{split}
		\ee
		%
		\end{theorem}
		\begin{proof}
		The only difference between $\mathcal{Q}_{EM}$ and $\mcH_{EM}$ is the last terms in the right hand side of Eqs.~\eqref{e:Qpre} and \eqref{e:HrealEM}, respectively. So, it is convenient to single out these terms in the convolutions  $\mathcal{W}\ast\psi$ and $\mathcal{K}\ast\psi$. Let us start with the convolution $\mathcal{W}\ast\psi$. Using Eqs.~\eqref{e:WrealEM} and \eqref{e:Qpre} we can write
		\be\label{e:<gz>}
		[\mathcal{W}\ast\psi](\bx) = \bra g(\bx , \bu) e^{-z^2}
		\ket_\bu
		\ee
		where we have introduced the change of variables $\bu = \bx-\bx^\prime$, so that $(\bx + \bx^\prime)/2 = \bx - \bu/2$. Here $z$ is given in Eq.~\eqref{e:z} and $\bra\cdots\ket_\bu$ 
		%
		%
		denotes the Gaussian average associated to the kinetic term in Eq.~\eqref{e:Qpre} (see Eq.~\eqref{e:Gauss...}). Furthermore,
		\be\label{e:gfunc}
		g(\bx , \bu)= e^{-\epsilon V(\bx-\bu/2)/\hbar + e\bu\cdot\mathbf{A}(\bx-\bu/2,t)/c\hbar}\psi(\bx - \bu),
		\ee
		denotes the remaining terms in the convolution $\mathcal{W}\ast\psi$. We can expand $e^{-z^2} = 1-z^2 + O(z^4)$ in Eq.~\eqref{e:<gz>} up to first order in $z^2$ since terms $O(z^4)$ give contributions of $O(\epsilon^2)$. So,
		\be\label{e:W*psi+O(eps)}
[\mathcal{W}\ast\psi](\bx) = \bra g(\bx , \bu) \ket_{\bu} - g(\bx, \bx)\bra z^2\ket_\bu + O(\epsilon^2)
	\ee
	where 
	%
	\be\label{e:averages}
	\begin{split}
	\bra z^2\ket_\bu &= \left(\frac{e}{\hbar c}\right)^2\bra\left[\bu\cdot\mathbf{A}\left(\bx-\frac{\bu}{2}, t\right)\right]^2\ket_\bu \\ 
	&=\left(\frac{e}{\hbar c}\right)^2\bra\bu^2\ket_\bu\cdot[\mathbf{A}\left(\bx, t\right)]^2 + O(\epsilon^2)\\ 
	&= \frac{\epsilon}{\hbar}\frac{e^2}{ m c^2}[\mathbf{A}(\bx, t)]^2 + O(\epsilon^2).
	\end{split}
	\ee
	%
	Furthermore, we have done $\mathbf{A}(\bx-\bu/2, t)= \mathbf{A}(\bx, t) + O(|\bu|)$ in Eq.~\eqref{e:averages} and $g(\bx , \bu) = g(\bx, \bx) + O(|\bu|,\epsilon)$ in Eq.~\eqref{e:W*psi+O(eps)}, respectively, because $\mathbf{A}(\bx ,t)$  and $g(\bx, \bx)$ are the only terms that contribute to first order in $\epsilon$ since $\bra \bu^2\ket_\bu $ is already $O(\epsilon)$ (see Eqs.~\eqref{e:<u>} and \eqref{e:<uu>})%

	Now, proceeding similarly with the convolution $\mathcal{K}\ast\psi$ we have
	\be
	[\mathcal{K}\ast\psi](\bx) = \bra g(\bx , \bu) e^{-y}\ket_\bu
	\ee
	where $g(\bx, \bu)$ is defined in Eq.~\eqref{e:gfunc} and
	\be\label{e:y=}
	y = \frac{\epsilon e^2}{ \hbar m c^2 } \left[\mathbf{A}\left(\bx-\bu/2, t\right)\right]^2.
	\ee
	Notice that $y$ is already of first order in $\epsilon$. So, we can neglect the dependency of $\mathbf{A}$ on $\bu$ and do the expansion
	\be
	e^{-y} = 1- \frac{\epsilon e^2}{\hbar m c^2 } \left[\mathbf{A}\left(\bx, t\right)\right]^2  + O(\epsilon^2, \epsilon |\bu|),
	\ee
	where we have introduced the explicit value of $y$ in the right hand side, which is given in Eq.~\eqref{e:y=}.
	So, 
	\be\label{e:K*psi+O(eps)}
[\mathcal{K}\ast\psi](\bx) = \bra g(\bx , \bu) \ket_{\bu} - g(\bx, \bx)\frac{\epsilon e^2}{ m c^2 \hbar} \left[\mathbf{A}\left(\bx, t\right)\right]^2  + O(\epsilon^2)
	\ee
	Finally, introducing Eq.~\eqref{e:averages} into Eq.~\eqref{e:W*psi+O(eps)} and comparing to Eq.~\eqref{e:K*psi+O(eps)} we can see that $\mathcal{K}\ast\psi = \mathcal{W}\ast\psi + O(\epsilon^2)$, that is $\mathcal{K}\ast\psi = \mathcal{W}\ast\psi$ in the limit $\epsilon\to 0$ as we wanted to prove.
	\end{proof}

	\

	\section{Modeling scientists doing experiments}\label{s:model_scientists}

	Here we discuss two well-known modeling frameworks in cognitive science, i.e., active inference (Appendix~\ref{s:active}) and enactive cognition (Appendix~\ref{s:enactive}), that are relevant for our purpose. However, we take a more relational approach than traditionally done in these two modeling frameworks. Indeed, somewhat analogous to the relational interpretation of quantum mechanics (RQM)~\cite{Rovelli-1996}, the modeling of a scientist doing an experiment is done from the perspective of another scientist (see Sec.~\ref{sm:others...}). 

	\subsection{Active inference: world as a generative process, scientists as generative models}\label{s:active}
	\begin{figure*}
	\includegraphics[width=0.8\textwidth]{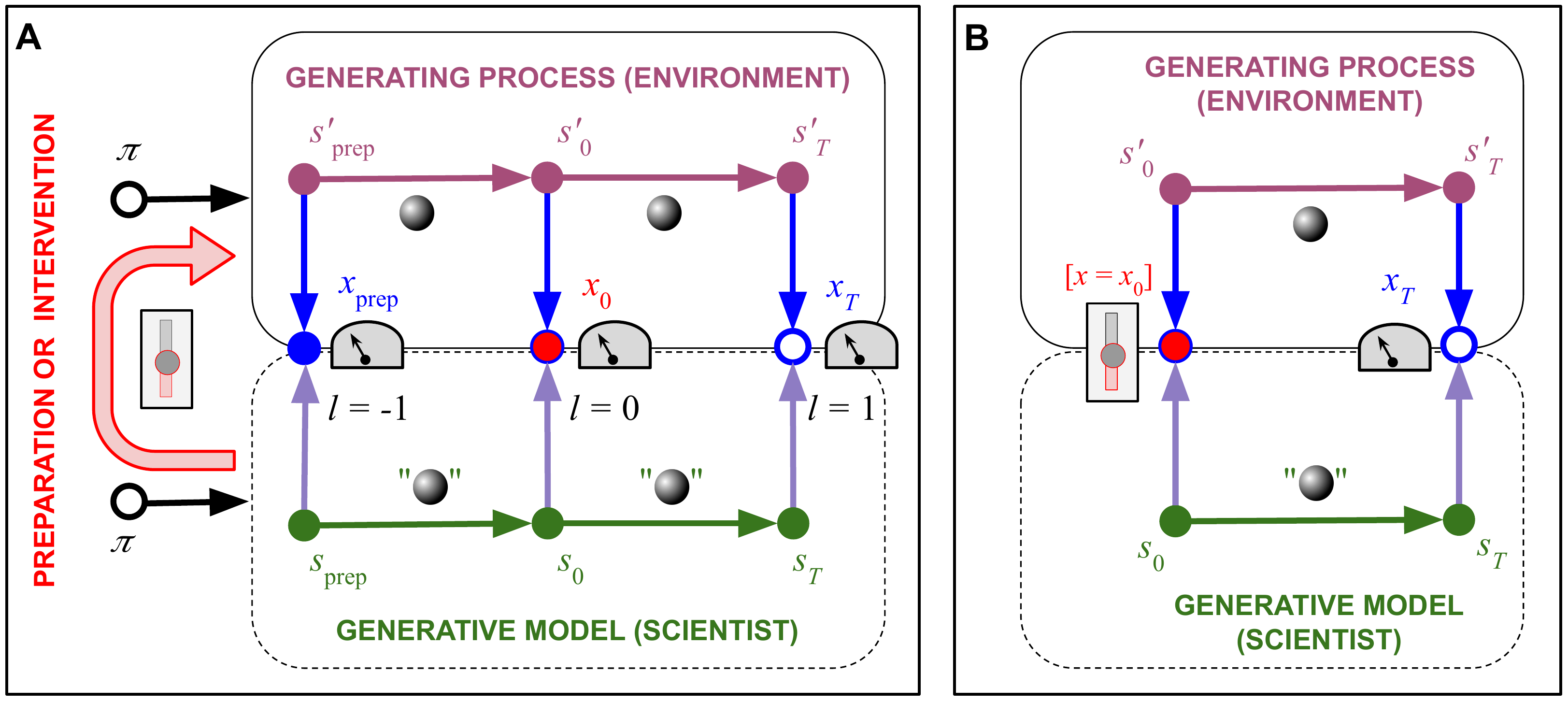}
	\caption{{\em Active inference:} (A) Graphical model characterizing active inference (cf. Figs. 1 and 2 in Ref.~\cite{friston2017graphical} as well as Figs.~2 and 3 in Ref.~\cite{schwobel2018active}). The upper graphical model enclosed within a solid line is the {\em generative process} associated to the external system. The only accessible information about this generative process is the data it generates on the observer's sensors. The lower graphical model enclosed within a dashed line is the {\em generative model} the observer has about the external world. (B) When the scientist performs the required actions to consistently transform the variable position $x_{\rm prep}$ into the same initial position $x_0$, she effectively removes all causal dependencies before the start of the experiment at time step $\ell=0$. This could be interpreted as a form of causal intervention on the system. We denote this here as $\textsc{do}[x=x_0]$. 
	}\label{f:active}
\end{figure*}

Here we briefly discuss some aspects of active inference in the framework of a scientist carrying out an experiment. Although we present some technical details for the reader that may not be familiar with it, our main purpose is to highlight the main underlying concepts. In active inference the external world---an experimental system in this case---is considered as a generative {\em process}, while the organism---here a scientist---perceiving, interacting with, and learning about such an external world is considered as (or to have) a generative {\em model} (see Fig.~\ref{f:active}; cf. Fig.~2 in Ref.~\cite{friston2017graphical} and Figs. 1 and 2 in Ref.~\cite{schwobel2018active}). We discuss these in the next subsections, closely following Ref.~\cite{schwobel2018active}.

\subsubsection{Experimental systems as generative processes}\label{s:gen_process}

Following active inference, the scientist's (controlled) environment, i.e. the experimental system, is considered hidden to her; she can only indirectly access it by the data it generates in her sensorium via her observations. In Fig.~\ref{f:active}A we represent the environment by a Bayesian network enclosed within a solid rounded rectangle, which depends on the actions of the organism (external arrow pointing towards the solid rounded rectangle; cf. Fig. 2 in Ref.~\cite{friston2017graphical}; see Sec.~2.1 in Ref.~\cite{schwobel2018active}). Accordingly, the state of the environment at time step $\ell$ is described by hidden variables $s_\ell^\prime$ (top dark magenta circles) which can generate an observation $x_\ell$ (center blue and red circles) with a probability $\Omega_\ell(x_\ell | s_\ell^\prime)$ (blue arrows pointing downwards).

The environment dynamics is specified by the transition probability $\Theta_\ell(s_{\ell+1}^\prime | s_\ell^\prime , a_\ell)$ that the environment is in state $s_{\ell +1}^\prime$ at time step $\ell +1$, given that at the previous time step its state was $s_\ell^\prime$ and the scientist performed action $a_\ell$, e.g., by moving some knobs. The dynamical dependency between hidden variables is represented in Fig.~\ref{f:active}A by the top horizontal dark magenta arrows. The dependency of these dynamics on the scientist's actions is represented by the black arrow external to the solid rounded rectangle and pointing towards it. This is to emphasize that the scientist can select a whole sequence of actions according to a behavioral policy~\cite{schwobel2018active,friston2017graphical}, $\pi$, as discussed in the next subsection. 

To keep the discussion at the minimal level of complexity required to illustrate the relevant concepts for our purpose, we focus here only on three time steps, $\ell = -1,0, 1$ (see Fig.~\ref{f:active}A). However, each transition from a time step $\ell$ to the next $\ell+1$ can be partitioned into as many time steps as desired~\cite{schwobel2018active}.

\subsubsection{Scientists as generative models}\label{s:gen_model}

Following active inference, the scientist is considered to be, or to have physically encoded in her neural system and perhaps body, a generative model of her (controlled) environment, i.e. of the experimental system. This generative model is represented in Fig.~\ref{f:active}A by a Bayesian network within a dashed rounded rectangle, which mirrors the Bayesian network representing the environment. The generative model is defined as a joint probability distribution over observations $x_\ell$ (middle blue and red circles), internal ``copies'' $s_\ell$ of the environment's hidden states $s_\ell^\prime$ (bottom green circles), which are encoded in the scientist's neural system or body, and behavioral policies $\pi$ (black node external to the solid rounded square). The latter could be specified, for instance, by a sequence of control states $u_\ell$ (see Sec.~2.2 in Ref.~\cite{schwobel2018active}), i.e. $\pi = (u_{-1}, u_0 , u_1)$, which denote a subjective abstraction of an action, such as a neuronal command to execute a specific action in the environment~\cite{schwobel2018active}. In Ref.~\cite{schwobel2018active} a one-to-one mapping is assumed between a selected control state $u_\ell$ and executed action $a_\ell$ in each time step $\ell$.

The generative model is represented in Fig.~\ref{f:active}A by a Bayesian network within a dashed rounded rectangle, which mirrors the Bayesian network representing the environment. It can be written as~\cite{schwobel2018active} (see Eq.~(2.4) therein)
\begin{widetext}
\be\label{e:GM}
\mathcal{P}^{\rm gen}(\bx, \bs, \pi) = p^{\rm pol}(\pi) p_{-1}(s_{-1})\prod_{\ell = 0}^1 \mathcal{P}^{\rm obs}_\ell(x_\ell | s_\ell)\mathcal{P}_\ell^{\rm dyn}(s_\ell |s_{\ell-1}, \pi ),
	\ee
	\end{widetext}
	where $\bx = (x_{-1},x_0, x_1)$ and $\bs = (s_{-1}, s_0 , s_1)$. Here $\mcP^{\rm dyn}$ (bottom horizontal green arrows in Fig.~\ref{f:active}A) specifies the scientist's model of the environment's hidden dynamics, which can be affected by the actions the scientist performs according to the behavioral policy $\pi$. Furthermore, $\mcP_\ell^{\rm obs}$ (bottom purple arrows pointing upwards in Fig.~\ref{f:active}A) specifies the model of how hidden states of the environment generate observations. Finally, $p_{-1}$ and $p^{\rm pol}$ are priors over the initial state of the environment and the policy, respectively.

	Now, when carrying out an experiment a scientist first prepares the state of the experimental system at the start of the experiment, i.e., at time step $\ell= 0$. Say the experimental system is a particle in a piece-wise linear potential (see Fig.~\ref{f:circular_noRQM}A in the main text). This could be done, for instance, by performing a measurement at a previous time step, $\ell = -1$, say of the position of the particle $x_{-1} = x_{\rm prep}$ as displayed in a reading device---this would correspond to the first time step in Fig.~\ref{f:active}A. Afterwards, the scientist can act on the system to consistently obtain a desired observation, $x_0=x_0^\ast$, at time step $\ell =0 $ when the experiment starts. 

	For instance, the scientist can generate some commands that would lunch a mechanism that moves the particle an amount $x_0^\ast-x_{\rm prep}$ in such a way that the scientist consistently observes a given position, $x_0=x_0^\ast$, as displayed on a reading device, modulo experimental error. Different observations $x_{\rm prep}$ at time step $\ell =-1$ would lead to different actions. The aim of those actions is precisely that an observation at time step $\ell=0$ always yields the same result, $x_0=x_0^\ast$. Since the observation at time step $\ell = 0$ yields consistently the same result, this effectively removes the dynamical dependencies before this time step, when the experiment starts. This amounts at a form of causal intervention. We denote this here as $\textsc{do}[x_0=x_0^\ast]$.

	\subsection{Enactivism: dynamical coupling between scientist and world}\label{s:enactive}
	\begin{figure*}
	\includegraphics[width=0.8\textwidth]{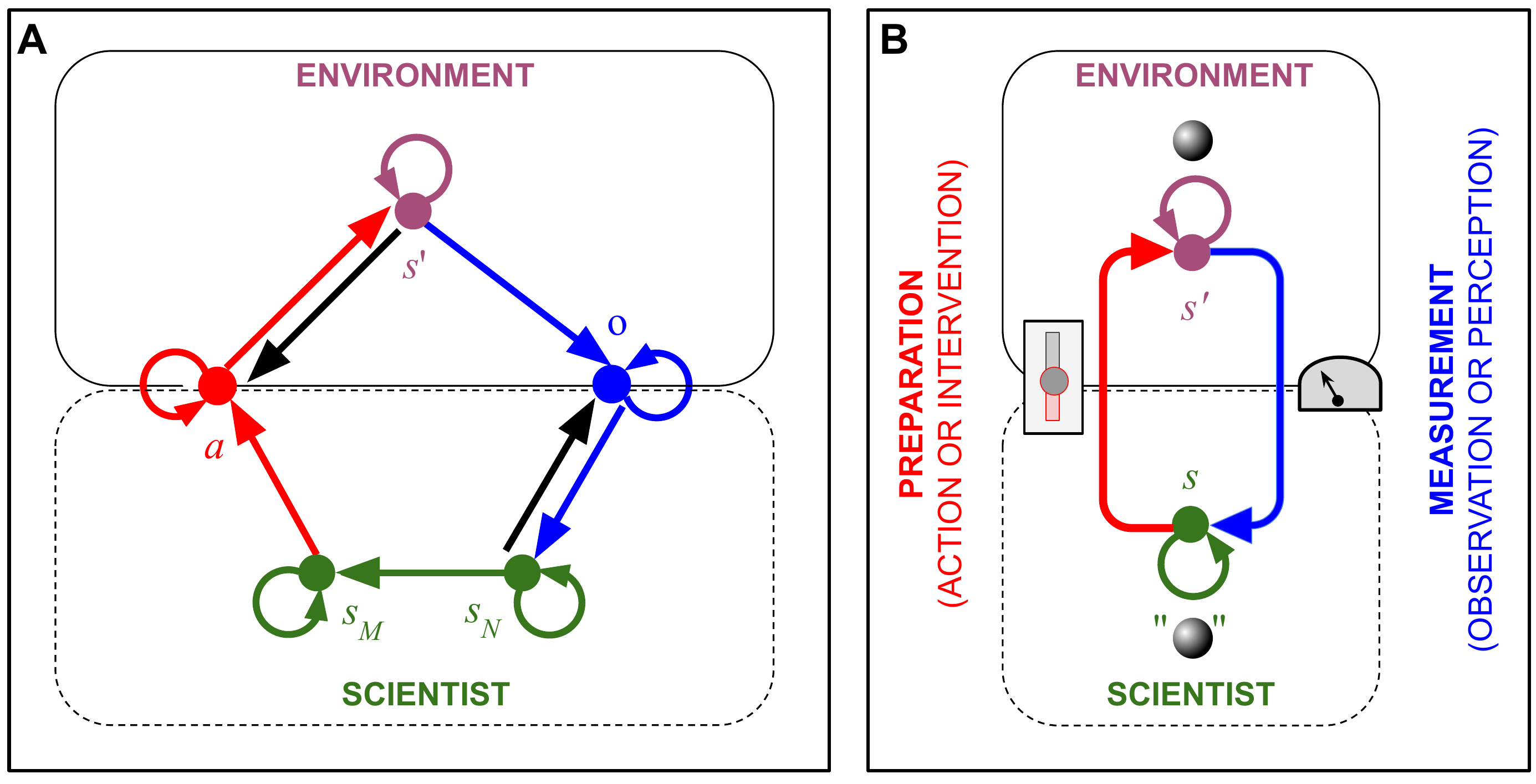}
	\caption{{\em Enactivist framework:} (A) Dependency graph of the enactive cognitive model described by Eqs.~\eqref{e:env}-\eqref{e:act}, as presented in Ref.~\cite{di2017sensorimotor} (see Ch. 3 and Fig.~3.5 therein; see also Ref.~\cite{buhrmann2013dynamical}). Nodes represent variables. An arrow indicates that the variable it points to depends on the variable in its tail---in particular, circular arrows indicate recurrent dependencies. This dependency graph represents a circular interaction: scientist's actions, $a$, influence the environment's state, $s^\prime$; environment's states influence the scientist's sensor activity, $o$, via observations; sensor activity influences neural activity, $s_N$; neural activity influences motor activity, $s_M$, i.e. outflowing movement-producing signals; finally, motor activity influences scientist's actions, which closes the interaction loop. Although, internal neural activity and environment's states can influence back, respectively, sensor activity and scientist's actions---e.g. by changing body configuration---the global dynamics is clockwise. (B) Simplified dependency graph that only shows the circular dependency between environment's states, $s^\prime$, and scientist's internal states, $s= (s_M, s_N)$. An action can  prepare a desired state of the experimental system, e.g., a hand movement to turn a knob that places a particle in a desired location---in this sense it may be considered as a form of causal intervention. An observation can be mediated via a reading device, e.g. to determine the final position of the particle (see Fig.~\ref{f:circular_noRQM}A in the main text).  }\label{f:enactive}
	\end{figure*}

	Active inference, as briefly described above, still has a representationalist flavour in that the task of the scientist is to learn a model, i.e. a representation, as accurate as possible of the environment's dynamics, including how her own actions affect it. The environment, which is described by the {\em fixed} probability distributions $\Theta_\ell$ and $\Omega_\ell$ in Sec.~\ref{s:gen_process}, is considered as something externally given. This is reflected in that the topology of the Bayesian network representing the scientist mirrors the topology of the Bayesian network representing the environment. In particular, the internal and external dynamics (horizontal arrows in Fig.~\ref{f:active}) flow in the same direction. 

	In contrast, the enactive approach~\cite{varela2017embodied,di2017sensorimotor,gallagher2017enactivist} puts a stronger emphasis on the dynamical coupling between scientist and environment~\cite{buhrmann2013dynamical, di2017sensorimotor}. The focus is often on the particular sensor and motor systems of an individual like, e.g., a human or a robot. However, scientists manage to transcend their own sensorimotor limitations with the aid of technological devices that therefore enable them to couple to the world in ``more fundamental'' ways. For instance, the kind of manipulations and observations associated to light-matter interaction experiments are enabled by, e.g., lasers and electron microscopes. These kinds of couplings between scientists and world are hardly possible without such technologies. Such technologies are created by scientists themselves in their quest for lawful regularities. In this quest scientists have to learn how to build suitable experimental devices, how to stabilize the experimental system and achieve repeatability, how to obtain a decent measurement precision, etc. In general, how to achieve reliability---conditions {\bf R1}-{\bf R3} in Sec.~\ref{sm:intro} in the main text.
	Our work is focused only on the post-learning stage, so it does not depend on a specific theory of learning. 

	From an enactive perspective, we could consider both the scientist and the environment as physical systems involved in a circular interaction possibly enabled by technological devices (see Fig.~\ref{f:enactive}; cf. Fig.~3.5 in Ref.~\cite{di2017sensorimotor}).  
	However, to the best of our knowledge, the mathematical formalization of enactivism is not as well developed as that of active inference. Indeed, we are aware of only a couple of rather recent works~\cite{buhrmann2013dynamical, di2017sensorimotor} that attempt to do that. Here we briefly discuss some of the main concepts underlying enactivism, closely following~\cite{di2017sensorimotor} (see Ch. 3 therein; see also Ref.~\cite{buhrmann2013dynamical}). 

	For instance, the (controlled) environment or experimental system could be described by state variables $s^\prime$, e.g., the position of a particle in a piece-wise linear potential (see Fig.~\ref{f:circular_noRQM}A in the main text). Similarly, we could use variables $a$ to represent actions the scientist perform on the experimental system, e.g. by moving her hand to turn a knob that puts the particle in a desired position---these kinds of actions could be considered effectively as state preparations or causal interventions. The dynamics of the environment can then be described by~\cite{di2017sensorimotor} 
	\be\label{e:env}
	\frac{\mathrm{d} s^\prime}{\mathrm{d} t} = \mathcal{E}(s^\prime , a),
	\ee
	where the function $\mathcal{E}$ captures the dependency of the environment's current state on its previous state and the scientist's previous actions. 

	The scientist's sensor activity, here denoted by variables $o$, is influenced by the environment via her observations that stimulate her sensorium.  Furthermore, in Refs.~\cite{di2017sensorimotor, buhrmann2013dynamical} the scientist is assumed to have an internal neural dynamics, here described by variable $s_N$, which modulates the sensors activity. The scientist's sensors' dynamics can then be described by
	\be\label{e:obs}
	\frac{\mathrm{d} o}{\mathrm{d} t} = \mathcal{O}(s^\prime , s_N),
	\ee
	where the function $\mathcal{O}$ captures the dependency of the scientist's sensor dynamics on the state of both the environment and the scientist's internal neural dynamics.

	The dynamics of neural activity is assumed to depend on sensor activity and on the neural activity itself, i.e.
	\be\label{e:neu}
	\frac{\mathrm{d} s_N}{\mathrm{d} t} = \mathcal{N}(s , s_N),
	\ee
	where the function $\mathcal{N}$ captures such dependencies. Additionally, the scientist's outflowing movement-producing signals, or motor activity, denoted here by $s_M$, is assumed to be influenced by the neural activity, $s_N$, i.e.
	\be\label{e:mot}
	\frac{\mathrm{d} s_M}{\mathrm{d} t} = \mathcal{M}(s_N),
	\ee
	where the function $\mathcal{M}$ captures such an influence. 

	Finally, the interaction loop is closed by assuming the scientist's actions, which can be implemented via body configurations, depend on the current actions she performs, on her internal motor activity, and on the state of the environment. So, the scientist's actions dynamics can be described as
	\be\label{e:act}
	\frac{\mathrm{d} a}{\mathrm{d} t} = \mathcal{A}(a, s_M, s^\prime),
	\ee
	where the function $\mathcal{A}$ captures such dependencies. 

	\

	\section{Principle of maximum caliber and factor graphs}\label{s:MaxCal}

	Here we discuss the principle of maximum dynamical entropy, or principle of maximum caliber. This is a general variational principle, similar to the free energy principle, from which a variety of models at, near, and far from equilibrium can be derived~\cite{presse2013principles}.  We have used this principle in the main text to derive the form of the stationary distribution over the dynamical trajectories characterizing a scientist interacting with an experimental system. 

	The principle of maximum entropy~\cite{jaynes2003probability} to derive some common equilibrium probability distributions in statistical physics can be extended to the so-called principle of maximum caliber to deal with non-equilibrium distributions on trajectories~\cite{presse2013principles}. In particular Markov chains and Markov processes can be derived from the principle of maximum caliber (see e.g. Sec. IX B in Ref.~\cite{presse2013principles}). We introduce this principle here with an example relevant for our discussion in the main text. 

	Consider a probability distribution $\widetilde{\mathcal{P}}(\widetilde{\bx})$ on (discretized) close paths $x_0\to x_1\to\cdots\to x_{k-1}\to x_0$, denoted as $\widetilde{\bx} = (x_0,\dotsc , x_{k-1})$, where $x_\ell$ refers to the position at time $t=\ell\epsilon$. Assume that we only have information about the average energy on the (discretized) paths given by
	\be\label{e:Eav}
	\mathcal{H}_{\rm av}[\widetilde{\mathcal{P}}] = \bra\frac{1}{T}\sum_{\ell=0}^{k-1}\mathcal{H}_\ell(x_{\ell+1}, x_\ell)\epsilon\ket_{\widetilde{\mcP}} ,
	\ee
	where $x_k=x_0$, $T=k\epsilon$ is the total time duration of the path, and $\mathcal{H}_\ell$ is the ``energy'' or Hamiltonian-like function at time step $\ell$. Here
	\be\label{e:Eav}
	\bra f\ket_{\widetilde{\mathcal{P}}} = \int \widetilde{\mathcal{P}}(x_0,\dotsc ,x_{k-1})f(x_0,\dotsc , x_{k-1})\prod_{\ell = 0}^{k-1}\mathrm{d}x_\ell  ,
	\ee
	denotes the average value of a generic function $f$ of a path, with respect to a generic path probability distribution $\widetilde{\mcP}$. For convenience, here we are using integrals instead of sums, as in the main text. However, our analysis is valid for discrete variables too by changing these integrals by sums, $\int\to\sum$. 

	The principle of maximum caliber tells us that among all possible probability distributions we should choose the one that both maximizes the entropy
	\be\label{e:entropy}
	\mathcal{S}[\widetilde{\mathcal{P}}] = -\bra\ln\widetilde{\mathcal{P}}(x_0,\dotsc ,x_{k-1})\ket_{\widetilde{\mcP}},
	\ee
	and is consistent with the information we have, i.e. $\mathcal{H}_{\rm av}[\widetilde{\mathcal{P}}] = E_{\rm av}$, where $E_{\rm av}$ is the fixed value of the average energy. Introducing a Lagrange multiplier $\lambda$ to enforce the constraint on the average energy, the constrained maximization of $\mathcal{S}[\widetilde{\mathcal{P}}]$ becomes equivalent to the maximization of the Lagrangian $\mathcal{S}[\widetilde{\mcP}] - \lambda \mathcal{H}_{\rm av}[\widetilde{\mathcal{P}}]$. The solution to this problem is the distribution
	\be\label{e:MaxCal}
	\widetilde{\mathcal{P}}(x_0,\dotsc , x_{k-1}) = \frac{1}{\mathcal{Z}}\exp\left[-\frac{\lambda}{T}\sum_{\ell=0}^{n-1}\mathcal{H}_\ell(x_{\ell+1} , x_{\ell})\epsilon\right],
	\ee
	where $\mathcal{Z}$ is the normalization factor. 

	Notice that $\widetilde{\mathcal{P}}$ in Eq.~\eqref{e:MaxCal} can be written as a product of factors
	\be\label{e:P_prodF}
	\widetilde{\mathcal{P}}(x_0,\dotsc , x_{k-1}) = \frac{1}{Z}\prod_{\ell=0}^{k-1} F_\ell(x_{\ell+1} , x_{\ell}),
	\ee
	with $x_k =x_0$. Without loss of generality, we can choose the factors as
	\be\label{e:F}
	F_\ell(x_{\ell + 1} , x_{\ell}) = \frac{1}{|\mathcal{A}|}\exp\left[-\frac{\lambda}{T} \mathcal{H}_\ell(x_{\ell+1} , x_{\ell})\epsilon\right],
	\ee
	with $|\mathcal{A}|= \sqrt{2\pi T\epsilon / m\lambda}$, so $Z = \mathcal{Z}/|\mathcal{A}|^{n}$ in Eq.~\eqref{e:P_prodF}.  

	\

	\section{Factor graphs on chains and Markov processes}\label{s:chainMarkov}
	Here we show the well-known fact that a factor graph with the topology of a chain can be written as a Markov chain. First, notice that by marginalizing the probability distribution defined in Eq.~\eqref{e:P_prodF} over all variables except $x_\ell$ and $x_{\ell+1}$ we obtain
	\begin{widetext}
	\BE
	\mathcal{P}_\ell(x_{\ell +1} , x_{\ell }) &=& \frac{1}{Z} F_\ell(x_{\ell +1}, x_{\ell }) Z_{\to\ell}(x_\ell)Z_{\ell+1\leftarrow}(x_{\ell +1}),\label{e:P_BP}\\
		p_\ell(x_\ell ) &=& \sum_{x_{\ell+1}}\mathcal{P}_\ell(x_{\ell +1} , x_{\ell }) = \frac{1}{Z} Z_{\to\ell}(x_\ell)Z_{\ell\leftarrow}(x_{\ell }),\label{e:p_BP}
		\EE
		\end{widetext}
		where the partial partition functions $Z_{\to\ell}(x_\ell)$ and $Z_{\ell\leftarrow}(x_{\ell })$ of the original factor graph are given by the partition functions of the modified factor graphs that contain all factors $F_{\ell^\prime}$ to the left (i.e. $\ell^\prime < \ell$) and to the right (i.e. $\ell^\prime \geq \ell$) of variable $x_\ell$, respectively; i.e. (cf. Eq. (14.2) in Ref.~\cite{Mezard-book-2009}).
		\BE
		Z_{\to\ell}(x_\ell) &=& \sum_{x_0,\dotsc , x_{\ell -1}} \prod_{\ell^\prime = 0}^{\ell-1}F_{\ell^\prime}(x_{\ell^\prime +1}, x_{\ell^\prime}),\label{e:Z->}\\
			Z_{\ell\leftarrow}(x_{\ell }) &=& \sum_{x_{\ell+1},\dotsc , x_{n}} \prod_{\ell^\prime = \ell}^{n-1}F_{\ell^\prime}(x_{\ell^\prime + 1}, x_{\ell^\prime }),\label{e:Z<-}
			\EE
			with 
			\BE
			Z_{\to 0}(x_0) &=& 1,\label{e:Z->0=1}\\ 
			Z_{n\leftarrow}(x_n) &=& 1. \label{e:Zn<-=1}
			\EE
			$Z_{\to\ell}(x_\ell)$ and $Z_{\ell\leftarrow}(x_{\ell })$ can be interpreted as information that arrives to variable $\ell$ from the left and from the right side of the graph, respectively.

			Indeed, from Eqs.~\eqref{e:P_BP} and \eqref{e:p_BP} we obtain
			\be
			\begin{split}
			\mcP_\ell^+(x_{\ell +1}|x_\ell) =& \frac{\mcP_\ell(x_{\ell +1},x_\ell)}{p_\ell(x_\ell)}\\ 
			=& F_\ell(x_{\ell +1},x_\ell)\frac{Z_{\ell +1\leftarrow}(x_{\ell+1})}{Z_{\ell\leftarrow}(x_\ell)}.\label{e:P(|)fromF}
			\end{split}
			\ee
			So, we can write
			\be
			\begin{split}
			\mcP_{\rm ch}(\bx) = &\frac{1}{Z}\prod_{\ell =0}^{n-1}F_\ell(x_{\ell +1},x_\ell)\\ 
			=& \prod_{\ell = 0}^{n-1}\mcP_\ell^+(x_{\ell +1}|x_\ell) \frac{Z_{0\leftarrow}(x_0)\prod_{\ell = 1}^{n-1}Z_{\ell\leftarrow}(x_\ell)}{Z\prod_{\ell = 0}^{n-2}Z_{\ell +1\leftarrow}(x_{\ell+1})}\\
			 =&  p_0(x_0)\prod_{\ell=0}^{n-1} \mcP_\ell^+(x_{\ell +1}|x_\ell).\label{e:fromFtoMarkov}
			 \end{split}
			 \ee

			 To go from the first line to the second line in Eq.~\eqref{e:fromFtoMarkov} we have written $F_\ell(x_{\ell +1 },x_\ell)$ in terms of $\mcP_\ell^+(x_{\ell +1 }|x_\ell)$, $Z_{\ell\leftarrow}(x_\ell)$, and $Z_{\ell + 1\leftarrow}(x_{\ell +1})$ using Eq.~\eqref{e:P(|)fromF}. Furthermore, we have taken the term $Z_{0\leftarrow}(x_0)$ out of the product in the numerator in the second line, so the remaining product starts at $\ell = 1$. We have also taken into account that $Z_{n\leftarrow}(x_n) = 1$, so the product in the denominator goes up to $\ell = n-2$ only. Notice that the products
	\be
\prod_{\ell = 1}^{n-1}Z_{\ell\leftarrow}(x_\ell),\hspace{0.5cm} \textrm{and}\hspace{0.5cm}\prod_{\ell = 0}^{n-2}Z_{\ell +1\leftarrow}(x_{\ell+1})
	\ee
	in the numerator and denominator, respectively, are equal, so they cancel out. Finally, to go from the second to the third line in Eq.~\eqref{e:fromFtoMarkov}, we have taken into account that
	\be
	p_0(x_0) = \frac{Z_{0\leftarrow}(x_0)}{Z},
	\ee
	since $Z_{\to 0}(x_0) = 1$ (see Eqs.~\eqref{e:p_BP} and \eqref{e:Z->0=1}).

	\section{Factor graphs on cycles and Bernstein processes}\label{s:cycleBernstein}
	Here we show that a factor graph with the topology of a circle, which is given by Eq.~\eqref{em:circular}, can be written as a Bernstein process rather than as a Markov chain. Indeed, from Eq.~\eqref{em:circular} (with $x_0^\prime = x_0$) we obtain for the two variable marginal in Eq.~\eqref{em:P_Bernstein}
	\be\label{e:P(xn,x0)}
	\begin{split}
	p(x_0, x_n) =& \sum_{x_1,\dotsc , x_{n-1}}\mcP(\bx) \\
		     =&\frac{1}{Z}\widetilde{F}_n(x_0,x_n)\sum_{x_1,\dotsc , x_{n-1}}\prod_{\ell =0}^{n-1} F_\ell(x_{\ell+1},x_\ell),
	\end{split}
	\ee
	so,
	\be
	\begin{split}\label{e:mcP(|x0xn)}
	\mcP(\bx^\prime| x_n, x_0) =& \frac{\mcP(\bx)}{p(x_0,x_n)} \\
				    =&\frac{1}{Z(x_n,x_0)}\prod_{\ell =0}^{n-1} F_\ell(x_{\ell+1},x_\ell),
	\end{split}
	\ee
	where $\bx^\prime=(x_1,\dotsc , x_{n-1})$ and 
	\be
	Z(x_n,x_0)=\sum_{x_1,\dotsc , x_{n-1}}\prod_{\ell =0}^{n-1} F_\ell(x_{\ell+1},x_\ell),
	\ee
	is the new normalization constant. In particular, the term $\widetilde{F}_n(x_0,x_n)$ which closes the loop has cancelled out (see Eq.~\eqref{em:Ftilde}). Thus, the new graphical model $\mcP(\bx^\prime| x_n, x_0)$ on variables $\bx^\prime = (x_1,\dotsc , x_{n-1})$ has the topology of a chain and we obtain the equivalent of Eqs.~\eqref{e:P_BP} and \eqref{e:p_BP}
	\begin{widetext}
	\BE
	\mathcal{P}_\ell(x_{\ell +1} , x_{\ell }|x_n, x_0) &=& \sum_{x_1,\dotsc , x_{\ell-1},x_{\ell+2},\dotsc x_{n-1}}\mcP(\bx^\prime|x_n,x_0) = \frac{ F_\ell(x_{\ell +1}, x_{\ell }) Z_{\to\ell}(x_\ell|x_0)Z_{\ell+1\leftarrow}(x_{\ell +1}|x_n)}{Z(x_n, x_0)},\label{e:P_BPx0xn}\\
		p_\ell(x_\ell |x_n, x_0) &=& 
		\frac{1}{Z(x_n , x_0)} Z_{\to\ell}(x_\ell|x_0)Z_{\ell\leftarrow}(x_{\ell }|x_n),\label{e:p_BPx0xn}
		\EE
		\end{widetext}
		where $\ell =1,\dotsc , n-2$ for Eq.~\eqref{e:P_BPx0xn} and $\ell =1,\dotsc , n-1$ for Eq.~\eqref{e:p_BPx0xn}. Here
		%


		\BE
		Z_{\to\ell}(x_\ell|x_0) &=& \sum_{x_1,\dotsc , x_{\ell -1}} \prod_{\ell^\prime = 0}^{\ell-1}F_{\ell^\prime}(x_{\ell^\prime +1}, x_{\ell^\prime}),\label{e:Z->x0xn}\\
			Z_{\ell\leftarrow}(x_{\ell }|x_n) &=& \sum_{x_{\ell +1}, \dotsc x_{n-1}} \prod_{\ell^\prime = \ell}^{n-1}F_{\ell^\prime}(x_{\ell^\prime + 1}, x_{\ell^\prime }).\label{e:Z<-x0xn}
			\EE
			%
			Importantly, $Z_{\to\ell}(x_\ell|x_0) $ is only conditioned on $x_0$ because it propagates information from ``past'' to ``future'' and $Z_{\ell\leftarrow}(x_{\ell }|x_n)$ is only conditioned on $x_n$ because it propagates information from ``future'' to ``past.'' Notice that, according to Eqs.~\eqref{e:Z->x0xn} and \eqref{e:Z<-x0xn}, we have 
			\BE
			Z_{\to 1}(x_1|x_0) &=& F_0(x_1, x_0),\label{e:Z->1x0}\\
				Z_{n-1\leftarrow}(x_{n-1}|x_n) &=& F_{n-1}(x_n, x_{n-1})\label{e:Zn-1<-xn}.
				\EE

				So, the equivalent of Eq.~\eqref{e:P(|)fromF} is
				\be
				\begin{split}
				\mcP_\ell^+(x_{\ell +1}|x_n,x_\ell)=& \mcP_\ell^+(x_{\ell +1}|x_n,x_\ell, x_0) \\
					=&\frac{\mcP_\ell(x_{\ell +1},x_\ell|x_n, x_0)}{p_\ell(x_\ell|x_n, x_0)}\\ 
					=& F_\ell(x_{\ell +1},x_\ell)\frac{Z_{\ell +1\leftarrow}(x_{\ell+1}|x_n)}{Z_{\ell\leftarrow}(x_\ell|x_n)},\label{e:P(|x0xn)fromF}
					\end{split}
					\ee
					for $\ell =1,\dotsc , n-2$. Notice that $\mcP_\ell^+(x_{\ell +1}|x_n,x_\ell) $ is not conditioned on $x_0$ because the terms $Z_{\to\ell}(x_\ell|x_0)$ in the numerator and denominator in Eq.~\eqref{e:P(|x0xn)fromF} cancel out (see Eqs.~\eqref{e:P_BPx0xn} and \eqref{e:p_BPx0xn}). 

					Therefore, the equivalent of Eq.~\eqref{e:fromFtoMarkov} is 
					\begin{widetext}
					\be
					\begin{split}
					\mcP(\bx^\prime|x_n, x_0) = &\frac{F_0(x_1,x_0)F_{n-1}(x_n, x_{n-1})}{Z(x_n,x_0)}\prod_{\ell =1}^{n-2}F_\ell(x_{\ell +1},x_\ell)\\ 
					=& \frac{Z_{\to 1}(x_1|x_0)Z_{n-1\leftarrow}(x_{n-1}|x_n)}{Z(x_n,x_0)}\prod_{\ell = 1}^{n-2}\mcP_\ell^+(x_{\ell +1}|x_n, x_\ell) \frac{\prod_{\ell = 1}^{n-2}Z_{\ell\leftarrow}(x_\ell|x_n)}{\prod_{\ell = 1}^{n-2}Z_{\ell +1\leftarrow}(x_{\ell+1}|x_n)}\\
					 =& \frac{Z_{\to 1}(x_1|x_0)Z_{1\leftarrow}(x_1|x_n)}{Z(x_n,x_0)}\prod_{\ell = 1}^{n-2}\mcP_\ell^+(x_{\ell +1}|x_n, x_\ell) \frac{\prod_{\ell = 2}^{n-1}Z_{\ell\leftarrow}(x_\ell|x_n)}{\prod_{\ell = 1}^{n-2}Z_{\ell +1\leftarrow}(x_{\ell+1}|x_n)}\\ 
					 =&  p_1(x_1|x_n,x_0)\prod_{\ell=1}^{n-2} \mcP_\ell^+(x_{\ell +1}|x_n,x_\ell).\label{e:fromFtoMarkov_x0xn}
					 \end{split}
					 \ee
					 \end{widetext}
					 In the first line of Eq.~\eqref{e:fromFtoMarkov_x0xn} we have used Eq.~\eqref{e:mcP(|x0xn)} and we have taken out of the corresponding product the terms $F_0(x_1,x_0)$ and $F_{n-1}(x_n,x_{n-1})$, which are equal to $Z_{\to 1}(x_1|x_0)$ and $Z_{n-1\leftarrow}(x_{n-1}|x_n)$ according to Eqs.~\eqref{e:Z->1x0} and \eqref{e:Zn-1<-xn}. In the second line of Eq.~\eqref{e:fromFtoMarkov_x0xn} we have used Eq.~\eqref{e:P(|x0xn)fromF} to write $F_\ell(x_{\ell+1},x_\ell)$ in terms of $\mcP_\ell^+(x_{\ell +1}|x_n, x_\ell)$, $Z_{\ell\leftarrow}(x_\ell|x_n)$, and $Z_{\ell +1\leftarrow}(x_{\ell+1}|x_n)$. In the third line of Eq.~\eqref{e:fromFtoMarkov_x0xn} we have incorporated the term $Z_{n-1\leftarrow}(x_{n-1}|x_n)$ into the product $\prod_{\ell=1}^{n-2}Z_{\ell\leftarrow}(x_\ell|x_n)$ in te numerator, so the product now runs until $\ell = n-1$ instead of $\ell =n-2$. We have also taken out of this same product the term $Z_{1\leftarrow}(x_1|x_n)$, so the product now runs from $\ell=2$ rather than $\ell =1$. Finally, to obtain the fourth line of Eq.~\eqref{e:fromFtoMarkov_x0xn} we have used the fact that the products
					 \be
					 \prod_{\ell = 2}^{n-1}Z_{\ell\leftarrow}(x_\ell|x_n)\hspace{0.3cm}\textrm{and}\hspace{0.3cm}\prod_{\ell = 1}^{n-2}Z_{\ell +1\leftarrow}(x_{\ell+1}|x_n),
					 \ee
					 in the numerator and denominator, respectively, are equal, so they cancel out. Furthermore, we have used Eq.~\eqref{e:p_BPx0xn} to introduce the marginal $p_1(x_1|x_n,x_0)$.
					 Now, writing $\mcP_0^+(x_{1}|x_n,x_0) = p_1(x_1|x_n,x_0)$ in Eq.~\eqref{e:fromFtoMarkov_x0xn} and using Eq.~\eqref{e:mcP(|x0xn)} to write $\mcP(\bx)=p(x_0,x_n)\mcP(\bx^\prime|x_n,x_0)$ we obtain Eq.~\eqref{em:P_Bernstein} in the main text.

					 \

					 \section{Derivation of Eq.~\eqref{em:F=1+J} for continuous variables}\label{s:derivationF=1+J}
					 Here we derive Eq.~\eqref{em:F=1+J} for continuous variables in the particular case where the factors $F_\ell$ are given by Eq.~\eqref{em:F_Gaussian}. Introducing Eq.~\eqref{em:F_Gaussian} into Eq.~\eqref{em:test} and expanding $g(x^\prime)$ up to second order in $\xi = x^\prime - x$ we obtain
					 \begin{widetext}
					 \be
					 [F_\ell g](x) = \left[1 - \frac{\epsilon}{\hobs} V(x) + O(\epsilon^2)\right]\left[g(x)\int p_\sigma(\xi) \mathrm{d}\xi +\frac{\partial g(x)}{\partial x}\int\xi p_\sigma(\xi) \mathrm{d}\xi + \frac{1}{2}\frac{\partial^2 g(x)}{\partial x^2}\int\xi^2 p_\sigma(\xi) \mathrm{d}\xi + O(\epsilon^2)\right],
					 \ee
					 \end{widetext}
					 where
					 \be
					 p_\sigma(\xi)=\tfrac{1}{\sqrt{2\pi\sigma^2}}e^{-\xi^2/{2\sigma^2}},
					 \ee
					 is the Gaussian factor in Eq.~\eqref{em:F_Gaussian}---so, $\sigma^2=\epsilon\hobs/m$ which goes to zero as $\epsilon\to 0$. Here we have also expanded the term
					 \be
					 e^{-\frac{\epsilon}{2\Gamma} [V(x)+V(x^\prime)]} = 1-\frac{\epsilon}{\Gamma} V(x) + O(\epsilon^2),
					 \ee
					 up to first order in $\epsilon$---we have replaced $V(x^\prime)$ for $V(x)$ since the term $\epsilon[V(x) + V(x^\prime)]/2\Gamma$ is already of order $\epsilon$. Taking into account that $\int p_\sigma(\xi)\mathrm{d}\xi = 1$, that $\int\xi p_\sigma(\xi)\mathrm{d}\xi = 0$, and that $\int\xi^2 p_\sigma(\xi)\mathrm{d}\xi = \sigma^2$ we get
					 \be
					 \begin{split}
					 [F_\ell g](x) =& \left[1 - \frac{\epsilon}{\hobs} V(x) \right]\left[g(x) + \frac{\epsilon\Gamma}{2 m}\frac{\partial^2 g(x)}{\partial x^2} \right]+ O(\epsilon^2)\\
							=& g(x) - \frac{\epsilon}{\hobs} \left[V(x) g(x) - \frac{\Gamma^2}{2 m}\frac{\partial^2 g(x)}{\partial x^2}\right] + O(\epsilon^2)\\
							=& \left[\id - \frac{\epsilon}{\hobs} H\right]g(x) + O(\epsilon^2),
					 \end{split}
					 \ee
					 where $H$ is given by Eq.~\eqref{em:H}. This yields Eq.~\eqref{em:F=1+J} with $J_\ell = -H/\hobs$.

					 \

					 \section{Hamiltonians with positive off-diagonal entries}\label{s:two_atom}

					 Here we discuss the example of an infinite-dimensional quantum system described by factors with non-negative entries which, after truncation to its first two energy levels, turns into an effective system described by factors with negative off-diagonal entries. The latter is known as a two-level atom interacting with a coherent radiation field~\cite{haken2005physics} (see Sec.~15.3 therein).

					 Indeed, consider the Hamiltonian of an atom modeled as an electron, described by the momentum operator $i\hbar\nabla_{\mathbf{x}}$, moving in a potential field $V(\bx)$ produced by the nucleus,
					 \be\label{e:H0_atom}
					 H_0 = -\frac{\hbar^2}{2m}\nabla_{\mathbf{x}}^2 + V(\bx) = \sum_{n=0}^\infty E_n \left.| n\ket\bra n\right|.
					 \ee
					 In the second equality we have expanded the Hamiltonian in terms of its eigenvalues $E_n$ and its eigenvectors $\left| n\ket$, where $n$ is an integer, $n\geq 0$. 

					 Now consider a perturbation 
					 \be\label{e:U(t)}
					 U(\mathbf{x},t) = e \mathbf{x}\cdot\mathbf{E}(t) = e\mathbf{x}\cdot\mathbf{E}_0\cos(\omega t),
					 \ee
					 so the perturbed Hamiltonian becomes ${H = H_0 + U}$. Notice that the full Hamiltonian operator, $H$, can be derived via a path integral with Lagrangian 
					 \be\label{e:L2}
					 L = \frac{m}{2}\dot{\mathbf{x}}^2 - V(\bx) - U(\mathbf{x},t).
					 \ee
We can also derive $H$ via a real non-negative factors or kernels (cf. Eq.~\eqref{em:F})
	\be
	\mathcal{K}_\epsilon(\mathbf{x}^\prime, \mathbf{x}) = e^{-\epsilon\mathcal{H}(\mathbf{x}^\prime, \mathbf{x})/\hbar}\geq 0,\label{e:K2}
	\ee
	where
	\be
	\mathcal{H}(\mathbf{x}^\prime, \mathbf{x}) = \frac{m}{2}\left(\frac{\mathbf{x}^\prime-\mathbf{x}}{\epsilon}\right)^2 + V\left(\frac{|\mathbf{x}^\prime + \mathbf{x}|}{2}\right) + U(t).\label{e:HK2}
	\ee

	We will now see that, after a standard truncation of the full Hamiltonian $H=H_0+U$ (see Eqs.~\eqref{e:H0_atom} and \eqref{e:U(t)}) into an effective two-level system, we lose the equivalence with the positive kernel given by Eqs.~\eqref{e:K2} and \eqref{e:HK2}. Indeed, in the derivation of the Hamiltonian of a two-level atom it is usually assumed that the perturbation defined in Eq.~\eqref{e:U(t)} is near resonance with two relevant energy levels of the Hamiltonian $H_0$, say $E_0$ and $E_1$, i.e. {$|\omega - \omega_0|\lll\omega_0$}, where ${\hbar\omega_0 = E_1 - E_0}$. In this case, it is usually assumed that only the dynamics of these two energy levels matter. So, we can write
	\be\label{e:H_R}
	\begin{split}
	H =& E_0\left| 0 \ket\bra 0\right| + E_1\left| 1 \ket\bra 1\right| + U_{01}\left(\left| 0 \ket\bra 1\right| + \left| 1 \ket\bra 0\right|\right)\\
	   &+ H_\mathcal{R},
	\end{split}
	\ee
	where the first three terms in the right hand side of Eq.~\eqref{e:H_R} correspond to the transitions taking place within the subspace spanned by $\{\left|0\ket ,\left|1\ket\}$, and 
	\be\label{e:Residual}
	\begin{split}
	H_\mathcal{R} =& \sum_{n=2}^\infty E_n \left.| n\ket\bra n\right|\\
		&+ \sum_{m=0}^\infty\sum_{n> m, n\neq 0, 1}^\infty \left(U_{mn}\left| m \ket\bra n\right| + U_{nm}\left| n \ket\bra m\right|\right),
	\end{split}
	\ee
	collects all the remaining transitions. Here we have written 
	\be 
	U_{mn} = U_{nm}^\ast =  \bra m\right| U(t) \left| n \ket,
	\ee
	for $m, n$ integers, $m,n\geq 0$. For the sake of illustration, we are restricting here to the case where $U_{01}=U^\ast_{01}$ can be chosen to be real and $U_{00} = U_{11} = 0$~\cite{haken2005physics} (see Sec.~15.3 therein); this explains the form of Eq.~\eqref{e:H_R}.

	At this point it is argued that we can neglect $H_\mathcal{R}$ since the system is near resonance.  This yields the effective two-level Hamiltonian 
	\be\label{e:H2eff}
	H_{\rm eff} = \overline{E}\id^{(01)} -\frac{\hbar\omega_0}{2} \sigma_Z^{(01)} + D\cos(\omega t)\sigma_X^{01}, 
	\ee
	where $\overline{E}=(E_0+E_1)/2$, $D = \bra 0\right| \mathbf{r}\cdot\mathbf{E}_0\left| 1 \ket$, and
	\BE
	\id^{(01)} &=& \left| 0 \ket\bra 0\right| + \left| 1 \ket\bra 1\right|,\\
		\sigma_X^{(01)} &=& \left| 0 \ket\bra 1\right| + \left| 1 \ket\bra 0\right|,\\
		\sigma_Z^{(01)} &=& \left| 0 \ket\bra 0\right| - \left| 1 \ket\bra 1\right|.
		\EE

		If we now try to write this as a real factor 
		\be\label{e:F_eff_atom}
		F_{\rm eff} = \id + \epsilon J_{\rm eff},
		\ee
		with $J_{\rm eff} = -H_{\rm eff}/\hbar$, we end up with off-diagonal negative entries due to the factor $\cos(\omega t)$ accompanying $\sigma_X$ in Eq.~\eqref{e:H2eff}. So, the full Hamiltonian $H=H_0 + U$ in Eq.~\eqref{e:H_R} can be represented in terms of the real positive kernel given by  Eqs.~\eqref{e:K2} and \eqref{e:HK2}, but the truncated effective Hamiltonian $H_{\rm eff}$ in Eq.~\eqref{e:H2eff} cannot. What happened? The full factor $F = \id-\epsilon H/\hbar$ associated to kernel $\mathcal{K}_\epsilon $ (see Eqs.~\eqref{e:K2} and \eqref{e:H_R}), which has only non-negative entries, can be written as 
		\be\label{e:F_R_atom}
		F = \mathcal{R} + F_{\rm eff} . 
		\ee
		So, even though the effective factor $F_{\rm eff}$ in Eq.~\eqref{e:F_eff_atom} can have negative off-diagonal entries, those would be ``corrected'' by the ``reference'' term $\mathcal{R}= -\epsilon H_\mathcal{R}/\hbar$ yielding only positive quantities with a clear probabilistic interpretation. 

		\section{Measuring momentum}\label{a:momentum}
		Here we describe how to measure the momentum of a system particle whose position is described by $x$, using a probe particle whose position is described by $X$. The initial states of the system particle and the probe particle are
		\BE
		\psi_{\rm sys}(x) &= \int c_k e^{i kx}\mathrm{d}k,\\
			\psi_{\rm dev}(X) &=\left[\frac{e^{-X^2/2\sigma^2}}{\sqrt{2\pi\sigma^2}}\right]^{\frac{1}{2}},
		\EE
		respectively. For convenience, here we write the initial state of the system particle in terms of its Fourier transform, $c_k$, and the initial state of the probe particle as a Gaussian packet centred around zero and with a very small standard deviation, i.e. $0<\sigma\ll 1$. Furthermore, the initial state of the total system composed of system and probe particles is $\Psi_0(x,X)=\psi_{\rm sys}(x)\psi_{\rm dev}(X)$. That is, the system and probe particles are initially independent of each other.

		The interaction between the system particle and the probe particle is given by the Hamiltonian
		\be
		H_{\rm int} = -g(t)\hbar^2\frac{\partial^2}{\partial x\partial X},
		\ee
		where $g(t)$ describes the strength of the interaction at time $t$. In the time interval $[0,T]$, the function $g(t)$ is equal to $g_0/T$, where $g_0$ and $T$ are constants,  and zero elsewhere. For convenience, this Hamiltonian can be written as $H_{\rm int}= g(t) p P$, in terms of the operators $p=-i\hbar\partial/\partial x$ and $P=-i\hbar\partial/\partial X$. For simplicity, the system and probe Hamiltonians are neglected, as usually done in the field of quantum measurement. So, $H_{\rm int}$ is the total Hamiltonian from here on.

		At time $T$ the initial state of the composed system is
		\be
		\begin{split}
		\Psi_T(x,X) &= e^{-i\int_0^T g(t) p P\mathrm{d}t/\hbar}\Psi_0(X, x), \\
			       &= \int c_k e^{ixk} e^{-i g_0 k P} \psi_{\rm dev}(X)\mathrm{d}k,\\
			       &= \int c_k e^{ixk} e^{- g_0 \hbar k \partial/\partial X} \psi_{\rm dev}(X)\mathrm{d}k,\\
			       &= \int c_k e^{ixk} \psi_{\rm dev}(X- g_0 \hbar k)\mathrm{d}k.
			       \end{split}
			       \ee
			       Here to go from the first to the second line we used the fact that $e^{ikx}$ is an eigenstate of the operator $\partial/\partial x$ with eigenvalue $ik$. We have also done the integral $\int_0^T g(t)\mathrm{d} t = g_0/T$. Simply using $P=-i\hbar\partial/\partial X$ takes us from the second to the third line. To go from the third to the fourth line we have used the fact that $e^{b \partial/\partial X}f(X) = f(X+b)$, where $f$ is a generic smooth function---this can be seen by doing a Taylor expansion of the operator $e^{b\partial/\partial X}$. 

			       Now, the joint probability that the system and probe particles, respectively, are at position $x$ and $X$ is given by $\mcP(x,X)=|\Psi_T(x,X)|^2$. So, 
			       \begin{widetext}
			       \be
			       \mcP(x,X) = \int c_k c_{k^\prime}^\ast e^{ix(k-k^\prime)} \psi_{\rm dev}(X- g_0 \hbar k)\psi_{\rm dev}(X- g_0 \hbar k^\prime)\mathrm{d}k\mathrm{d}k^\prime.
			       \ee
			       \end{widetext}

			       Marginalizing $\mcP(x,X)$ over $x$, as we are only interested in {\it inferring} the momentum of the system particle from measuring position $X$, we have
			       \begin{widetext}
			       \be
			       \begin{split}
			       \mcP_{\rm dev}(X) &= \int c_k c_{k^\prime}^\ast \left[\int e^{ix(k-k^\prime)}\mathrm{d}x\right] \psi_{\rm dev}(X- g_0 \hbar k)\psi_{\rm dev}(X- g_0 \hbar k^\prime)\mathrm{d}k\mathrm{d}k^\prime,\\
				       &= \int c_k c_{k^\prime}^\ast \delta(k-k^\prime) \psi_{\rm dev}(X- g_0 \hbar k)\psi_{\rm dev}(X- g_0 \hbar k^\prime)\mathrm{d}k\mathrm{d}k^\prime,\\
				       &= \int |c_k|^2 \frac{e^{(X- g_0 \hbar k)^2/2\sigma^2}}{\sqrt{2\pi\sigma^2}}\mathrm{d}k,\\
				       &\approx \int |c_k|^2 \delta(X- g_0 \hbar k)\mathrm{d}k.
				       \end{split}
				       \ee
				       \end{widetext}
				       Here, to go from the first to the second line, we used the fact that $\int e^{ix(k-k^\prime)}\mathrm{d}x = 2\pi\delta(k-k^\prime)$ is a representation of the Dirac delta function---the factor $2\pi$ has been absorbed in the terms $c_k$ and $c_k^\ast$ for simplicity. To go from the second to the third line, we used the fact that $\int\delta(k-k^\prime)f(k,k^\prime)\mathrm{d}k^\prime = f(k,k)$, where $f$ is a generic smooth function. We have also used the fact that $[\psi_{\rm dev}(X-g_0\hbar k)]^2$ yields a properly normalized Gaussian distribution. Finally, using the fact that a Gaussian distribution tends to a Dirac delta function when its standard deviation tends to zero takes us from the third line to the fourth line.


\begin{thebibliography}{10}

\bibitem{varela2017embodied}
F.~J. Varela, E.~Thompson, and E.~Rosch, {\em The embodied mind: Cognitive
  science and human experience (revised edition)}.
\newblock MIT press, 2017.

\bibitem{thompson2010mind}
E.~Thompson, {\em Mind in life}.
\newblock Harvard University Press, 2010.

\bibitem{di2017sensorimotor}
E.~Di~Paolo, T.~Buhrmann, and X.~Barandiaran, {\em Sensorimotor life: An
  enactive proposal}.
\newblock Oxford University Press, 2017.

\bibitem{shapiro2019embodied}
L.~Shapiro, {\em Embodied cognition---2nd edition}.
\newblock Routledge, 2019.

\bibitem{djebbara2019sensorimotor}
Z.~Djebbara, L.~B. Fich, L.~Petrini, and K.~Gramann, ``Sensorimotor brain
  dynamics reflect architectural affordances,'' {\em Proceedings of the
  National Academy of Sciences}, vol.~116, no.~29, pp.~14769--14778, 2019.

\bibitem{wilson2002six}
M.~Wilson, ``Six views of embodied cognition,'' {\em Psychonomic Bulletin \&
  Review}, vol.~9, no.~4, pp.~625--636, 2002.

\bibitem{bridgeman2011embodied}
B.~Bridgeman and P.~Tseng, ``Embodied cognition and the perception--action
  link,'' {\em Physics of Life Reviews}, vol.~8, no.~1, pp.~73--85, 2011.

\bibitem{friston2013life}
K.~Friston, ``Life as we know it,'' {\em Journal of the Royal Society
  Interface}, vol.~10, no.~86, p.~20130475, 2013.

\bibitem{velmans2009understanding}
M.~Velmans, {\em Understanding consciousness}.
\newblock Routledge, 2009.

\bibitem{thompson2014waking}
E.~Thompson, {\em Waking, dreaming, being: Self and consciousness in
  neuroscience, meditation, and philosophy}.
\newblock Columbia University Press, 2014.

\bibitem{bitbol2008consciousness}
M.~Bitbol, ``Is consciousness primary?,'' {\em NeuroQuantology}, vol.~6,
  pp.~53--71, 2008.

\bibitem{realpe2}
J.~Realpe, ``Observers, relational quantum mechanics, and {B}uddhist
  philosophy,'' {\em Mind and Matter}, vol.~22, no.~1, pp.~95--126, 2024.

\bibitem{friston2010free}
K.~Friston, ``The free-energy principle: a unified brain theory?,'' {\em Nature
  Reviews Neuroscience}, vol.~11, no.~2, pp.~127--138, 2010.

\bibitem{Rovelli-1996}
C.~Rovelli, ``Relational quantum mechanics,'' {\em International Journal of
  Theoretical Physics}, vol.~35, p.~1637, 1996.

\bibitem{schwobel2018active}
S.~Schw{\"o}bel, S.~Kiebel, and D.~Markovi{\'c}, ``Active inference, belief
  propagation, and the bethe approximation,'' {\em Neural Computation},
  vol.~30, no.~9, pp.~2530--2567, 2018.

\bibitem{presse2013principles}
S.~Press{\'e}, K.~Ghosh, J.~Lee, and K.~A. Dill, ``Principles of maximum
  entropy and maximum caliber in statistical physics,'' {\em Reviews of Modern
  Physics}, vol.~85, no.~3, p.~1115, 2013.

\bibitem{pearl2009causality}
J.~Pearl, {\em Causality}.
\newblock Cambridge university press, 2009.

\bibitem{Zambrini-1987}
J.~C. Zambrini, ``Euclidean quantum mechanics,'' {\em Physical Review A},
  vol.~35, pp.~3631--3649, May 1987.

\bibitem{Mezard-book-2009}
M.~Mezard and A.~Montanari, {\em Information, Physics, and Computation}.
\newblock Oxford Graduate Texts, Oxford University Press, USA, 2009.

\bibitem{debrota2018faqbism}
J.~B. DeBrota and B.~C. Stacey, ``{FAQBism},'' {\em arXiv preprint
  arXiv:1810.13401}, 2018.

\bibitem{mermin2018making}
N.~D. Mermin, ``Making better sense of quantum mechanics,'' {\em Reports on
  Progress in Physics}, vol.~82, no.~1, p.~012002, 2018.

\bibitem{fuchs2014introduction}
C.~A. Fuchs, N.~D. Mermin, and R.~Schack, ``An introduction to qbism with an
  application to the locality of quantum mechanics,'' {\em American Journal of
  Physics}, vol.~82, no.~8, pp.~749--754, 2014.

\bibitem{crutchfield1984space}
J.~P. Crutchfield, ``Space-time dynamics in video feedback,'' {\em Physica D:
  Nonlinear Phenomena}, vol.~10, no.~1-2, pp.~229--245, 1984.

\bibitem{rovelli2007quantum}
C.~Rovelli, {\em Quantum gravity}.
\newblock Cambridge University Press, 2007.

\bibitem{deacon2011incomplete}
T.~W. Deacon, {\em Incomplete nature: How mind emerged from matter}.
\newblock WW Norton \& Company, 2011.

\bibitem{dehaene2017consciousness}
S.~Dehaene, H.~Lau, and S.~Kouider, ``What is consciousness, and could machines
  have it?,'' {\em Science}, vol.~358, no.~6362, pp.~486--492, 2017.

\bibitem{haken2005physics}
H.~Haken and H.~C. Wolf, {\em The physics of atoms and quanta: introduction to
  experiments and theory}, vol.~1439.
\newblock Springer Science \& Business Media, 2005.

\bibitem{vinci2017non}
W.~Vinci and D.~A. Lidar, ``Non-stoquastic hamiltonians in quantum annealing
  via geometric phases,'' {\em npj Quantum Information}, vol.~3, no.~1, p.~38,
  2017.

\bibitem{padmanabhan2011nonrelativistic}
H.~Padmanabhan and T.~Padmanabhan, ``Nonrelativistic limit of quantum field
  theory in inertial and noninertial frames and the principle of equivalence,''
  {\em Physical Review D}, vol.~84, no.~8, p.~085018, 2011.

\bibitem{feynman2010quantum}
R.~Feynman, A.~Hibbs, and D.~Styer, {\em Quantum mechanics and path
  integrals---Emended edition}.
\newblock Dover Books on Physics, Dover Publications, 2010.

\bibitem{svensson2013pedagogical}
B.~E. Svensson, ``Pedagogical review of quantum measurement theory with an
  emphasis on weak measurements,'' {\em Quanta}, vol.~2, no.~1, pp.~18--49,
  2013.

\bibitem{aharonov2023conservation}
Y.~Aharonov, S.~Popescu, and D.~Rohrlich, ``Conservation laws and the
  foundations of quantum mechanics,'' {\em Proceedings of the National Academy
  of Sciences}, vol.~120, no.~41, p.~e2220810120, 2023.

\bibitem{tang2015neuroscience}
Y.-Y. Tang, B.~K. H{\"o}lzel, and M.~I. Posner, ``The neuroscience of
  mindfulness meditation,'' {\em Nature Reviews Neuroscience}, vol.~16, no.~4,
  pp.~213--225, 2015.

\bibitem{deguchi2021can}
Y.~Deguchi, J.~L. Garfield, G.~Priest, and R.~H. Sharf, {\em What Can't be
  Said: Paradox and Contradiction in East Asian Thought}.
\newblock Oxford University Press, 2021.

\bibitem{merleau1962phenomenology}
M.~Merleau-Ponty, ``Phenomenology of perception,'' {\em Routledge}, vol.~5,
  1962.

\bibitem{brukner2020facts}
{\v{C}}.~Brukner, ``Facts are relative,'' {\em Nature Physics}, vol.~16,
  no.~12, pp.~1172--1174, 2020.

\bibitem{mermin2014physics}
N.~D. Mermin, ``Physics: Qbism puts the scientist back into science,'' {\em
  Nature}, vol.~507, no.~7493, pp.~421--423, 2014.

\bibitem{fuchs2013quantum}
C.~A. Fuchs and R.~Schack, ``Quantum-{B}ayesian coherence,'' {\em Reviews of
  Modern Physics}, vol.~85, no.~4, p.~1693, 2013.

\bibitem{pienaar2021qbism}
J.~Pienaar, ``Qbism and relational quantum mechanics compared,'' {\em
  Foundations of Physics}, vol.~51, no.~5, pp.~1--18, 2021.

\bibitem{brukner2017quantum}
{\v{C}}.~Brukner, ``On the quantum measurement problem,'' in {\em Quantum [Un]
  Speakables II}, pp.~95--117, Springer, 2017.

\bibitem{rovelli2021helgoland}
C.~Rovelli, {\em Helgoland: Making sense of the quantum revolution}.
\newblock Penguin, 2021.

\bibitem{hofstadter2013strange}
D.~R. Hofstadter, {\em I am a strange loop}.
\newblock Basic Books, 2013.

\bibitem{garfield1995fundamental}
J.~L. Garfield {\em et~al.}, {\em The fundamental wisdom of the middle way:
  Nagarjuna's Mulamadhyamakakarika}.
\newblock Oxford University Press, 1995.

\bibitem{westerhoff2009nagarjuna}
J.~Westerhoff, {\em Nagarjuna's Madhyamaka: A philosophical introduction}.
\newblock Oxford University Press, 2009.

\bibitem{westerhoff2024candrakirti}
J.~Westerhoff, {\em Candrak{\=\i}rti's introduction to the middle way: A
  guide}.
\newblock Oxford University Press, 2024.

\bibitem{garfield2014engaging}
J.~L. Garfield, {\em Engaging Buddhism: Why it matters to philosophy}.
\newblock Oxford University Press, 2014.

\bibitem{bitbol2019two}
M.~Bitbol, ``Two aspects of {\'s}{\=u}nyat{\=a} in quantum physics: relativity
  of properties and quantum non-separability,'' in {\em Quantum Reality and
  Theory of {\'S}{\=u}nya}, pp.~93--117, Springer, 2019.

\bibitem{friston2017graphical}
K.~J. Friston, T.~Parr, and B.~de~Vries, ``The graphical brain: belief
  propagation and active inference,'' {\em Network Neuroscience}, vol.~1,
  no.~4, pp.~381--414, 2017.

\bibitem{buhrmann2013dynamical}
T.~Buhrmann, E.~A. Di~Paolo, and X.~Barandiaran, ``A dynamical systems account
  of sensorimotor contingencies,'' {\em Frontiers in Psychology}, vol.~4,
  p.~285, 2013.

\bibitem{gallagher2017enactivist}
S.~Gallagher, {\em Enactivist interventions: Rethinking the mind}.
\newblock Oxford University Press, 2017.

\bibitem{jaynes2003probability}
E.~T. Jaynes, {\em Probability theory: The logic of science}.
\newblock Cambridge University Press, 2003.

\end{thebibliography}

\end{document}